\def\llncs{0}
\def\fullpage{1}
\def\anonymous{0}
\def\authnote{1}
\def\notxfont{0}
\def\submission{0}

\ifnum\submission=1
\def\llncs{1}
\fi

\ifnum\llncs=1 	\documentclass[envcountsect,a4paper,runningheads,10pt]{llncs}
\else
	\documentclass[letterpaper,hmargin=1.0in,vmargin=1.0in,11pt]{article}
			\ifnum\fullpage=1
		\usepackage{fullpage}
		\fi
\fi

\usepackage[%
  colorlinks=true,
  citecolor=blue,
  pagebackref=true
]{hyperref}

\usepackage{amsmath, amsfonts, amssymb, mathtools,amscd}

\usepackage{amsthm}

\usepackage{lmodern}
\usepackage[T1]{fontenc}
\usepackage[utf8]{inputenc}

\usepackage{arydshln} 
\usepackage{url}
\usepackage{ifthen}
\usepackage{bm}
\usepackage{multirow}
\usepackage[dvips]{graphicx}
\usepackage[usenames]{color}
\usepackage{xcolor,colortbl} 
\usepackage{threeparttable}
\usepackage{comment}
\usepackage{paralist,verbatim}
\usepackage{cases}
\usepackage{booktabs}
\usepackage{braket}
\usepackage{cancel} 
\usepackage{ascmac} 
\usepackage{framed}
\usepackage{authblk}
\usepackage{pifont}
\usepackage{qcircuit}
\usepackage{tikz}
\definecolor{darkblue}{rgb}{0,0,0.6}
\definecolor{darkgreen}{rgb}{0,0.5,0}
\definecolor{maroon}{rgb}{0.5,0.1,0.1}
\definecolor{dpurple}{rgb}{0.2,0,0.65}

\usepackage[capitalise,noabbrev]{cleveref}
\usepackage[absolute]{textpos}
\usepackage[final]{microtype}
\usepackage[absolute]{textpos}
\usepackage{everypage}
\DeclareMathAlphabet{\mathpzc}{OT1}{pzc}{m}{it}

\usepackage{algorithmic}
\usepackage{algorithm}
\usepackage{here}

\newtheoremstyle{thicktheorem}%
{\topsep}
{\topsep}
{\itshape}{}%
{\bfseries}%
{.}
{ }%
{\thmname{#1}\thmnumber{ #2}%
		\thmnote{ (#3)}%
}

\newtheoremstyle{remark}
{\topsep}
{\topsep}
	{}
	{}
	{}
	{.}
	{ }
	{\textit{\thmname{#1}}\thmnumber{ #2}
			\thmnote{ (#3)}%
	}

\ifnum\llncs=0
	\theoremstyle{thicktheorem}
	\newtheorem{theorem}{Theorem}[section]
	\newtheorem{lemma}[theorem]{Lemma}

	\newtheorem{definition}[theorem]{Definition}

	\theoremstyle{remark}
	\newtheorem{claim}[theorem]{Claim}
	\newtheorem{remark}[theorem]{Remark}

\else
\fi

\Crefname{MyClaim}{Claim}{Claims}

	\crefname{theorem}{Theorem}{Theorems}
	\crefname{assumption}{Assumption}{Assumptions}
	\crefname{construction}{Construction}{Constructions}
	\crefname{corollary}{Corollary}{Corollaries}
	\crefname{conjecture}{Conjecture}{Conjectures}
	\crefname{definition}{Definition}{Definitions}
	\crefname{example}{Example}{Examples}
	\crefname{experiment}{Experiment}{Experiments}
	\crefname{counterexample}{Counterexample}{Counterexamples}
	\crefname{lemma}{Lemma}{Lemmata}
	\crefname{observation}{Observation}{Observations}
	\crefname{proposition}{Proposition}{Propositions}
	\crefname{remark}{Remark}{Remarks}
	\crefname{claim}{Claim}{Claims}
	\crefname{fact}{Fact}{Facts}
	\crefname{note}{Note}{Notes}

\ifnum\llncs=1
 \crefname{appendix}{App.}{Appendices}
 \crefname{section}{Sec.}{Sections}
\else
\fi

\ifnum\llncs=1
\renewcommand*{\backref}[1]{}
\else
	\renewcommand*{\backref}[1]{(Cited on page~#1.)}
	\ifnum\notxfont=1
	\else
		\usepackage{newtxtext}
	\fi
\fi
\ifnum\authnote=0  
\newcommand{\fuyuki}[1]{}
\newcommand{\minki}[1]{}
\newcommand{\takashi}[1]{}
\newcommand{\ryo}[1]{}

\else
\newcommand{\fuyuki}[1]{$\ll$\textsf{\color{red} Fuyuki: { #1}}$\gg$}
\newcommand{\takashi}[1]{$\ll$\textsf{\color{orange} Takashi: { #1}}$\gg$}
\newcommand{\shira}[1]{$\ll$\textsf{\color{blue} Yuki: { #1}}$\gg$}
\newcommand{\ryo}[1]{$\ll$\textsf{\color{green} Ryo: { #1}}$\gg$}
\fi

\newcommand{\TDF}{\mathsf{TDF}}
\newcommand{\PTDF}{\mathsf{PTDF}}

\newcommand{\UPO}{\mathsf{UPO}}
\newcommand{\PPRF}{\mathsf{PPRF}}
\newcommand{\SDE}{\mathsf{SDE}}
\newcommand{\OWF}{\mathsf{OWF}}
\newcommand{\qKG}{\mathpzc{KG}}
\newcommand{\qInv}{\mathpzc{Inv}}
\newcommand{\qDel}{\mathpzc{Del}}
\newcommand{\qtd}{\mathpzc{td}}
\newcommand{\qDec}{\mathpzc{Dec}}
\newcommand{\qdk}{\mathpzc{dk}}
\newcommand{\HPRG}{\mathsf{HPRG}}
\newcommand{\PRG}{\mathsf{PRG}}
\newcommand{\qSamp}{\mathpzc{Samp}}
\newcommand{\qGen}{\mathpzc{Gen}}
\newcommand{\qObf}{\mathpzc{Obf}}
\newcommand{\qEval}{\mathpzc{Eval}}
\newcommand{\qA}{\mathpzc{A}}
\newcommand{\qB}{\mathpzc{B}}
\newcommand{\qC}{\mathpzc{C}}
\newcommand{\qD}{\mathpzc{D}}
\newcommand{\qE}{\mathpzc{E}}
\newcommand{\qaux}{\mathpzc{aux}}
\newcommand{\Puncture}{\mathsf{Puncture}}
\newcommand{\PuncInv}{\mathsf{PuncInv}}

\newcommand{\coin}{\mathsf{coin}}

\newcommand{\Tr}{\mathrm{Tr}}

\newcommand{\qsk}{\mathpzc{sk}}

\newcommand{\cert}{\keys{cert}}

\newcommand{\iO}{i\cO}

\newcommand{\Samp}{\algo{Samp}}

\newcommand{\cA}{\mathcal{A}}
\newcommand{\cB}{\mathcal{B}}
\newcommand{\cC}{\mathcal{C}}
\newcommand{\cD}{\mathcal{D}}

\newcommand{\cM}{\mathcal{M}}

\newcommand{\cO}{\mathcal{O}}

\newcommand{\cU}{\mathcal{U}}

\def\makeuppercase#1{
\expandafter\newcommand\csname tl#1\endcsname{\widetilde{#1}}
}

\def\makelowercase#1{
\expandafter\newcommand\csname tl#1\endcsname{\widetilde{#1}}
}

\newcommand{\N}{\mathbb{N}}

\newcommand{\regR}{\mathbf{R}}

\newcommand{\secp}{\lambda}

\newcommand{\aux}{\mathsf{aux}}

\newcommand*{\sk}{\keys{sk}}
\newcommand*{\pk}{\keys{pk}}

\newcommand{\ct}{\keys{ct}}

\newcommand*{\pp}{\keys{pp}}
\newcommand*{\dk}{\keys{dk}}
\newcommand*{\vk}{\keys{vk}}
\newcommand*{\ek}{\keys{ek}}

\newcommand*{\td}{\keys{td}}

\newcommand*{\keys}[1]{\mathsf{#1}}

\newcommand*{\algo}[1]{\ensuremath{\mathsf{#1}}}

\newenvironment{boxfig}[2]{\begin{figure}[#1]\fbox{\begin{minipage}{0.97\linewidth}
                        \vspace{0.2em}
                        \makebox[0.025\linewidth]{}
                        \begin{minipage}{0.95\linewidth}
            {{
                        #2 }}
                        \end{minipage}
                        \vspace{0.2em}
                        \end{minipage}}}{\end{figure}}

\newcommand{\bit}{\{0,1\}}

\newcommand{\Setup}{\algo{Setup}}

\newcommand{\Gen}{\algo{Gen}}

\newcommand{\KG}{\algo{KG}}
\newcommand{\Enc}{\algo{Enc}}
\newcommand{\Dec}{\algo{Dec}}

\newcommand{\Vrfy}{\algo{Vrfy}}

\newcommand{\Inv}{\algo{Inv}}

\newcommand{\TD}{\algo{TD}}

\newcommand{\Eval}{\algo{Eval}}

\newcommand{\PKE}{\algo{PKE}}

\newcommand{\negl}{{\mathsf{negl}}}

\newcommand{\poly}{{\mathrm{poly}}}

\makeatletter
\DeclareRobustCommand
  \myvdots{\vbox{\baselineskip4\p@ \lineskiplimit\z@
    \hbox{.}\hbox{.}\hbox{.}}}
\makeatother
 
\title{Trapdoor Functions with Secure Key Leasing and Copy Protection}

\ifnum\anonymous=1
\ifnum\llncs=1
\author{\empty}\institute{\empty}
\else
\author{}
\fi
\else
\ifnum\llncs=1
\author{
Fuyuki Kitagawa\inst{1,2} \and 
Ryo Nishimaki\inst{1,2} \and 
Yuki Shirakawa\inst{3} \and 
Takashi Yamakawa\inst{1,2,3}
}
\institute{
NTT Social Informatics Laboratories, Tokyo, Japan \and 
NTT Research Center for Theoretical Quantum Information, Atsugi, Japan \and 
Yukawa Institute for Theoretical Physics, Kyoto University, Kyoto, Japan
}
\else
\author[1,2]{Fuyuki Kitagawa}
\author[1,2]{Ryo Nishimaki}
\author[3]{ Yuki Shirakawa}
\author[1,2,3]{ Takashi Yamakawa}
\affil[1]{{\small NTT Social Informatics Laboratories, Tokyo, Japan}\authorcr{\small \{fuyuki.kitagawa,ryo.nishimaki,takashi.yamakawa\}@ntt.com}}
\affil[2]{{\small NTT Research Center for Theoretical Quantum Information, Atsugi, Japan}}
\affil[3]{{\small Yukawa Institute for Theoretical Physics, Kyoto University, Kyoto, Japan}\authorcr{\small yuki.shirakawa@yukawa.kyoto-u.ac.jp}}
\fi 
\fi

\date{}

\begin{document}

\maketitle

\begin{abstract}

Inspired by the no-cloning theorem in quantum theory, a variety of quantum cryptographic primitives with unclonable functionalities, such as secure key leasing and copy protection, have been proposed and attracted significant attention. 
However, trapdoor functions (TDFs), fundamental primitives in public-key cryptography, have not been extensively studied in these areas. 
In this work, we initiate a study of TDFs in both secure key leasing and copy protection settings. 
We first introduce the definition of TDFs with secure key leasing (TDF-SKL), which enables leasing and deleting of quantum trapdoors. 
We formalize TDF-SKL both with and without domain sampler, and give a construction of TDF-SKL without domain sampler based on the LWE assumption and a construction of TDF-SKL with domain sampler based on any standard PKE schemes combined with hinting pseudorandom generators [Koppula and Waters, CRYPTO 2019].
Next, we define TDFs with copy protection (TDF-CP), where the inversion functionality is copy protected by a quantum trapdoor. We establish a construction of TDF-CP assuming indistinguishability obfuscation and the LWE assumption, following a modular framework of copy protection proposed by Ananth and Behera [CRYPTO 2024].
We also present applications of TDF-SKL and TDF-CP.
Existing constructions of public-key encryption with secure key leasing (PKE-SKL) and single-decryptor encryption (SDE) suffer from a critical vulnerability: quantum decryption keys may be destroyed after decrypting maliciously chosen ciphertexts.
We construct PKE-SKL schemes and SDE schemes with robust quantum decryption keys that remain reusable after decrypting arbitrary ciphertexts from TDF-SKL and TDF-CP, respectively.


\end{abstract}

\newpage

 \setcounter{tocdepth}{2}
 \tableofcontents

 \newpage

\section{Introduction}
The no-cloning theorem \cite{WZ82,Die82} in quantum theory states that unknown quantum states cannot be copied.
Inspired by this theorem, a large number of novel quantum cryptographic primitives that have unclonable functionalities have been proposed over the past few decades.
Since these primitives are impossible in classical cryptography, they have attracted significant attention. 
Major examples include quantum money \cite{Wiesner83}, unclonable encryption \cite{TQC:BroLor20}, secure software and key leasing \cite{EC:AnaLaP21,AC:KitNis22,EC:AKNYY23,TCC:AnaPorVai23}, copy protection \cite{Aar09}, single-decryptor encryption \cite{GZ20,C:CLLZ21}, and more.

\paragraph{Secure Key Leasing.}
Secure key leasing \cite{AC:KitNis22,EC:AKNYY23,TCC:AnaPorVai23} enables the delegation and revocation of tasks that require cryptographic secret keys.
It is achieved by adding two additional algorithms to a base cryptographic primitive (e.g. public-key encryption), which are a deletion algorithm that deletes a secret key and issues its certificate, and a verification algorithm that verifies the validity of the certificate.
Here, the secret key is a quantum state and verification is performed using a verification key generated together with the quantum secret key.
Kitagawa, Morimae, and Yamakawa \cite{EC:KitMorYam25} recently proposed a simple framework for PKE-SKL, which enables the construction of PKE-SKL from standard PKE schemes.
They also introduced a new security notion called IND-VRA, which requires primitives to be secure even if the adversary obtains the verification key together with the challenge ciphertext, and showed that their PKE-SKL scheme satisfies this strong security.

\paragraph{Copy Protection and Single-Decryptor Encryption.}
Copy protection is a cryptographic scheme introduced by Aaronson \cite{Aar09}, which prevents generating multiple copies of a functionality to compute a given function $f$ by encoding $f$ into a quantum program state $\rho_f$.
Informally, a copy protection scheme requires that, given a single copy of the quantum program state $\rho_f$, one can efficiently compute $f$, whereas it is infeasible to efficiently generate two (or more) quantum programs that are each capable of computing $f$.
\cite{GZ20,C:CLLZ21} introduced and constructed single decryptor encryption (SDE), which is a public-key encryption scheme with copy protected decryption functionalities. 
Subsequently, Ananth and Behera \cite{C:AnaBeh24} proposed a modular and simple framework for constructing copy protection based on the primitives they call unclonable puncturable obfuscation. Kitagawa and Yamakawa \cite{TCC:KitYam25} recently introduced a new security notion called CPA$^+$ anti-piracy security, which is stronger than the conventional security notions of SDE such as CPA-style anti-piracy security.
They also established a construction of CPA$^+$ anti-piracy secure SDE assuming only polynomially secure iO and OWFs.

\paragraph{Motivations.}
A variety of results have been established in secure key leasing and copy protection settings.
Nevertheless, in both areas, there is a lack of work on fundamental primitives, namely \emph{trapdoor functions}.
In classical cryptography, it has been well known that TDFs can be used to construct IND-CPA secure PKE via the Goldreich-Levin theorem \cite{STOC:GolLev89}, and more recently, Hohenberger, Koppula, and Waters \cite{C:HohKopWat20} showed that IND-CCA secure PKE can also be constructed from TDFs.
Moreover, Garg, Hajiabadi, Malavolta, and Ostrovsky \cite{AC:GHMO21} established that TDFs can be constructed by combining any IND-CPA secure PKE with pseudorandom ciphertexts and hinting pseudorandom generators (PRGs), where hinting PRGs are PRGs with a stronger pseudorandomness requirement introduced by Koppula and Waters \cite{C:KopWat19}. 
While TDFs play a crucial role in classical public-key cryptography, they have received little attention so far in the field of unclonable cryptography.
This naturally leads to the following question:
\begin{center}
    \emph{
    Can we construct TDFs with secure key leasing and TDFs with copy protection from standard assumptions?
    Moreover, what are their applications?
    }
\end{center}

Furthermore, although various constructions of PKE-SKL and SDE are known, they suffer from a certain vulnerability, namely the fragility of quantum decryption keys.
Let us consider the following scenario:
Alice obtains a quantum decryption key. 
A malicious Bob then sends an arbitrary ciphertext to Alice and asks Alice to decrypt it, where the ciphertext is not necessarily generated by the encryption algorithm.
If Alice decrypts Bob's ciphertext, the quantum decryption key may collapse and become unusable thereafter.
It is trivial that such an issue does not arise when the secret keys are classical, and we need to address this problem only when the secret keys are quantum states, as in PKE-SKL and SDE.
Unfortunately, existing constructions of PKE-SKL and SDE are not secure against such key-destruction attacks.
This current situation naturally leads to the following question:
\begin{center}
    \emph{
    Can we construct PKE-SKL and SDE that are secure against key-destruction attacks from standard assumptions?
    }
\end{center}

\subsection{Our Results}
In this work, we initiate the study of TDFs in both secure key leasing and copy protection settings, and resolve the above questions affirmatively.
More precisely, our results are summarized as follows:

\paragraph{Trapdoor Functions with Secure Key Leasing.}
We first introduce the definition of trapdoor functions with secure key leasing (TDF-SKL).
Informally speaking, a TDF-SKL is a trapdoor function that allows leasing and deleting a quantum trapdoor.
In classical cryptography, a trapdoor function is defined as a tuple of a key generation algorithm, an evaluation algorithm, and an inversion algorithm.
Typically, one-wayness is required to hold for inputs sampled uniformly at random.
Alternatively, some definitions include an additional algorithm called a domain sampler, in which case one-wayness is required to hold for inputs generated by this sampler.\footnote{One might suspect that trapdoor functions with domain samplers are so weak that they become equivalent to PKE, as observed in \cite{C:HohKopWat20}. However, this is not the case in our setting, where we require correctness over the entire domain, rather than merely on inputs sampled by the domain sampler. See  \Cref{sec:TDF_sampler} for further details.} 
We define and construct TDF-SKL in both settings.
We construct TDF-SKL without domain sampler (or simply TDF-SKL) based on the LWE assumption, and TDF-SKL with domain sampler from any standard PKE together with hinting PRGs.
Notably, the latter construction admits broader instantiations, as its building blocks can also be realized 
under group-action-based assumptions~\cite{Couveignes06,cryptoeprint:2006/145,AC:AlaPat22}. Furthermore, we remark that this second construction suffices for our application to robust PKE-SKL discussed below. 


\paragraph{Trapdoor Functions with Copy Protection.}
Next, we introduce the definition of trapdoor functions with copy protection (TDF-CP).
Here, a TDF-CP is a trapdoor function whose trapdoor is a quantum state and the inversion algorithm is copy protected.
We define the security of TDF-CP analogously to the search anti-piracy security of SDE.
Specifically, any efficient adversary given a single copy of a quantum trapdoor cannot generate two quantum trapdoors that simultaneously succeed in inverting two independently generated images $y_1$ and $y_2$, respectively.
We give a construction of TDF-CP assuming indistinguishability obfuscation (iO) and the LWE assumption.

\paragraph{Applications.}
As applications of TDF-SKL and TDF-CP, we obtain the following results:
\begin{itemize}
    \item We construct an IND-VRA secure robust PKE-SKL scheme assuming TDF-SKL with domain sampler. 
    \item We construct a CPA$^+$ anti-piracy secure robust SDE scheme assuming TDF-CP.
\end{itemize}
Here, robust PKE-SKL and robust SDE are secure against key-destruction attacks.
In other words, we say that a PKE-SKL (or SDE) scheme is robust if the quantum decryption key remains reusable even after decrypting any ciphertext.
Existing constructions of PKE-SKL and SDE do not satisfy this robustness property, and therefore, this work gives the first construction of robust PKE-SKL and SDE.

We remark that there is an alternative approach to constructing robust PKE-SKL and robust SDE using non-interactive zero-knowledge proofs (NIZKs). Specifically, by appending a NIZK proof attesting to the well-formedness of a ciphertext, one can generically add robustness to existing PKE-SKL and SDE schemes. However, our approach offers the following two advantages, particularly in the case of PKE-SKL.

First, our construction of robust PKE-SKL relies only on black-box use of cryptographic primitives, whereas the NIZK-based construction is inherently non-black-box. Black-box constructions have long been a central and well-studied topic in both classical and quantum cryptography~\cite{STOC:IKLP06,TCC:Haitner08,C:IshPraSah08,C:CCLY22,C:CLPY25}; thus, this distinction is of independent theoretical interest.

Second, our approach admits a broader range of instantiations. Although NIZK proofs are known to exist under the LWE assumption, their existence under group-action-based assumptions  is not known.\footnote{One could employ reusable DV-NIZK \emph{arguments} based on group-action-based assumptions \cite{AC:ADMP20} to obtain PKE-SKL with computational robustness, meaning that quantum decryption keys are not destroyed by decrypting efficiently generated ciphertexts. However, our definition of robustness is statistical, requiring that the quantum decryption keys remain intact even after decrypting ciphertexts that are not necessarily efficiently generated.} In contrast, our construction relies only on PKE and hinting PRG, both of which are known to exist under group-action-based assumptions~\cite{Couveignes06,cryptoeprint:2006/145,AC:AlaPat22}.
\section{Technical Overview}

\subsection{Trapdoor Functions with Secure Key Leasing and Their Application}

\paragraph{Definition of TDFs with secure key leasing.}
We first introduce the definition of trapdoor functions with secure key leasing (TDF-SKL).
Analogously to PKE-SKL, a TDF-SKL scheme $\mathsf{TDFSKL}$ consists of a tuple of five algorithms $(\qKG,\Eval,\qInv,\qDel,\Vrfy)$, where $(\qKG,\Eval,\qInv)$ are defined similarly to standard TDFs except that the key generation algorithm outputs a quantum trapdoor $\qtd$ and a verification key $\vk$ together with a classical evaluation key $\ek$.
The additional algorithms $(\qDel,\Vrfy)$ behave as follows:
\begin{itemize}
    \item $\qDel$: This is a quantum polynomial-time deletion algorithm that takes a quantum trapdoor $\qtd$ and outputs a classical certificate.
    \item $\Vrfy$: This is a classical deterministic polynomial-time verification algorithm that takes the verification key $\vk$ and the certificate $\cert$, and outputs $\top/\bot$.
\end{itemize}
In addition to the inversion correctness, we require the deletion verification correctness, which implies that for honestly generated $\qtd$, $\qDel(\qtd)$ outputs a certificate that passes the verification with overwhelming probability.
The one-wayness for TDF-SKL is defined via a security game between the challenger and an adversary $\qA$.
The security game consists of the following two phases:
\begin{enumerate}
    \item In the first phase, the challenger generates $(\ek,\qtd,\vk)\gets\qKG(1^\secp)$ and sends $(\ek,\qtd)$ to $\qA$. $\qA$ outputs $\cert$. If $\Vrfy(\vk,\cert)=\bot$, then $\qA$ loses. Otherwise, the game proceeds to the next phase.
    \item In the second phase, the challenger generates $x\gets\bit^{n}$, $y\coloneqq\Eval(\ek,x)$, where $\bit^n$ is the domain of the function. The challenger sends $y$ to $\qA$, and $\qA$ outputs $x'$. $\qA$ wins if $x'=x$.
\end{enumerate}
We say that $\mathsf{TDFSKL}$ satisfies one-wayness if for any QPT adversary $\qA$, the winning probability of $\qA$ in the above security game is negligible.

We also define TDF-SKL with domain sampler.
A TDF-SKL with domain sampler is a tuple of six algorithms $(\qKG,\Samp,\Eval,\qInv,\qDel,\Vrfy)$.
The only difference between TDF-SKL and TDF-SKL with domain sampler is that the syntax of TDF-SKL with domain sampler includes an additional algorithm $\Samp$, which behaves as follows:
\begin{itemize}
    \item $\Samp$: This is a PPT domain-sampling algorithm that takes the evaluation key $\ek$ as input, and outputs a bit string $x\in\bit^n$.  
\end{itemize}
We require the same inversion correctness and deletion verification correctness as those of TDF-SKL.
However, for the security of TDF-SKL with domain sampler, we consider a security game that differs in two points from the above security game.
The first difference is that, in the second phase of the game, the challenge input $x$ is sampled according to $\Samp$, rather than being sampled uniformly at random.
The second difference is that, in the second phase, $\qA$ can obtain the verification key $\vk$ together with the challenge image $y$.
This situation is the same as IND-VRA security for PKE-SKL introduced by \cite{EC:KitMorYam25}.
We consider this strong security for TDF-SKL with domain sampler because we will construct IND-VRA secure PKE-SKL that is secure against key-destruction attacks from TDF-SKL with domain sampler.

\paragraph{Construction of TDF-SKL.}
Next, we present the construction of TDF-SKL assuming PKE-SKL with pseudorandom ciphertexts and hinting PRGs.
It is known that PKE-SKL with pseudorandom ciphertexts exist under the LWE assumption \cite{TCC:AnaHuHua24}, and that hinting PRGs can be constructed from several assumptions, such as LWE \cite{C:KopWat19} and a group-action-based assumption \cite{AC:AlaPat22}.
Thus, our construction of TDF-SKL relies on the LWE assumption.
For the variant of TDF-SKL equipped with domain sampler, we give the construction assuming IND-VRA secure PKE-SKL and hinting PRGs.
As shown in \cite{EC:KitMorYam25}, IND-CPA secure PKE schemes imply IND-VRA secure PKE-SKL schemes, and therefore our construction of TDF-SKL with domain sampler is based on the assumption of IND-CPA secure PKE and hinting PRGs.
Since both constructions follow quite similar approaches, we provide only a technical overview of the construction of TDF-SKL here.
We refer the reader to \cref{sec:TDFSKL_DS} for details of the construction of TDF-SKL with domain sampler.

Our construction of TDF-SKL is inspired by \cite{AC:GHMO21}.
While \cite{AC:GHMO21} constructed a TDF based on the mirroring technique of \cite{EC:GarGayHaj19}, this approach does not directly extend to our setting.
Roughly speaking, their TDF takes $x\in\bit^n$ as input, and outputs $(y_1,...,y_n)$, where $y_i=\ct_i$ if $x_i=0$ and $y_i=\ct_i\oplus r_i$ otherwise.
Here $\ct_i$ is a ciphertext of a PKE scheme with pseudorandom ciphertexts, and $r_i$ is sampled uniformly at random.
In their construction, the values $(r_1,...,r_n)$ are included in the evaluation key, and they showed the one-wayness by considering a hybrid game in which $(r_1,...,r_n)$ are defined after $(y_1,...,y_n)$ are fixed.
This reordering enables the use of ciphertext pseudorandomness to make the challenge image independent of the input.
However, this approach fails in the secure key leasing setting.
In the TDF-SKL security game, the adversary first submits a certificate using the evaluation key that includes $(r_1,...,r_n)$, and only later receives the challenge image $(y_1,...,y_n)$. 
Consequently, we cannot apply the same argument as in \cite{AC:GHMO21}, because invoking ciphertext pseudorandomness requires generating $r_i$ in a way that depends on the challenge ciphertext before it is available.
To overcome this difficulty, we modify their construction so that it does not rely on the mirroring technique, using the idea of \cite{C:GarHaj18}, and combine it with Naor’s commitment scheme \cite{JC:Naor91} to achieve almost-all-keys correctness.
More precisely, we construct a TDF-SKL $(\qKG,\Eval,\qInv,\qDel,\Vrfy)$ as follows:
\begin{itemize}
    \item $\qKG$: On input $1^\secp$, generate keys $(\ek_\PKE,\qdk_\PKE,\vk_\PKE)$ of a PKE-SKL with pseudorandom ciphertexts, a public parameter $\pp$ of a hinting PRG, and $n$ random strings $(r_1,...,r_n)$, where $n=n(\secp)$ is a polynomial.
    The evaluation key is $\ek=(\ek_\PKE,\pp,r_1,...,r_n)$.
    The trapdoor is $\qtd=\qdk_\PKE$ and the verification key is $\vk=\vk_\PKE$.
    \item $\Eval$: On input $\ek=(\ek_\PKE,\pp,r_1,...,r_n)$ and $X$, parse $X=(x,s_1,...,s_n,\ct_1,...,\ct_n)$, where $x\in\bit^n$ and the other strings have appropriate lengths. For each $i\in[n]$, 
    \begin{itemize}
        \item if $x_i=0$, let 
        \begin{align}
            y_i\coloneqq
            \begin{pmatrix}
            \PKE.\Enc(\ek_\PKE,s_i;\HPRG.\Eval(\pp,x,i)) \\ 
            \ct_i
            \end{pmatrix}, 
            \quad z_i\coloneqq\PRG(s_i).
        \end{align}
        \item if $x_i=1$, let
        \begin{align}
            y_i\coloneqq
            \begin{pmatrix}
            \ct_i \\
            \PKE.\Enc(\ek_\PKE,s_i;\HPRG.\Eval(\pp,x,i)) 
            \end{pmatrix}, 
            \quad z_i\coloneqq\PRG(s_i) \oplus r_i.
        \end{align}
    \end{itemize}
    Here, $\PKE.\Enc$ is the encryption algorithm of the PKE-SKL with pseudorandom ciphertexts, $\HPRG.\Eval$ is the evaluation algorithm of the hinting PRG, and $\PRG$ is a PRG.
    The output of this algorithm is $(y_1,...,y_n,z_1,...,z_n)$.
    \item $\qInv$: On input $\qtd=\qdk_\PKE$ and $Y$, parse $Y=(y_1,...,y_n,z_1,...,z_n)$, where $y_i=\begin{pmatrix}a_i \\ b_i \end{pmatrix}$.
    For each $i\in[n]$, run $(s_i',\qdk_\PKE')\gets\PKE.\qDec(\qdk_\PKE,a_i)$ coherently.
    \begin{itemize}
        \item If $z_i=\PRG(s_i')$, let $x_i\coloneqq 0$, $s_i\coloneqq s_i'$, and $\ct_i\coloneqq b_i$.
        \item If $z_i\neq\PRG(s_i')$, uncompute $\PKE.\qDec$, and run $(s_i'',\qdk_\PKE'')\gets\PKE.\Dec(\qdk_\PKE,b_i)$. Let $x_i\coloneqq 1$, $s_i\coloneqq s_i''$, and $\ct_i \coloneqq a_i$.
    \end{itemize}
    The output of this algorithm is $(x,s_1,...,s_n,\ct_1,...,\ct_n)$.
    \item $\qDel$ and $\Vrfy$ are the same as those of the underlying PKE-SKL with pseudorandom ciphertexts.
\end{itemize}
To show the one-wayness of TDF-SKL, we consider the following sequence of hybrids:
\begin{itemize}
    \item $\mathsf{Hyb_0}$: This is the original security game that works as follows:
    \begin{enumerate}
    \item The challenger generates keys $(\ek_\PKE,\qdk_\PKE,\vk_\PKE)$ of the PKE-SKL with pseudorandom ciphertexts, a public parameter $\pp$ of the hinting PRG, and $n$ random strings $(r_1,...,r_n)$.
    The challenger sends $(\ek_\PKE,\pp,r_1,...,r_n,\qdk_\PKE)$ to $\qA$.
    \item $\qA$ outputs $\cert$.
    \item If $\PKE.\Vrfy(\vk_\PKE,\cert)=\bot$, then $\qA$ loses.
    Otherwise, the challenger generates $x\gets\bit^n$, $s_i\gets\bit^{m}$, and $\ct_i\gets\bit^c$ for each $i\in[n]$, where $m=m(\secp)$ and $c=c(\secp)$ are polynomials denoting the message and ciphertext lengths of the underlying PKE-SKL scheme, respectively.
    For each $i\in[n]$,
    \begin{itemize}
        \item if $x_i=0$, let
        \begin{align}
            y_i \coloneqq 
            \begin{pmatrix}
                \PKE.\Enc(\ek_\PKE,s_i;\HPRG.\Eval(\pp,x,i)) \\ 
                \ct_i
            \end{pmatrix}, 
            \quad z_i\coloneqq\PRG(s_i).
        \end{align}
        \item if $x_i=1$, let
        \begin{align}
            y_i \coloneqq
            \begin{pmatrix}
                \ct_i \\
                \PKE.\Enc(\ek_\PKE,s_i;\HPRG.\Eval(\pp,x,i)) 
            \end{pmatrix}, 
            \quad z_i\coloneqq\PRG(s_i) \oplus r_i.
        \end{align}
    \end{itemize}
    The challenger sends $(y_1,...,y_n,z_1,...,z_n)$ to $\qA$.
    \item $\qA$ outputs $X'$.
    If $X'=(x,s_1,...,s_n,\ct_1,...,\ct_n)$, then $\qA$ wins.
    Otherwise, $\qA$ loses.
    \end{enumerate}
    \item $\mathsf{Hyb}_1$: This is identical to $\mathsf{Hyb}_0$ except that $(x,s_1,...,s_n)$ is sampled before $\qA$ outputs $\cert$ and for each $i\in[n]$, $r_i$ is replaced with $\PRG(s_i)\oplus u_i$, where $u_i$s are sampled uniformly at random.
    Note that the winning probability of $\qA$ in $\mathsf{Hyb}_1$ is the same as in $\mathsf{Hyb}_0$.
    \item $\mathsf{Hyb}_2$: This is identical to $\mathsf{Hyb}_1$ except that for each $i\in[n]$, $r_i$ and $s_i$ are replaced with $\PRG(s_i^0)\oplus\PRG(s_i^1)$ and $s_i^{x_i}$ respectively, where $s_i^0$s and $s_i^1$s are sampled uniformly at random.
    Then, by the pseudorandomness of $\PRG$, the difference in the winning probability of $\qA$ between $\mathsf{Hyb}_1$ and $\mathsf{Hyb}_2$ is negligible. 
    \item $\mathsf{Hyb}_3$: This is identical to $\mathsf{Hyb}_2$ except that for each $i\in[n]$, $\ct_i$ is replaced with $\PKE.\Enc(\ek_\PKE,s_i^{1-x_i})$.
    Then, the difference in the winning probability of $\qA$ between $\mathsf{Hyb}_2$ and $\mathsf{Hyb}_3$ is negligible because the output of $\PKE.\Enc$ is computationally indistinguishable from a truly random string.
    \item $\mathsf{Hyb}_4$: This is identical to $\mathsf{Hyb}_3$ except that for each $i\in[n]$, $\HPRG.\Eval(\pp,x,i)$ is replaced with a random string. 
    Then, by the security of the hinting PRG, the winning probability of $\qA$ in $\mathsf{Hyb}_3$ is negligibly close to that in $\mathsf{Hyb}_4$.
    In $\mathsf{Hyb}_4$, $\qA$ obtains $(y_1,...,y_n,z_1,...,z_n)$, where
    \begin{align}
            y_i =
            \begin{pmatrix}
                \PKE.\Enc(\ek_\PKE,s_i^0) \\ 
                \PKE.\Enc(\ek_\PKE,s_i^1)
            \end{pmatrix}, 
            \quad z_i=\PRG(s_i^0)
    \end{align}
    for each $i\in[n]$.
    Clearly, the distribution of $(y_1,...,y_n,z_1,...,z_n)$ is independent of $x$.
    Hence, the probability that $\qA$ can successfully invert $(y_1,...,y_n,z_1,...,z_n)$ in $\mathsf{Hyb}_4$ is negligible. 
\end{itemize}
Thus, we conclude that the winning probability of $\qA$ in $\mathsf{Hyb}_0$ is negligible, and we show the one-wayness of TDF-SKL.

\paragraph{Robust PKE-SKL.}
As an application of TDF-SKL, we construct a robust PKE-SKL scheme.
We first formalize the notion of robustness for PKE-SKL schemes as follows:
\begin{itemize}
    \item Let $\mathsf{PKESKL}=(\qKG,\Enc,\qDec,\qDel,\Vrfy)$ be a PKE-SKL scheme. We say that $\mathsf{PKESKL}$ is robust if for any $\ct$, $\TD(\qdk,\qdk')\le\negl(\secp)$, where $(\ek,\qdk,\vk)\gets\qKG(1^\secp)$ and $(m',\qdk')\gets\qDec(\qdk,\ct)$.
\end{itemize}
In this work, we construct an IND-VRA secure robust PKE-SKL scheme assuming TDF-SKL with domain sampler.
However, for consistency with the previous paragraph, we provide an overview of the IND-CPA secure construction assuming TDF-SKL here.
We note that the two constructions are essentially identical and refer the reader to \cref{sec:robustPKESKL} for details of the IND-VRA secure construction.
Our construction relies on the well-known (quantum) Goldreich-Levin theorem \cite{STOC:GolLev89,AC02,C:CLLZ21} and its application to PKE-SKL \cite{EC:AKNYY23}.
We construct an IND-CPA secure robust PKE-SKL $(\qKG,\Enc,\qDec,\qDel,\Vrfy)$ as follows:
\begin{itemize}
    \item $\qKG$: On input $1^\secp$, generate keys $(\ek_\TDF,\qtd_\TDF,\vk_\TDF)$ of the underlying TDF-SKL.
    \item $\Enc$: On input $\ek_\TDF$ and a message $m\in\bit$, output $(y,r,b)$, where $y$ is obtained by evaluating the TDF-SKL on a uniformly random input $x$, $r$ is a random string, and $b=(x\cdot r)\oplus m$.
    \item $\qDec$: On input $\td_\TDF$ and $(y,r,b)$, run the following coherently: Run the inversion algorithm of the TDF-SKL on $y$ to obtain $x'$. If $x'$ is the correct preimage of $y$, then output $(x'\cdot r)\oplus b$. Otherwise, output $\bot$. Copy the result to an ancilla register and uncompute the above evaluation. Measure the ancilla register and return the measurement outcome.
    \item $\qDel$ and $\Vrfy$ are the same as those of the underlying TDF-SKL.
\end{itemize}
We establish IND-CPA security by following the approach of \cite{EC:AKNYY23}, which showed that unpredictability (referred to as OW-KLA security) implies indistinguishability (referred to as IND-KLA security) in the secure key leasing setting.
Moreover, to show robustness, we rely on the gentle measurement lemma \cite{Win99}.
Roughly speaking, the gentle measurement lemma states that if a measurement outcome is almost deterministic, then the post-measurement state remains close to the original state.
Thus, robustness follows if the decryption algorithm of the constructed PKE-SKL behaves almost deterministically on any ciphertext.
We show it by considering the following two cases:
\begin{enumerate}
    \item If $y=\TDF.\Eval(\ek_\TDF,x)$ for some $x$: In this case, by the almost-all-keys inversion correctness of the TDF-SKL, the inversion algorithm outputs $x$ with overwhelming probability.
    \item If $y\neq\TDF.\Eval(\ek_\TDF,x)$ for any $x$: In this case, by the construction of the decryption algorithm, it outputs $\bot$ with probability 1.
\end{enumerate}
Here, $\TDF.\Eval$ denotes the evaluation algorithm of the underlying TDF-SKL.
In both cases, the output of the decryption algorithm is almost deterministic and robustness follows from the gentle measurement lemma.

\subsection{Trapdoor Functions with Copy Protection and Their Application}

\paragraph{Definition of TDFs with copy protection.}
We introduce the definition of trapdoor functions with copy protection (TDF-CP). 
Analogously to SDE, a TDF-CP consists of four algorithms $(\Setup,\qKG,\Eval,\qInv)$, where $\Setup$ generates a pair $(\ek,\td)$ of an evaluation key and a trapdoor, and $\qKG$ takes $\td$ as input and generates a quantum trapdoor $\qtd$.
$\Eval$ is an evaluation algorithm and $\qInv$ is a quantum inversion algorithm.
We require TDF-CP to satisfy  search anti-piracy security, which is defined via the following security game between a challenger and an adversary $\qA'=(\qA,\qB,\qC)$:
\begin{enumerate}
    \item The challenger generates $(\ek,\td)\gets\Setup(1^\secp)$, $\qtd\gets\qKG(\td)$, $x_\qB\gets\bit^n$, $x_\qC\gets\bit^n$ and sends $(\ek,\qtd)$ to $\qA$, where $\bit^n$ is the domain of the TDF-CP.
    \item $\qA$ generates a bipartite state $\mathpzc{q}$ over registers $R_\qB$ and $R_\qC$, and sends $R_\qB$ to $\qB$ and $R_\qC$ to $\qC$.
    \item The challenger sends $y_\qB\coloneqq \Eval(\ek,x_\qB)$ to $\qB$ and $y_\qC\coloneqq \Eval(\ek,x_\qC)$ to $\qC$.
    \item $\qB$ and $\qC$ output $x_\qB'$ and $x_\qC'$, respectively. $\qA'$ wins if $x_\qB=x'_\qB$ and $x_\qC=x'_\qC$.
\end{enumerate}
We say that the TDF-CP $(\Setup,\qKG,\Eval,\qInv)$ is search anti-piracy secure if the probability that $\qA'$ wins the above security game is negligible.

\paragraph{Construction of TDF-CP.}
We present a construction of TDF-CP assuming the existence of indistinguishability obfuscation (iO) and the LWE assumption.
Our construction follows the approach of \cite{C:AnaBeh24}, which provided a simple and modular framework for constructing cryptographic primitives with copy protection.
\cite{C:AnaBeh24} introduced a primitive called unclonable puncturable obfuscation (UPO).
A UPO for a classical circuit class $\{\cC_\secp\}_\secp$ consists of two QPT algorithms $(\qObf,\qEval)$.
$\qObf$ takes as input a circuit $C\in\cC_\secp$, and outputs a quantum program state $\mathpzc{q}$. 
The evaluation algorithm $\qEval$, on input $(\mathpzc{q},x)$, outputs $C(x)$ for all $x$ with overwhelming probability over the randomness of quantum algorithms.
The security of UPO requires that for any QPT adversary $\qA'=(\qA,\qB,\qC)$, $\qA'$ can win the following security game with probability at most negligibly better than the trivial success probability:
\begin{enumerate}
    \item $\qA$ sends $C\in\cC_\secp$ to the challenger. The challenger samples $b\gets\bit$, $x_\qB\gets\cD$, $x_\qC\gets\cD$, where $\cD$ is a distribution with high min-entropy. If $b=0$, $\qA$ obtains $\qObf(1^\secp,C)$. If $b=1$, $\qA$ obtains $\qObf(1^\secp,C^*)$, where $C^*$ is a circuit that on input $x$, outputs $C(x)$ for any $x\notin\{x_\qB,x_\qC\}$ and outputs $\bot$ for $x\in\{x_\qB,x_\qC\}$.
    \item $\qA$ generates a bipartite state and sends the corresponding registers to $\qB$ and $\qC$, respectively.
    \item The challenger sends $x_\qB$ to $\qB$ and $x_\qC$ to $\qC$. $\qB$ and $\qC$ output $b_\qB$ and $b_\qC$, respectively. $\qA'$ wins if $b=b_\qB=b_\qC$.
\end{enumerate}
\cite{C:AnaBeh24} showed that if there exist cryptographic primitives satisfying puncturable security, then by combining them with UPO, one can equip the original primitives with copy protection.
Thus, to construct TDF-CP, it is sufficient to construct TDFs with puncturable security.
Roughly speaking, puncturable security for TDFs requires that one-wayness holds even against a QPT adversary given a punctured trapdoor, where the punctured trapdoor has the same functionality as the original trapdoor on all inputs except for the punctured points.
\cite{AC:CheWanZho18} defined and constructed TDFs with puncturable security for a single punctured point.
In this work, we extend their definition to require puncturable security for two punctured points.
Moreover, following the construction of TDFs in \cite{STOC:SahWat14}, we construct puncturable TDFs from iO, puncturable PRFs, and injective OWFs.
It is known that OWFs imply puncturable PRFs \cite{AC:BonWat13,CCS:KPTZ13,PKC:BoyGolIva14} and iO and OWFs imply injective OWFs \cite{TCC:BitPanWic16}.
Hence, our construction of puncturable TDFs relies on iO and OWFs.
Furthermore, the recent work \cite{cryptoeprint:2025/1880} provided a construction of UPO assuming the existence of iO and the LWE assumption. Combining puncturable TDFs with UPO, we obtain TDF-CP under the assumptions of iO and LWE.

\paragraph{Robust SDE.}
As an application of TDF-CP, we construct a CPA$^+$ anti-piracy secure robust SDE scheme. 
We first define robustness for SDE analogously to that for PKE-SKL.
In other words, we say that an SDE scheme is robust if its quantum decryption key remains reusable after decrypting any ciphertext even if it is not generated by the encryption algorithm.
Our construction of robust SDE follows essentially the same approach as that for robust PKE-SKL.
That is, we rely on the (quantum) Goldreich-Levin theorem \cite{STOC:GolLev89,AC02,C:CLLZ21} and its application to CPA$^+$ anti-piracy secure SDE \cite{TCC:KitYam25}.
\section{Preliminaries}\label{sec:preliminaries}
\subsection{Basic Notations} 
PPT stands for (classical) probabilistic polynomial-time and QPT stands for quantum polynomial-time. 
In this paper, standard math or sans serif font stands for PPT algorithms (e.g., $C$ or $\Gen$) and classical variables (e.g., $x$ or $\pk$). Calligraphic font stands for QPT algorithms (e.g., $\qGen$) and calligraphic
font and/or the bracket notation for (mixed) quantum states (e.g., $\mathpzc{q}$ or $|\psi\rangle$).
We use standard notations of quantum computing and cryptography.
For a bit string $x$, $|x|$ is its length.
$\mathbb{N}$ is the set of natural numbers.
We use $\secp$ as the security parameter.
$[n]$ means the set $\{1,2,...,n\}$.
For a finite set $S$, $x\gets S$ means that an element $x$ is sampled uniformly at random from the set $S$.
$\negl$ is a negligible function, and $\poly$ is a polynomial.
All polynomials that appear in this paper are positive, but for simplicity we do not explicitly mention it.
If we say that an adversary is QPT, it implicitly means non-uniform QPT.
For an algorithm $\cA$, $y\gets \cA(x)$ means that the algorithm $\cA$ outputs $y$ on input $x$.

\subsection{Classical Cryptography}
\begin{definition}[Hinting PRGs \cite{C:KopWat19,AC:GHMO21}]
Let $n,\ell$ be polynomials.
A hinting PRG consists of two PPT algorithms, $\Setup$ and $\Eval$ with the following syntax.
\begin{itemize}
    \item $\Setup(1^\secp)\to\pp$: The setup algorithm takes as input the security parameter $1^\secp$ and outputs the public parameter $\pp$.
    \item $\Eval(\pp,x,i)\to y$: The evaluation algorithm takes as input the public parameter $\pp$, a string $x\in\bit^{n(\secp)}$, and an index $i\in\{0,...,n(\secp)\}$, and outputs $y\in\bit^{\ell(\secp)}$.
\end{itemize}
A hinting PRG $(\Setup,\Eval)$ is required to satisfy the following security:
\begin{itemize}
    \item \textbf{Security:} For any QPT adversary $\qA$, 
    \begin{align}
        &\left|\Pr\left[\beta\gets\qA(\pp,y_0^\beta,\{y_{i,b}^\beta\}_{i\in[n(\secp)],b\in\bit}): 
        \begin{gathered}
            \pp\gets\Setup(1^\secp),x\gets\bit^{n(\secp)}, \beta\gets\bit,\\
            y_0^0\gets\bit^{\ell(\secp)}, y_0^1\coloneqq\Eval(\pp,x,0), \\
            \forall i\in[n(\secp)], \\
            y_{i,0}^0\gets\bit^{\ell(\secp)}, y_{i,1}^0\gets\bit^{\ell(\secp)}, \\
            y_{i,x_i}^1\coloneqq\Eval(\pp,x,i), y_{i,1\oplus x_i}^1\gets\bit^{\ell(\secp)}
        \end{gathered}\right] - \frac{1}{2} \right| \\
        &\le \negl(\secp).
    \end{align}
\end{itemize}
\end{definition}

Hinting PRGs are known to be constructed from several post-quantum secure assumptions such as LWE \cite{C:KopWat19} and a group-action-based assumption \cite{AC:AlaPat22}.

\begin{definition}[iO \cite{JACM:BGIRSVY12}]
    A PPT algorithm $\iO$ is an indistinguishability obfuscation (iO) for a classical circuit class $\{\cC_\secp\}_{\secp\in\N}$ if it satisfies the following:
    \begin{itemize}
        \item \textbf{Correctness:} For all $\secp\in\N$, $C\in\cC_\secp$, and $x$, we have
        \begin{align}
            \Pr[\tilde{C}(x)=C(x) : \tilde{C}\gets\iO(C)] = 1.
        \end{align}
        \item \textbf{Security:} For any QPT adversaries $(\qSamp,\qA)$, the following holds: 
        if 
        \begin{align}
            \Pr[\forall x, ~ C_0(x)=C_1(x) \land |C_0|=|C_1| : (C_0,C_1,\qaux)\gets\qSamp(1^\secp)]\ge1-\negl(\secp),
        \end{align}
        then
        \begin{align}
            &|\Pr[1\gets\qA(\iO(C_0),\qaux):(C_0,C_1,\qaux)\gets\qSamp(1^\secp)] \\
            &\quad - \Pr[1\gets\qA(\iO(C_1),\qaux):(C_0,C_1,\qaux)\gets\qSamp(1^\secp)] | \le \negl(\secp).
        \end{align}
    \end{itemize}
\end{definition}

\begin{definition}[Puncturable PRFs \cite{AC:BonWat13,CCS:KPTZ13,PKC:BoyGolIva14}]
    Let $n,m$ be polynomials.
    A tuple of PPT algorithms $(\KG,\Eval,\Puncture)$ is a puncturable PRF mapping $n(\secp)$-bit strings to $m(\secp)$-bit strings if it satisfies the following conditions:
    \begin{itemize}
        \item \textbf{Punctured correctness:} For any polynomial-sized set $S\subseteq\bit^{n(\secp)}$ and any $x\in\bit^{n(\secp)}\backslash S$, 
        \begin{align}
            \Pr[\Eval(k_S,x)=\Eval(k,x) : k\gets\KG(1^\secp),k_S\gets\Puncture(k,S)] = 1.
        \end{align}
        \item \textbf{Pseudorandom at punctured points:} For any polynomial-sized set $S\subseteq\bit^{n(\secp)}$ and any QPT adversary $\qA$,
        \begin{align}
            |\Pr[1\gets\qA(k_S,\{\Eval(k,x)\}_{x\in S})] - \Pr[1\gets\qA(k_S,(\cU_{m(\secp)})^{|S|})]| \le \negl(\secp),
        \end{align}
        where $k\gets\KG(1^\secp)$, $k_S\gets\Puncture(k,S)$ and $\cU_{m(\secp)}$ is the uniform distribution over $m(\secp)$ bits.
    \end{itemize}
\end{definition}

It is known that puncturable PRFs can be constructed from OWFs \cite{AC:BonWat13,CCS:KPTZ13,PKC:BoyGolIva14}.

\begin{definition}[Injective OWFs]
    Let $n,m$ be polynomials. 
    An injective OWF is given by a PPT algorithm $\Gen$ that takes $1^\secp$ as input, and outputs a description of a classical-polynomial-time-computable function $f:\bit^{n(\secp)}\to\bit^{m(\secp)}$. It satisfies the following properties:
    \begin{itemize}
        \item \textbf{Injectivity:} With overwhelming probability over the choice of $f\gets\Gen(1^\secp)$, $f$ is injective.
        \item \textbf{One-Wayness:} For any QPT adversary $\qA$,
        \begin{align}
            \Pr[f(x')=f(x): f\gets\Gen(1^\secp), x\gets\bit^{n(\secp)}, x'\gets\qA(1^\secp,f,f(x))] \le \negl(\secp).
        \end{align}
    \end{itemize}
\end{definition}

It is known that iO and standard OWFs imply injective OWFs.
\begin{lemma}[\cite{TCC:BitPanWic16}]
    If iO and OWFs exist, then injective OWFs exist.
\end{lemma}

\subsection{Unclonable Cryptography}

\paragraph{Secure Key Leasing.}
We review the definition of PKE schemes with secure key leasing (PKE-SKL).
Throughout this paper, we focus on PKE-SKL schemes with classical revocation \cite{TCC:AnaPorVai23,EC:KitMorYam25}.
Moreover, we consider almost-all-keys decryption correctness that is introduced in \cite{EC:DwoNaoRei04} for PKE schemes.
\begin{definition}[PKE-SKL]
A PKE-SKL scheme for message space $\cM$ is a tuple of algorithms $(\mathpzc{KG},\Enc,\mathpzc{Dec},\mathpzc{Del},\Vrfy)$.
\begin{itemize}
    \item $\mathpzc{KG}(1^\secp)\to(\ek,\mathpzc{dk},\vk)$: The key generation algorithm is a QPT algorithm that takes a security parameter $1^\secp$, and outputs an encryption key $\ek$, a quantum decryption key $\mathpzc{dk}$, and a deletion verification key $\vk$.
    \item $\Enc(\ek,m)\to\ct$: The encryption algorithm is a PPT algorithm that takes an encryption key $\ek$ and a message $m\in\cM$, and outputs a ciphertext $\ct$.
    \item $\mathpzc{Dec}(\mathpzc{dk},\ct)\to (m',\mathpzc{dk}'):$ The decryption algorithm is a QPT algorithm that takes a decryption key $\mathpzc{dk}$ and a ciphertext $\ct$ as input,
    and outputs $m'$ and a resulting decryption key $\mathpzc{dk}'$.
    \item $\mathpzc{Del}(\mathpzc{dk})\to\cert:$ The deletion algorithm is a QPT algorithm that takes a decryption key $\mathpzc{dk}$ as input and outputs a classical certificate $\cert$.
    \item $\Vrfy(\vk,\cert)\to\top/\bot:$ The deletion verification algorithm is a deterministic classical polynomial-time algorithm that takes a deletion verification key $\vk$ and a deletion certificate $\cert$, and outputs $\top/\bot$.
\end{itemize}
We require the following correctness properties:
\begin{itemize}
    \item \textbf{Almost-all-keys decryption correctness:} 
    \begin{align}
        \Pr_{(\ek,\qdk,\vk)\gets\qKG(1^\secp)} \left[
        \begin{lgathered}
            \forall m\in\cM, \forall r\in\bit^{n(\secp)}, \\
            \Pr[m\gets\qDec(\qdk,\Enc(\ek,m;r))]\ge 1-\negl(\secp)
        \end{lgathered}
        \right] \ge 1-\negl(\secp),
    \end{align}
    where $n(\secp)$ denotes the randomness length of $\Enc$.
    \item \textbf{Deletion verification correctness:}
    \begin{align}
        \Pr \left[\Vrfy(\vk,\cert)=\top:
        \begin{gathered}
            (\ek,\mathpzc{dk},\vk)\gets\mathpzc{KG}(1^\secp), \\
            \cert\gets\mathpzc{Del}(\mathpzc{dk})
        \end{gathered}
        \right] \ge 1-\negl(\secp).
    \end{align}
\end{itemize}
\end{definition}

\begin{remark}
    As discussed in \cite{EC:KitMorYam25}, we can assume without loss of generality that a decryption key of a PKE-SKL scheme is reusable, i.e., it can be reused to decrypt polynomially many honestly generated ciphertexts. 
    In particular, for honestly generated $\ct$ and $\qdk$, decrypting $\ct$ with $\qdk$ leaves the resulting decryption key that is negligibly close to $\qdk$ in trace distance. This follows from decryption correctness of PKE-SKL and the gentle measurement lemma~(\cref{lem:gentle}). The same remark applies to all encryption schemes considered in this paper (SDE, TDF-SKL, and TDF-CP).
    However, we emphasize that when decrypting ciphertexts that are not honestly generated, the post-decryption state of $\qdk$ is not necessarily close to $\qdk$.
    One of our contributions is the construction of PKE-SKL and SDE whose decryption keys remain reusable even after decrypting \textit{any} ciphertexts.
    See \cref{sec:robustPKESKL,sec:robustSDE} for details.
\end{remark}

We review several security definitions for PKE-SKL, namely IND-CPA security, IND-VRA security, and PRCT security.

\begin{definition}[IND-CPA Security]
We say that a PKE-SKL scheme $\mathsf{PKESKL}=(\qKG,\Enc,\qDec,\qDel,\Vrfy)$ for message space $\cM$ is IND-CPA secure if it satisfies the following requirement, formalized by the experiment $\mathsf{Exp}^{\mathsf{IND}\textrm{-}\mathsf{CPA}}_{\mathsf{PKESKL},\qA}(\secp,b)$ between the challenger and an adversary $\qA$:
\begin{enumerate}
    \item The challenger generates $(\ek,\mathpzc{dk},\vk)\gets\mathpzc{KG}(1^\secp)$ and sends $(\ek,\mathpzc{dk})$ to $\qA$.
    \item $\qA$ sends $\cert$ and $(m_0,m_1)\in\cM^2$ to the challenger.
    If $\Vrfy(\vk,\cert)=\bot$, the challenger outputs 0 as the final output of this experiment.
    Otherwise, the challenger generates $\ct\gets\Enc(\ek,m_b)$, and sends $\ct$ to $\qA$.
    \item $\qA$ outputs a guess $b'$ for $b$. The challenger outputs $b'$ as the final output of the experiment.
    Then, for any QPT $\qA$, 
    \begin{align}
        |\Pr [1\gets\mathsf{Exp}^{\mathsf{IND\textrm{-}CPA}}_{\mathsf{PKESKL},\qA}(\secp,0)]-\Pr[1\gets\mathsf{Exp}^{\mathsf{IND\textrm{-}CPA}}_{\mathsf{PKESKL},\qA}(\secp,1)]| \le\negl(\secp).
    \end{align}
\end{enumerate}
\end{definition}

\begin{definition}[IND-VRA Security \cite{EC:KitMorYam25}]
We say that a PKE-SKL scheme $\mathsf{PKESKL}=(\qKG,\Enc,\qDec,\qDel,\Vrfy)$ for message space $\cM$ is IND-VRA secure if it satisfies the following requirement, formalized by the experiment $\mathsf{Exp}^{\mathsf{IND\textrm{-}VRA}}_{\mathsf{PKESKL},\qA}(\secp,b)$ between the challenger and an adversary $\qA$:
\begin{enumerate}
    \item The challenger generates $(\ek,\mathpzc{dk},\vk)\gets\mathpzc{KG}(1^\secp)$ and sends $(\ek,\mathpzc{dk})$ to $\qA$.
    \item $\qA$ sends $\cert$ and $(m_0,m_1)\in\cM^2$ to the challenger.
    If $\Vrfy(\vk,\cert)=\bot$, the challenger outputs 0 as the final output of this experiment.
    Otherwise, the challenger generates $\ct\gets\Enc(\ek,m_b)$, and sends $(\ct,\vk)$ to $\qA$.
    \item $\qA$ outputs a guess $b'$ for $b$. The challenger outputs $b'$ as the final output of the experiment.
    Then, for any QPT $\qA$, 
    \begin{align}
        |\Pr [1\gets\mathsf{Exp}^{\mathsf{IND\textrm{-}VRA}}_{\mathsf{PKESKL},\qA}(\secp,0)]-\Pr[1\gets\mathsf{Exp}^{\mathsf{IND\textrm{-}VRA}}_{\mathsf{PKESKL},\qA}(\secp,1)]| \le\negl(\secp).
    \end{align}
\end{enumerate}
\end{definition}

\cite{EC:KitMorYam25} showed that IND-CPA secure PKE schemes imply IND-VRA secure PKE-SKL schemes.
It is straightforward to verify that their construction satisfies almost-all-keys decryption correctness if the underlying PKE scheme satisfies almost-all-keys correctness.
\begin{lemma}[Derived from \cite{EC:KitMorYam25}]
    If IND-CPA secure PKE schemes that satisfy almost-all-keys correctness exist, then IND-VRA secure PKE-SKL schemes that satisfy almost-all-keys decryption correctness exist.
\end{lemma}

\begin{definition}[PRCT Security \cite{TCC:AnaHuHua24}]
We say that a PKE-SKL scheme $\mathsf{PKESKL}=(\qKG,\Enc,\qDec,\qDel,\Vrfy)$ for message space $\cM$ and ciphertext space $\bit^{\ell(\secp)}$ is PRCT secure if it satisfies the following requirement, formalized by the experiment $\mathsf{Exp}^{\mathsf{PRCT}}_{\mathsf{PKESKL},\qA}(\secp,b)$ between the challenger and an adversary $\qA$:
\begin{enumerate}
    \item The challenger generates $(\ek,\mathpzc{dk},\vk)\gets\mathpzc{KG}(1^\secp)$ and sends $(\ek,\mathpzc{dk})$ to $\qA$.
    \item $\qA$ sends $\cert$ and $m\in\cM$ to the challenger.
    If $\Vrfy(\vk,\cert)=\bot$, the challenger outputs 0 as the final output of this experiment.
    Otherwise, the challenger generates $\ct_b$, where $\ct_0\gets\Enc(\ek,m)$ and $\ct_1\gets\bit^{\ell(\secp)}$, and sends $\ct_b$ to $\qA$.
    \item $\qA$ outputs a guess $b'$ for $b$. The challenger outputs $b'$ as the final output of the experiment.
    Then, for any QPT $\qA$, 
    \begin{align}
        |\Pr [1\gets\mathsf{Exp}^{\mathsf{PRCT}}_{\mathsf{PKESKL},\qA}(\secp,0)]-\Pr[1\gets\mathsf{Exp}^{\mathsf{PRCT}}_{\mathsf{PKESKL},\qA}(\secp,1)]| \le\negl(\secp).
    \end{align}
\end{enumerate}
\end{definition}

\cite{TCC:AnaHuHua24} constructed PRCT secure PKE-SKL schemes under the LWE assumption.
\begin{lemma}
    There exist PRCT secure PKE-SKL schemes under the LWE assumption.
\end{lemma}


\paragraph*{Copy Protection.}

We review the definition of single-decryptor encryption (SDE).
As with the case of PKE-SKL, we consider almost-all-keys correctness.
\begin{definition}[Single-Decryptor Encryption \cite{GZ20,C:CLLZ21}]
    A single-decryptor encryption (SDE) scheme $\mathsf{SDE}$ is a tuple of four algorithms $(\Setup,\qKG,\Enc,\qDec)$.
    Below, let $\cM$ be the message space of $\mathsf{SDE}$.
    \begin{itemize}
        \item $\Setup(1^\secp)\to(\pk,\sk)$: The setup algorithm takes a security parameter $1^\secp$, and outputs a public key $\pk$ and a secret key $\sk$.
        \item $\qKG(\sk)\to\qsk$: The key generation algorithm takes a secret key $\sk$, and outputs a quantum decryption key $\qsk$.
        \item $\Enc(\pk,m)\to\ct$: The encryption algorithm takes a public key $\pk$ and a message $m\in\cM$, and outputs a ciphertext $\ct$.
        \item $\qDec(\qsk,\ct)\to (m',\qsk')$: The decryption algorithm takes a quantum decryption key $\qsk$ and a ciphertext $\ct$, and outputs a message $m'\in\{\bot\}\cup\cM$ and a resulting decryption key $\qsk'$. 
    \end{itemize}
    \textbf{Correctness:} For all $m\in\cM$, 
    \begin{align}
        \Pr\left[m\gets\qDec(\qsk,\ct):
        \begin{gathered}
            (\pk,\sk)\gets\Setup(1^\secp) \\ 
            \qsk\gets\qKG(\sk) \\ 
            \ct\gets\Enc(\pk,m)
        \end{gathered}\right] \ge 1-\negl(\secp).
    \end{align}
\end{definition}

In this paper, we focus on the CPA$^+$ anti-piracy security of SDE that is introduced in \cite{TCC:KitYam25}.
\begin{definition}[CPA$^+$ Anti-Piracy Security \cite{TCC:KitYam25}]
    Let $\SDE=(\Setup,\qKG,\Enc,\qDec)$ be an SDE scheme.
    We consider the CPA$^+$ anti-piracy game $\mathsf{Exp}^{\mathsf{CPA}^+}_{\SDE,\qA_\SDE}(\secp)$ between the challenger and an adversary $\qA_\SDE=(\qA,\qB,\qC)$ below. 
    \begin{enumerate}
        \item The challenger generates $(\pk,\sk)\gets\Setup(1^\secp)$, $\qsk\gets\qKG(\sk)$ and sends $(\pk,\qsk)$ to $\qA$.
        \item $\qA$ chooses two pairs of messages $(m_{\qB,0},m_{\qB,1})\in\cM^2$ and $(m_{\qC,0},m_{\qC,1})\in\cM^2$ such that $|m_{\qB,0}|=|m_{\qB,1}|$ and $|m_{\qC,0}|=|m_{\qC,1}|$, and creates a bipartite state $\mathpzc{q}$ over registers $\regR_\qB$ and $\regR_\qC$.
        Then, $\qA$ sends $(m_{\qB,0},m_{\qB,1},m_{\qC,0},m_{\qC,1})$ to the challenger, register $\regR_\qB$ to $\qB$, and register $\regR_\qC$ to $\qC$.
        \item The challenger chooses $\coin_\qB\gets\bit$ and $\coin_\qC\gets\bit$, generates $\ct_\qB\gets\Enc(\pk,m_{\qB,\coin_\qB})$ and $\ct_\qC\gets\Enc(\pk,m_{\qC,\coin_{\qC}})$, and sends $\ct_\qB$ and $\ct_\qC$ to $\qB$ and $\qC$, respectively.
        \item $\qB$ and $\qC$ respectively output $\coin'_\qB$ and $\coin'_\qC$. The challenger outputs 1 if $\coin'_\qB\oplus\coin'_\qC=\coin_\qB\oplus\coin_\qC$. Otherwise, the challenger outputs 0.
    \end{enumerate}
    We say that $\SDE$ is CPA$^+$ anti-piracy secure if for any QPT adversary $\qA_\SDE$, 
    \begin{align}
        \Pr[1\gets\mathsf{Exp}^{\mathsf{CPA}^+}_{\SDE,\qA_\SDE}(\secp)] \le \frac{1}{2} + \negl(\secp).
    \end{align}
\end{definition}

Ananth and Behera \cite{C:AnaBeh24} introduced unclonable puncturable obfuscation (UPO) as a general framework for constructing copy protected cryptographic primitives.
In this work, we focus on UPO with almost-all-keys correctness, disregarding errors caused by quantum algorithms.

\begin{definition}[Unclonable Puncturable Obfuscation \cite{C:AnaBeh24}]
\label{def:UPO}
    An unclonable puncturable obfuscation (UPO) scheme for a circuit class $\cC=\{\cC_\secp:\bit^{\ell_\text{in}}\to\bit^{\ell_\text{out}}\}_{\secp\in\N}$ is a pair of algorithms $(\qObf,\qEval)$ with the following syntax:
    \begin{itemize}
        \item $\qObf(1^\secp,C)$: The obfuscation algorithm is a QPT algorithm that takes a security parameter $1^\secp$ and $C\in\cC_\secp$ as input and outputs a quantum state $\mathpzc{p}$.
        \item $\qEval(\mathpzc{p},x)\to (y,\mathpzc{p'})$: The evaluation algorithm is a QPT algorithm that takes $\mathpzc{p}$ and $x\in\bit^{\ell_\text{in}}$ and outputs $y\in\bit^{\ell_\text{out}}$ and a resulting state $\mathpzc{p'}$.
    \end{itemize}
    We require the following properties:
    \begin{itemize}
        \item \textbf{Almost-All-Keys Correctness:} For any $C\in\cC_\secp$, 
        \begin{align}
            \Pr_{\mathpzc{p}\gets\qObf(1^\secp,C)}\left[\forall x\in\bit^{\ell_{in}},~ \Pr_{(y,\mathpzc{p'})\gets\qEval(\mathpzc{p},x)}[y=C(x)]\ge 1-\negl(\secp) \right] \ge 1-\negl(\secp).
        \end{align}
        \item \textbf{$\cD$-UPO Security:} Let $\cD$ be a distribution over $\bit^{\ell_{\text{in}}}\times\bit^{\ell_{\text{in}}}$.
        Consider the experiment $\mathsf{Exp}_{\cD,\qA'}$ between the challenger and an adversary $\qA'=(\qA,\qB,\qC)$ that works as follows:
        \begin{enumerate}
            \item $\qA$ sends $C\in\cC_\secp$ to the challenger.
            \item The challenger generates $b\gets\bit$, $(x_\qB,x_\qC)\gets\cD$ and generates a quantum state $\mathpzc{p}$ as follows:
            \begin{itemize}
                \item If $b=0$, $\mathpzc{p}\gets\qObf(1^\secp,C)$.
                \item If $b=1$, $\mathpzc{p}\gets\qObf(1^\secp,C^*[x_\qB,x_\qC])$, 
            \end{itemize}
            where 
            \begin{align}
                C^*[x_\qB,x_\qC] = 
                \begin{cases}
                    C(x) & \text{if } x\in\bit^{\ell_{\text{in}}}\backslash\{x_\qB,x_\qC\} \\
                    \bot & \text{if } x\in \{x_\qB,x_\qC\}.
                \end{cases}
            \end{align}
            The challenger sends $\mathpzc{p}$ to $\qA$.
            \item $\qA$ generates a bipartite state $\mathpzc{q}$ over registers $\regR_\qB$ and $\regR_\qC$.
            $\qA$ sends $\regR_\qB$ to $\qB$ and $\regR_\qC$ to $\qC$.
            \item The challenger sends $x_\qB$ and $x_\qC$ to $\qB$ and $\qC$, respectively.
            \item $\qB$ and $\qC$ output $b_\qB$ and $b_\qC$. 
            The output of the experiment is 1 if $b=b_\qB=b_\qC$.
        \end{enumerate}
        Then, for any QPT adversary $\qA'=(\qA,\qB,\qC)$,
        \begin{align}
            \Pr[1\gets\mathsf{Exp}_{\cD,\qA'}(\secp)] \le \frac{1}{2} + \negl(\secp).
        \end{align}
    \end{itemize}
\end{definition}

The recent work \cite{cryptoeprint:2025/1880} showed that UPOs exist assuming the existence of iO and LWE.

\begin{lemma}[\cite{cryptoeprint:2025/1880}]
    Assuming the existence of iO and the LWE assumption, for any constant $c>0$ and any polynomials $\ell_{inp}$ and $\ell_{out}$, there exists a UPO for the circuit class $\mathsf{Circ}= \{C : \bit^{\ell_{inp}} \to\bit^{\ell_{out}} \}$ that satisfies $\cD$-UPO security for any QPT-samplable product distributions $\cD=\cD_\qB\times\cD_\qC$, where $\cD_\qB$ and $\cD_\qC$ have min-entropy at least $\secp^c$.
\end{lemma}
Moreover, if iO exist, we can assume that the UPO also satisfies iO-security without loss of generality \cite{cryptoeprint:2025/1880}.

We review a variant of puncturable secure circuits \cite{C:AnaBeh24,cryptoeprint:2025/1880}.
This variant allows the circuit and its challenge-input distribution to be correlated.
\begin{definition}[Puncturable Secure Circuits]
    For any polynomial $\ell\coloneqq\ell(\secp)$, let $\mathsf{Circ}=\{C : \bm{X}\to\bm{Y}\}$ be a circuit class, where $\bm{X}=\bit^n$, $\bm{Y}=\bit^m$ for polynomials $m\coloneqq m(\secp)$, $n\coloneqq n(\secp)$, equipped with an efficient deterministic algorithm $\Puncture$ that satisfies puncturing correctness.
    Namely, for every $C\in\mathsf{Circ}$ and tuple $(x_1,\ldots,x_{\ell'})$, the circuit $\widehat C\gets\Puncture(C,(x_1,\ldots,x_{\ell'}))$ outputs $C(x)$ for every $x\notin\{x_1,\ldots,x_{\ell'}\}$ and outputs $\bot$ on $\{x_1,\ldots,x_{\ell'}\}$.
    Let $\Gen_{\mathsf{Circ}}(1^\secp)$ be an efficient algorithm that outputs a public parameter $z$ and a circuit $C\in\mathsf{Circ}$, and let $\{\cD_z\}_z$ be a family of efficiently samplable distributions over $\bm{X}$.
    We say that $(\mathsf{Circ},\Puncture)$ satisfies $\ell$-point $m$-bit $(\Gen_{\mathsf{Circ}},\{\cD_z\}_z)$-unpredictability-style puncturing security if for every QPT adversary $\qA$ and every $\ell'\le\ell$, the probability that $\qA$ succeeds in the following security game is at most $1-\left(1-\frac{1}{2^m}\right)^{\ell'}+\negl(\secp)$:
    \begin{enumerate}
        \item $(z,C)\gets\Gen_{\mathsf{Circ}}(1^\secp)$.
        \item Independently sample $x_1\gets\cD_z,\ldots,x_{\ell'}\gets\cD_z$.
        \item $\hat{C}\gets\Puncture(C,(x_1,...,x_{\ell'}))$.
        \item $y'_1,...,y'_{\ell'}\gets\qA(z,\hat{C},(x_1,...,x_{\ell'}))$.
        \item $\qA$ wins if there exists $i\in[\ell']$ such that $y'_i=C(x_i)$.
    \end{enumerate}
\end{definition}

We can show the following lemma by modifying the modular implication from puncturable secure circuits and UPO to copy protection \cite{C:AnaBeh24,cryptoeprint:2025/1880}.

\begin{lemma}
\label{lem:modular}
    Suppose $m=\omega(\log\secp)$ and $(\mathsf{Circ},\Puncture)$ satisfies 2-point $m$-bit $(\Gen_{\mathsf{Circ}},\{\cD_z\}_z)$-unpredictability-style puncturing security.
    Let $(\mathsf{UPO}.\qObf,\mathsf{UPO}.\qEval)$ be a UPO for the circuit class $\mathsf{Circ}$ that satisfies security as iO and $(\cD_z\times\cD_z)$-UPO security for every $z$ in the support of the first output of $\Gen_{\mathsf{Circ}}$.
    Then for any QPT adversary $\qA'=(\qA,\qB,\qC)$, 
    \begin{align}
        \Pr[1\gets\mathsf{Exp}_{\qA'}(\secp)] \le \negl(\secp),
    \end{align}
    where $\mathsf{Exp}_{\qA'}$ is the experiment between the challenger and an adversary $\qA'=(\qA,\qB,\qC)$ that works as follows:
    \begin{enumerate}
        \item The challenger generates $(z,C)\gets\Gen_{\mathsf{Circ}}(1^\secp)$ and $\mathpzc{p}\gets\mathsf{UPO.}\qObf(1^\secp,C)$, and sends $(z,\mathpzc{p})$ to $\qA$.
        \item $\qA$ produces a bipartite state $\mathpzc{q}$ over registers $\regR_\qB$ and $\regR_\qC$. $\qA$ sends $\regR_\qB$ to $\qB$ and $\regR_\qC$ to $\qC$.
        \item The challenger independently samples $x_\qB\gets\cD_z$ and $x_\qC\gets\cD_z$, and sends $x_\qB$ to $\qB$ and $x_\qC$ to $\qC$.
        \item $\qB$ and $\qC$ output $y_\qB$ and $y_\qC$.
        \item The output of the experiment is 1 if $y_\qB=C(x_\qB)$ and $y_\qC=C(x_\qC)$.
    \end{enumerate}
\end{lemma}

\subsection{Useful Lemmas}
\begin{lemma}[Gentle Measurement \cite{Win99}]
\label{lem:gentle}
    Let $\rho$ be a quantum state and $X$ a positive operator with $X\le I$ and $1-\Tr[X\rho]\le \epsilon \le 1$.
    Then,
    \begin{align}
        \left\|\rho-\sqrt{X}\rho\sqrt{X} \right\|_1 \le \sqrt{8\epsilon}.
    \end{align}
\end{lemma}

\begin{lemma}[Quantum Goldreich-Levin with Quantum Auxiliary Input \cite{C:CLLZ21}]
\label{lem:quantGL}
    There exists a QPT algorithm $\mathpzc{Ext}$ that satisfies the following:
    Let $n\in\N$, $x\in\bit^n$, $\epsilon\in[0,1/2]$, and $\qA$ be a quantum algorithm with a quantum auxiliary input $\mathpzc{aux}$ such that
    \begin{align}
        \Pr[x\cdot r\gets\qA(\mathpzc{aux,r}):r\gets\bit^n] \ge \frac{1}{2}+\epsilon.
    \end{align}
    Then, we have 
    \begin{align}
        \Pr[x\gets \mathpzc{Ext}([\qA],\mathpzc{aux})] \ge 4\epsilon^2,
    \end{align}
    where $[\qA]$ denotes the classical description of $\qA$.
\end{lemma}


\section{Trapdoor Functions with Secure Key Leasing} 
In this section, we introduce the definition of trapdoor functions with secure key leasing (TDF-SKL) and give the construction of TDF-SKL under the LWE assumption (\cref{sec:TDFSKL}).
Moreover, we define TDF-SKL with domain sampler and give the construction assuming IND-CPA secure PKE and hinting PRGs (\cref{sec:TDFSKL_DS}).
As an application of TDF-SKL with domain sampler, we construct an IND-VRA secure PKE-SKL scheme with the additional property of \textit{robustness} in \cref{sec:robustPKESKL}.
Robustness ensures that the quantum decryption key of a PKE-SKL scheme remains reusable after decrypting any ciphertext, even if it is not honestly generated by the encryption algorithm. 

\subsection{Trapdoor Functions with Secure Key Leasing}
\label{sec:TDFSKL}
We define trapdoor functions with secure key leasing as follows:
\begin{definition}[Trapdoor Functions with Secure Key Leasing]
Let $n$ be a polynomial.
A trapdoor function with secure key leasing (TDF-SKL) is a tuple of algorithms $\mathsf{TDFSKL}=(\mathpzc{KG},\Eval,\mathpzc{Inv},\mathpzc{Del},\Vrfy)$.
\begin{itemize}
    \item $\mathpzc{KG}(1^\secp)\to(\ek,\mathpzc{td},\vk):$ The key generation algorithm takes as input a security parameter $1^\secp$ and outputs a classical evaluation key $\ek$, a quantum trapdoor $\mathpzc{td}$, and a classical deletion verification key $\vk$.
    \item $\Eval(\ek,x)\to y:$ The evaluation algorithm takes an evaluation key $\ek$ and a string $x\in\bit^{n(\secp)}$ as input and outputs $y$. This algorithm is deterministic.
    \item $\mathpzc{Inv}(\mathpzc{td},y)\to(x',\mathpzc{td}'):$ The inversion algorithm takes a quantum trapdoor $\mathpzc{td}$ and a string $y$ as input and outputs $x'$ and a resulting trapdoor $\mathpzc{td}'$.
    \item $\mathpzc{Del}(\mathpzc{td})\to\cert:$ The trapdoor deletion algorithm takes a quantum trapdoor $\mathpzc{td}$ as input and outputs a classical certificate $\cert$.
    \item $\Vrfy(\vk,\cert)\to\top/\bot:$ The deletion verification algorithm takes a deletion verification key $\vk$ and a certificate $\cert$ as input and outputs $\top/\bot$. This algorithm is deterministic.
\end{itemize}
$\mathsf{TDFSKL}$ is required to satisfy the following conditions:
\begin{itemize}
    \item \textbf{Almost-all-keys inversion correctness:} 
    \begin{align}
        \Pr_{(\ek,\mathpzc{td},\vk)\gets\mathpzc{KG}(1^\secp)} \left[ \forall x\in\bit^{n(\secp)}, 
        \Pr_{(x',\mathpzc{td}')\gets\mathpzc{Inv}(\mathpzc{td},\Eval(\ek,x))}[x'=x]\ge 1-\negl(\secp) \right] \ge 1-\negl(\secp).
    \end{align}
    \item \textbf{Deletion verification correctness:} 
    \begin{align}
        \Pr \left[ \Vrfy(\vk,\cert)=\top : \begin{gathered} (\ek,\mathpzc{td},\vk)\gets\mathpzc{KG}(1^\secp) \\ \cert\gets\mathpzc{Del}(\mathpzc{td})
        \end{gathered} \right] \ge 1-\negl(\secp).
    \end{align}
    \item \textbf{One-wayness:}
    Consider the experiment $\mathsf{Exp}^{\mathsf{OW}}_{\mathsf{TDFSKL},\qA}(\secp)$ between the challenger and an adversary $\qA$:
    \begin{enumerate}
        \item The challenger generates $(\ek,\mathpzc{td},\vk)\gets\mathpzc{KG}(1^\secp)$ and sends $(\ek,\mathpzc{td})$ to $\qA$.
        \item $\qA$ sends $\cert$ to the challenger.
        \item If $\Vrfy(\vk,\cert)=\bot$, then the output of the experiment is 0. 
        Otherwise, the challenger generates $x\gets\bit^{n(\secp)}, y\coloneqq\Eval(\ek,x)$ and sends $y$ to $\qA$.
        \item $\cA$ sends $x'$ to the challenger.
        \item If $x'=x$, the output of the experiment is 1.
        Otherwise, the output of the experiment is 0.
    \end{enumerate}
    Then, for any QPT adversary $\qA$, 
    \begin{align}
        \Pr[1\gets\mathsf{Exp}^{\mathsf{OW}}_{\mathsf{TDFSKL},\qA}(\secp)] \le\negl(\secp).
    \end{align}
\end{itemize}
\end{definition}

We show the following theorem.
\begin{theorem}
    If PRCT secure PKE-SKL and hinting PRGs exist, then TDF-SKL exist.
\end{theorem}
Our construction of TDF-SKL is as follows:
\paragraph{Construction.}
Let $\mathsf{PKESKL}=(\PKE.\qKG,\PKE.\Enc,\PKE.\qDec,\PKE.\qDel,\PKE.\Vrfy)$ be a PRCT secure PKE-SKL scheme with message length $m=m(\secp)$, randomness length $\ell=\ell(\secp)$, and ciphertext length $c=c(\secp)$.
Let $(\HPRG.\Setup,\HPRG.\Eval)$ be a hinting PRG with input length $n=n(\secp)$ and output length $\ell$.
Let $\PRG$ be a PRG mapping $m$-bit strings to $3m$-bit strings.
We construct a TDF-SKL $\mathsf{TDFSKL}=(\qKG,\Eval,\qInv,\qDel,\Vrfy)$ with domain $\bit^{n(1+m+c)}$ as follows:
\begin{itemize}
    \item $\qKG(1^\secp)\to(\ek,\qtd,\vk)$: 
    \begin{enumerate}
        \item Run $(\ek_\PKE,\qdk_\PKE,\vk_\PKE)\gets\PKE.\qKG(1^\secp)$, $\pp\gets\HPRG.\Setup(1^\secp)$.
        For $i\in[n]$, sample $r_i\gets\bit^{3m}$.
        \item Return $\ek\coloneqq (\ek_\PKE,\pp,r_1,...,r_n)$, $\qtd\coloneq \qdk_\PKE$, and $\vk\coloneqq \vk_\PKE$.
    \end{enumerate}
    \item $\Eval(\ek,X)\to Y$:
    \begin{enumerate}
        \item Parse $\ek=(\ek_\PKE,\pp,r_1,...,r_n)$ and $X=(x,s_1,...,s_n,\ct_1,...,\ct_n)$, where $x\in\bit^n$, $s_i\in\bit^m$, and $\ct_i\in\bit^c$ for each $i\in[n]$.
        \item For $i\in[n]$,
        \begin{itemize}
            \item if $x_i=0$, let 
            \begin{align}
                y_i\coloneqq  
                \begin{pmatrix}
                    \PKE.\Enc(\ek_\PKE,s_i;\HPRG.\Eval(\pp,x,i)) \\ 
                    \ct_i
                \end{pmatrix}, 
                \quad z_i\coloneqq \PRG(s_i).
            \end{align}
            \item if $x_i=1$, let
            \begin{align}
                y_i\coloneqq  
                \begin{pmatrix}
                    \ct_i \\
                    \PKE.\Enc(\ek_\PKE,s_i;\HPRG.\Eval(\pp,x,i)) 
                \end{pmatrix}, 
                \quad z_i\coloneqq \PRG(s_i) \oplus r_i.
            \end{align}
        \end{itemize}
        \item Return $Y\coloneqq (y_1,...,y_n,z_1,...,z_n)$.
    \end{enumerate}
    \item $\qInv(\qtd,Y)\to (X',\qtd')$:
    \begin{enumerate}
        \item Parse $\qtd=\qdk_\PKE$ and $Y=(y_1,...,y_n,z_1,...,z_n)$, where for each $i\in[n]$, $y_i=\begin{pmatrix} a_i \\ b_i \end{pmatrix}$.
        Initialize the decryption-key register as $K_0\coloneqq\qdk_\PKE$.
        \item For each $i\in[n]$, perform the following operation coherently, using $K_{i-1}$ as the decryption-key register:
        \begin{enumerate}
            \item Apply a unitary implementation of $\PKE.\qDec$ to $(K_{i-1},a_i)$ and obtain a message register containing $s_i^{(a)}$.
            \item Compute a control bit $d_i$ that is $0$ if $z_i=\PRG(s_i^{(a)})$ and is $1$ otherwise.
            \item Controlled on $d_i=1$, uncompute $\PKE.\qDec$, and then coherently apply $\PKE.\qDec$ to $b_i$ to obtain a message register containing $s_i^{(b)}$.
            \item Coherently compute
            \begin{align}
                (x_i,s_i,\ct_i)\coloneqq
                \begin{cases}
                    (0,s_i^{(a)},b_i) & \text{if } d_i=0,\\
                    (1,s_i^{(b)},a_i) & \text{if } d_i=1.
                \end{cases}
            \end{align}
            \item Copy $(x_i,s_i,\ct_i)$ to the output registers and uncompute the previous steps.
            Let $K_i$ be the resulting decryption-key register.
        \end{enumerate}
        \item Measure the output registers in the computational basis to obtain $X'\coloneqq(x,s_1,...,s_n,\ct_1,...,\ct_n)$.
        Return $(X',\qtd')$, where $\qtd'\coloneqq K_n$.
    \end{enumerate}
    \item $\qDel(\qtd)\to\cert$: Run $\cert\gets\PKE.\qDel(\qtd)$.
    \item $\Vrfy(\vk,\cert)\to\top/\bot$: Run $\top/\bot\gets\PKE.\Vrfy(\vk,\cert)$.
\end{itemize}

Then, the deletion verification correctness of $\mathsf{TDFSKL}$ immediately follows from that of $\mathsf{PKESKL}$.
We show almost-all-keys inversion correctness and one-wayness.

\begin{lemma}\label{lem:TDFSKL_correct}
    $\mathsf{TDFSKL}$ satisfies the almost-all-keys inversion correctness.
\end{lemma}

\begin{proof}[Proof of \cref{lem:TDFSKL_correct}]
Our goal is to show that 
\begin{align}
    \Pr_{(\ek,\qtd,\vk)\gets\qKG(1^\secp)}\left[\forall X\in\bit^{N}, \Pr_{(X',\qtd')\gets\qInv(\qtd,\Eval(\ek,X))}[X'=X]\ge 1-\negl(\secp)\right] \ge 1-\negl(\secp),
\end{align}
where $N=n+nm+nc$.
Note that $\ek=(\ek_\PKE,\pp,r_1,...,r_n)$, $\qtd=\qdk_\PKE$, and $\vk=\vk_\PKE$, where $(\ek_\PKE,\qdk_\PKE,\vk_\PKE)\gets\PKE.\qKG(1^\secp)$, $\pp\gets\HPRG.\Setup(1^\secp)$, and $r_i\gets\bit^{3m}$ for each $i\in[n]$.
We define the following sets:
\begin{align}
    G_{1,\secp} \coloneqq \left\{ (\ek_\PKE,\qdk_\PKE) : 
    \begin{lgathered}\forall s\in\bit^m,\forall r\in\bit^{\ell}, \\
    \Pr_{(s',\qdk'_\PKE)\gets\PKE.\qDec(\qdk_\PKE,\PKE.\Enc(\ek_\PKE,s;r))}[s'=s]\ge 1-\negl(\secp) 
    \end{lgathered}
    \right\}.
\end{align}
\begin{align}
    G_{2,\secp} \coloneqq \left\{ (r_1,...,r_n): \forall i\in[n],s_i,s_i'\in\bit^{m}, ~ r_i\neq\PRG(s_i)\oplus\PRG(s_i') \right\}.
\end{align}
By the almost-all-keys decryption correctness of $\mathsf{PKESKL}$,
\begin{align}
    \Pr_{(\ek_\PKE,\qdk_\PKE,\vk_\PKE)\gets\PKE.\qKG(1^\secp)} [(\ek_\PKE,\qdk_\PKE)\in G_{1,\secp}] \ge 1-\negl(\secp).
\end{align}
Moreover, we have
\begin{align}
    \Pr_{\forall i\in[n], r_i\gets\bit^{3m}} [(r_1,...,r_n)\in G_{2,\secp}] \ge 1- \frac{n2^{2m}}{2^{3m}}
    = 1-\negl(\secp).
\end{align}
Fix $(\ek_\PKE,\qdk_\PKE)\in G_{1,\secp}$ and $(r_1,...,r_n)\in G_{2,\secp}$.
We show correctness for an arbitrary $X=(x,s_1,...,s_n,\ct_1,...,\ct_n)$, where $x\in\bit^n$, $s_i\in\bit^m$, and $\ct_i\in\bit^c$ for every $i\in[n]$.
Let $\rho_i$ denote the state of the decryption-key register $K_i$ after the first $i$ rounds of $\qInv$, and let $\rho_0=\qdk_\PKE$.

We first consider the $i$th round of inversion on the original key $\rho_0$.
If $x_i=0$, then
\begin{align}
    a_i=\PKE.\Enc(\ek_\PKE,s_i;\HPRG.\Eval(\pp,x,i))
\end{align}
is honestly generated.
Thus, its decryption yields $s_i$ except with negligible probability, in which case the inversion algorithm copies $(x_i,s_i,\ct_i)$ to the output registers.
If $x_i=1$, then $a_i=\ct_i$ may be arbitrary, but
\begin{align}
    z_i=\PRG(s_i)\oplus r_i\neq\PRG(s_i')
\end{align}
holds for every decryption output $s_i'$, by the definition of $G_{2,\secp}$.
Therefore, the decryption of $a_i$ is always coherently uncomputed, after which the algorithm decrypts the honestly generated ciphertext
\begin{align}
    b_i=\PKE.\Enc(\ek_\PKE,s_i;\HPRG.\Eval(\pp,x,i)).
\end{align}
This decryption yields $s_i$ except with negligible probability, and the algorithm copies $(x_i,s_i,\ct_i)$ to the output registers.
Hence, in both cases, the $i$th round inversion is correctly evaluated from $\rho_0$ except with negligible probability.
Moreover, the gentle measurement lemma shows that the inversion algorithm changes the key state $\rho_0$ by at most $\negl(\secp)$ in trace distance.

We now consider the sequential use of the key.
Suppose that $\TD(\rho_{i-1},\rho_0)\le\delta_{i-1}$.
By the data-processing inequality, the probability that the $i$th round inversion is incorrect is at most $\negl(\secp)+\delta_{i-1}$, and the triangle inequality gives
\begin{align}
    \TD(\rho_i,\rho_0)\le\delta_{i-1}+\negl(\secp).
\end{align}
Since $\delta_0=0$, we have $\delta_i\le i\cdot \negl(\secp)$.
A union bound over all $i\in[n]$ therefore gives
\begin{align}
    \Pr[X'\neq X]\le \sum_{i=1}^{n}(\negl(\secp)+\delta_{i-1}) \le \negl(\secp),
\end{align}
where the last equality follows because $n$ is polynomially bounded.

Thus, we have
\begin{align}
    &\Pr_{(\ek,\qtd,\vk)\gets\qKG(1^\secp)}\left[\forall X\in\bit^{N}, \Pr_{(X',\qtd')\gets\qInv(\qtd,\Eval(\ek,X))}[X'=X]\ge 1-\negl(\secp)\right] \\
    &\ge\Pr_{\substack{(\ek_\PKE,\qdk_\PKE,\vk_\PKE)\gets\PKE.\qKG(1^\secp) \\ \forall i\in[n], r_i\gets\bit^{3m}}} [(\ek_\PKE,\qdk_\PKE)\in G_{1,\secp} \land (r_1,...,r_n)\in G_{2,\secp}] \\ 
    &\ge 1-\negl(\secp).
\end{align}

\end{proof}

\begin{lemma}\label{lem:TDF-SKL_security}
    $\mathsf{TDFSKL}$ satisfies one-wayness.
\end{lemma}

\ifnum\submission=0
\begin{proof}[Proof of \cref{lem:TDF-SKL_security}]
We consider the following sequence of hybrids.

\noindent
$\mathsf{Hyb}_0$: This is the original security experiment between the challenger and an adversary $\qA$ that works as follows:
\begin{enumerate}
    \item The challenger generates $(\ek_\PKE,\qdk_\PKE,\vk_\PKE)\gets\PKE.\qKG(1^\secp)$, $\pp\gets\HPRG.\Setup(1^\secp)$ and $r_i\gets\bit^{3m}$ for each $i\in[n]$.
    The challenger sends $(\ek_\PKE,\pp,r_1,...,r_n,\qdk_\PKE)$ to $\qA$.
    \item $\qA$ outputs $\cert$.
    \item If $\PKE.\Vrfy(\vk_\PKE,\cert)=\bot$, the output of the experiment is 0.
    Otherwise, the challenger generates $x\gets\bit^n$, $s_i\gets\bit^{m}$, and $\ct_i\gets\bit^c$ for each $i\in[n]$.
    For each $i\in[n]$,
    \begin{itemize}
        \item if $x_i=0$, let
        \begin{align}
            y_i \coloneqq 
            \begin{pmatrix}
                \PKE.\Enc(\ek_\PKE,s_i;\HPRG.\Eval(\pp,x,i)) \\ 
                \ct_i
            \end{pmatrix}, 
            \quad z_i\coloneqq\PRG(s_i).
        \end{align}
        \item if $x_i=1$, let
        \begin{align}
            y_i \coloneqq
            \begin{pmatrix}
                \ct_i \\
                \PKE.\Enc(\ek_\PKE,s_i;\HPRG.\Eval(\pp,x,i)) 
            \end{pmatrix}, 
            \quad z_i\coloneqq\PRG(s_i) \oplus r_i.
        \end{align}
    \end{itemize}
    The challenger sends $(y_1,...,y_n,z_1,...,z_n)$ to $\qA$.
    \item $\qA$ outputs $X'$.
    If $X'=(x,s_1,...,s_n,\ct_1,...,\ct_n)$, the output of the experiment is 1.
    Otherwise, the output of the experiment is 0.
\end{enumerate}
Our goal is to show that for any QPT adversary $\qA$,
\begin{align}
    \Pr[1\gets\mathsf{Hyb}_0] \le \negl(\secp).
\end{align}

\noindent
$\mathsf{Hyb}_1$: This is identical to $\mathsf{Hyb}_0$ except that $(x,s_1,...,s_n)$ is sampled before $\cA$ outputs $\cert$ and for each $i\in[n]$, $r_i$ is replaced with the XOR of a uniformly random string and the output of $\PRG$.
The experiment works as follows:
\begin{enumerate}
    \item The challenger generates $(\ek_\PKE,\qdk_\PKE,\vk_\PKE)\gets\PKE.\qKG(1^\secp)$, $\pp\gets\HPRG.\Setup(1^\secp)$, and $x\gets\bit^n$.
    For each $i\in[n]$, the challenger generates $s_i^0\gets\bit^{m}$, $s_i^1\gets\bit^{m}$, $u_i\gets\bit^{3m}$, and $r_i\coloneqq\PRG(s_i^{x_i})\oplus u_i$.
    The challenger sends $(\ek_\PKE,\pp,r_1,...,r_n,\qdk_\PKE)$ to $\qA$.
    \item $\qA$ outputs $\cert$.
    \item If $\PKE.\Vrfy(\vk_\PKE,\cert)=\bot$, the output of the experiment is 0.
    Otherwise, the challenger generates $\ct_i\gets\bit^c$ for each $i\in[n]$.
    For each $i\in[n]$,
    \begin{itemize}
        \item if $x_i=0$, let
        \begin{align}
            y_i \coloneqq
            \begin{pmatrix}
                \PKE.\Enc(\ek_\PKE,s^{0}_i;\HPRG.\Eval(\pp,x,i)) \\ 
                \ct_i
            \end{pmatrix}, 
            \quad z_i\coloneqq\PRG(s^0_i).
        \end{align}
        \item if $x_i=1$, let
        \begin{align}
            y_i \coloneqq 
            \begin{pmatrix}
                \ct_i \\
                \PKE.\Enc(\ek_\PKE,s^1_i;\HPRG.\Eval(\pp,x,i)) 
            \end{pmatrix}, 
            \quad z_i\coloneqq\PRG(s^1_i) \oplus r_i.
        \end{align}
    \end{itemize}
    The challenger sends $(y_1,...,y_n,z_1,...,z_n)$ to $\qA$.
    \item $\qA$ outputs $X'$.
    If $X'=(x,s^{x_1}_1,...,s^{x_n}_n,\ct_1,...,\ct_n)$, the output of the experiment is 1.
    Otherwise, the output of the experiment is 0.
\end{enumerate}

\begin{claim}
\label{clm:hyb0}
For any adversary $\qA$,
\begin{align}
    \Pr[1\gets\mathsf{Hyb}_0] = \Pr[1\gets\mathsf{Hyb}_1]
\end{align}
\end{claim}
\begin{proof}
    The claim follows because the distribution of $(x,s^{x_1}_1,...,s^{x_n}_n,r_1,...,r_n)$ in $\mathsf{Hyb}_1$ is identical to the distribution of $(x,s_1,...,s_n,r_1,...,r_n)$ in $\mathsf{Hyb}_0$, and reordering the sampling procedure as in $\mathsf{Hyb}_1$ does not affect the output of the experiment.
\end{proof}

\noindent
$\mathsf{Hyb}_2$: This is identical to $\mathsf{Hyb}_1$ except that $r_i$ is replaced with $\PRG(s_i^0)\oplus\PRG(s_i^1)$ for each $i\in[n]$. 
The experiment works as follows:
\begin{enumerate}
    \item The challenger generates $(\ek_\PKE,\qdk_\PKE,\vk_\PKE)\gets\PKE.\qKG(1^\secp)$, $\pp\gets\HPRG.\Setup(1^\secp)$ and $x\gets\bit^n$.
    For each $i\in[n]$, the challenger generates $s_i^0\gets\bit^{m}$ and $s_i^1\gets\bit^{m}$, and $r_i\coloneqq\PRG(s_i^0)\oplus\PRG(s_i^1)$.
    The challenger sends $(\ek_\PKE,\pp,r_1,...,r_n,\qdk_\PKE)$ to $\qA$.
    \item $\qA$ outputs $\cert$.
    \item If $\PKE.\Vrfy(\vk_\PKE,\cert)=\bot$, the output of the experiment is 0.
    Otherwise, the challenger generates $\ct_i\gets\bit^c$ for each $i\in[n]$.
    For each $i\in[n]$,
    \begin{itemize}
        \item if $x_i=0$, let
        \begin{align}
            y_i \coloneqq
            \begin{pmatrix}
                \PKE.\Enc(\ek_\PKE,s^0_i;\HPRG.\Eval(\pp,x,i)) \\ 
                \ct_i
            \end{pmatrix}, 
            \quad z_i\coloneqq\PRG(s^0_i).
        \end{align}
        \item if $x_i=1$, let
        \begin{align}
            y_i \coloneqq 
            \begin{pmatrix}
                \ct_i \\
                \PKE.\Enc(\ek_\PKE,s^1_i;\HPRG.\Eval(\pp,x,i)) 
            \end{pmatrix}, 
            \quad z_i\coloneqq\PRG(s^1_i) \oplus r_i=\PRG(s_i^0).
        \end{align}
    \end{itemize}
    The challenger sends $(y_1,...,y_n,z_1,...,z_n)$ to $\qA$.
    \item $\qA$ outputs $X'$.
    If $X'=(x,s^{x_1}_1,...,s^{x_n}_n,\ct_1,...,\ct_n)$, the output of the experiment is 1.
    Otherwise, the output of the experiment is 0.
\end{enumerate}
\begin{claim}
\label{clm:hyb1}
    For any QPT adversary $\qA$, 
    \begin{align}
        |\Pr[1\gets\mathsf{Hyb}_1]-\Pr[1\gets\mathsf{Hyb}_2]|\le\negl(\secp).
    \end{align}
\end{claim}
\begin{proof}
The claim follows from the security of $\PRG$.
For each $i\in\{0,...,n\}$, consider the experiment $G_i$ that is identical to $\mathsf{Hyb}_1$ except that for all $j\le i$, $r_j$ is replaced with $\PRG(s_j^0)\oplus\PRG(s_j^1)$.
Then, $G_0$ is identical to $\mathsf{Hyb}_1$ and $G_n$ is identical to $\mathsf{Hyb}_2$.
For the sake of contradiction, we assume that there exists a QPT adversary $\qA$ and a polynomial $p$ such that
\begin{align}
    |\Pr[1\gets\mathsf{Hyb}_1] - \Pr[1\gets\mathsf{Hyb}_2]| \ge \frac{1}{p(\secp)}
\end{align}
holds for infinitely many $\secp$.
By using $\qA$, we construct a QPT adversary $\qB$ that breaks the security of $\PRG$ as follows:
\begin{enumerate}
    \item $\qB$ takes $a$ as input, where $a=\PRG(s)$ for $s\gets\bit^{m}$ or $a\gets\bit^{3m}$.
    \item $\qB$ samples $j\gets[n]$.
    $\qB$ simulates $\mathsf{Hyb}_1$, where $\qB$ generates $(r_1,...,r_n)$ as follows:
    For each $i\in[n]$, $\qB$ generates $s_i^0\gets\bit^{m}$, $s_i^1\gets\bit^{m}$, and $u_i\gets\bit^{3m}$ and lets
    \begin{align}
        r_i \coloneqq
        \begin{cases}
            \PRG(s_i^0)\oplus \PRG(s_i^1) & \text{if } i<j \\ 
            a \oplus \PRG(s_i^{x_i}) & \text{if } i=j \\
            u_i\oplus \PRG(s_i^{x_i}) & \text{if } i>j.
        \end{cases}
    \end{align}
\end{enumerate}
Then, 
\begin{align}
    \left|\Pr_{s\gets\bit^{m}}[1\gets\qB(\PRG(s))] - \Pr_{a\gets\bit^{3m}}[1\gets\qB(a)] \right| 
    &= \frac{1}{n} \left|\sum_{j\in[n]}(\Pr[1\gets G_{j-1}] - \Pr[1\gets G_{j}]) \right| \\
    &= \frac{1}{n} \left| \Pr[1\gets \mathsf{Hyb}_1 ] - \Pr[1\gets\mathsf{Hyb}_2] \right| \\
    &\ge \frac{1}{np(\secp)}
\end{align}
for infinitely many $\secp$ and $\qB$ breaks the security of $\PRG$.
\end{proof}

\noindent
$\mathsf{Hyb}_3$: This is identical to $\mathsf{Hyb}_2$ except that $\ct_i$ is replaced with $\PKE.\Enc(\ek_\PKE,s_i^{x_i\oplus 1})$ for each $i\in[n]$. The experiment works as follows:
\begin{enumerate}
    \item The challenger generates $(\ek_\PKE,\qdk_\PKE,\vk_\PKE)\gets\PKE.\qKG(1^\secp)$, $\pp\gets\HPRG.\Setup(1^\secp)$ and $x\gets\bit^n$.
    For each $i\in[n]$, the challenger generates $s_i^0\gets\bit^{m}$, $s_i^1\gets\bit^{m}$, and $r_i\coloneqq\PRG(s_i^0)\oplus\PRG(s_i^1)$.
    The challenger sends $(\ek_\PKE,\pp,r_1,...,r_n,\qdk_\PKE)$ to $\qA$.
    \item $\qA$ outputs $\cert$.
    \item If $\PKE.\Vrfy(\vk_\PKE,\cert)=\bot$, the output of the experiment is 0.
    Otherwise, for each $i\in[n]$,
    \begin{itemize}
        \item if $x_i=0$, the challenger generates $\ct_i\gets\PKE.\Enc(\ek_\PKE,s_i^1)$ and lets
        \begin{align}
            y_i \coloneqq
            \begin{pmatrix}
                \PKE.\Enc(\ek_\PKE,s_i^0;\HPRG.\Eval(\pp,x,i)) \\ 
                \ct_i
            \end{pmatrix}, 
            \quad z_i\coloneqq\PRG(s_i^0).
        \end{align}
        \item if $x_i=1$, the challenger generates $\ct_i\gets\PKE.\Enc(\ek_\PKE,s_i^0)$ and lets
        \begin{align}
            y_i \coloneqq
            \begin{pmatrix}
                \ct_i \\
                \PKE.\Enc(\ek_\PKE,s_i^1;\HPRG.\Eval(\pp,x,i)) 
            \end{pmatrix}, 
            \quad z_i\coloneqq\PRG(s_i^0).
        \end{align}
    \end{itemize}
    The challenger sends $(y_1,...,y_n,z_1,...,z_n)$ to $\qA$.
    \item $\qA$ outputs $X'$.
    If $X'=(x,s^{x_1}_1,...,s^{x_n}_n,\ct_1,...,\ct_n)$, the output of the experiment is 1.
    Otherwise, the output of the experiment is 0.
\end{enumerate}
\begin{claim}
\label{clm:hyb2}
    For any QPT adversary $\qA$,
    \begin{align}
        |\Pr[1\gets\mathsf{Hyb}_2] - \Pr[1\gets\mathsf{Hyb}_3]| \le \negl(\secp).
    \end{align}
\end{claim}
\begin{proof}
The claim follows from the PRCT security of $\mathsf{PKESKL}$.
For each $i\in\{0,...,n\}$, consider the experiment $H_i$ that is identical to $\mathsf{Hyb}_2$ except that for all $j\le i$, $\ct_j$ is replaced with $\PKE.\Enc(\ek_\PKE,s_j^{x_j\oplus 1})$.
Then, $H_0$ is identical to $\mathsf{Hyb}_2$ and $H_n$ is identical to $\mathsf{Hyb}_3$.
For the sake of contradiction, we assume that there exists a QPT adversary $\qA$ and a polynomial $p$ such that
\begin{align}
    |\Pr[1\gets\mathsf{Hyb}_2] - \Pr[1\gets\mathsf{Hyb}_3]| \ge \frac{1}{p(\secp)}
\end{align}
holds for infinitely many $\secp$.
By using $\qA$, we construct a QPT adversary $\qC$ that breaks the PRCT security of $\mathsf{PKESKL}$. 
$\qC$ behaves during the PRCT security game $\mathsf{Exp}^{\mathsf{PRCT}}_{\mathsf{PKESKL},\qC}(\secp,b)$ as follows:
\begin{enumerate}
    \item The challenger generates $(\ek_\PKE,\qdk_\PKE,\vk_\PKE)\gets\PKE.\qKG(1^\secp)$ and sends $(\ek_\PKE,\qdk_\PKE)$ to $\qC$.
    \item $\qC$ generates $\pp\gets\HPRG.\Setup(1^\secp)$, $x\gets\bit^n$, $s_i^0\gets\bit^{m}$, $s_i^1\gets\bit^{m}$, and $r_i=\PRG(s_i^0)\oplus\PRG(s_i^1)$. 
    $\qC$ runs $\cert\gets\qA(\ek_\PKE,\pp,r_1,...,r_n,\qdk_\PKE)$. 
    $\qC$ samples $j\gets[n]$ and sets $m\coloneqq s_{j}^{x_j\oplus 1}$.
    $\qC$ sends $\cert$ and $m$ to the challenger.
    \item If $\PKE.\Vrfy(\vk_\PKE,\cert)=\bot$, then the output of the experiment is 0.
    Otherwise, the challenger generates $\ct_b^*$, where $\ct^*_0\gets\bit^c$ and $\ct^*_1\gets\PKE.\Enc(\ek_\PKE,m)$ and sends $\ct_b^*$ to $\qC$.
    \item $\qC$ simulates $\mathsf{Hyb}_2$, where $\qC$ generates $(\ct_1,...,\ct_n)$ as follows:
    \begin{align}
        \ct_i \coloneqq 
        \begin{cases}
            \PKE.\Enc(\ek_\PKE,s_i^{x_i\oplus 1}) & \text{if } i<j \\ 
            \ct_b^* & \text{if } i=j \\
            \ct_i\gets\bit^c & \text{if } i>j.
        \end{cases}
    \end{align}
\end{enumerate}
Then,
\begin{align}
    |\Pr [1\gets \mathsf{Exp}^{\mathsf{PRCT}}_{\mathsf{PKESKL},\qC}(\secp,0)] - \Pr[1\gets \mathsf{Exp}^{\mathsf{PRCT}}_{\mathsf{PKESKL},\qC}(\secp,1)] | 
    &= \frac{1}{n} \left|\sum_{j\in[n]}(\Pr[1\gets H_{j-1}] - \Pr[1\gets H_{j}]) \right| \\
    &= \frac{1}{n} \left|(\Pr[1\gets H_{0}] - \Pr[1\gets H_{n}])\right| \\
    &= \frac{1}{n} \left|(\Pr[1\gets \mathsf{Hyb}_2] - \Pr[1\gets\mathsf{Hyb}_3])\right| \\
    &\ge \frac{1}{np(\secp)}
\end{align}
for infinitely many $\secp$ and $\qC$ breaks the PRCT security of $\mathsf{PKESKL}$.
\end{proof}

\noindent
$\mathsf{Hyb}_4$: This is identical to $\mathsf{Hyb}_3$ except that for each $i\in[n]$, $\HPRG.\Eval(\pp,x,i)$ is replaced with a random string. The experiment works as follows:
\begin{enumerate}
    \item The challenger generates $(\ek_\PKE,\qdk_\PKE,\vk_\PKE)\gets\PKE.\qKG(1^\secp)$, $\pp\gets\HPRG.\Setup(1^\secp)$ and $x\gets\bit^n$.
    For $i\in[n]$, the challenger generates $s_i^0\gets\bit^{m}$, $s_i^1\gets\bit^{m}$, and $r_i\coloneqq\PRG(s_i^0)\oplus\PRG(s_i^1)$.
    The challenger sends $(\ek_\PKE,\pp,r_1,...,r_n,\qdk_\PKE)$ to $\qA$.
    \item $\qA$ outputs $\cert$.
    \item If $\PKE.\Vrfy(\vk_\PKE,\cert)=\bot$, the output of the experiment is 0.
    Otherwise, for $i\in[n]$, the challenger generates $v_i\gets\bit^{\ell}$.
    For each $i\in[n]$,
    \begin{itemize}
        \item if $x_i=0$, the challenger generates $\ct_i\gets\PKE.\Enc(\ek_\PKE,s_i^1)$ and lets
        \begin{align}
            y_i \coloneqq
            \begin{pmatrix}
                \PKE.\Enc(\ek_\PKE,s_i^0;v_i) \\ 
                \ct_i
            \end{pmatrix}, 
            \quad z_i\coloneqq\PRG(s_i^0).
        \end{align}
        \item if $x_i=1$, the challenger generates $\ct_i\gets\PKE.\Enc(\ek_\PKE,s_i^0)$ and lets
        \begin{align}
            y_i \coloneqq
            \begin{pmatrix}
                \ct_i \\
                \PKE.\Enc(\ek_\PKE,s_i^1;v_i) 
            \end{pmatrix}, 
            \quad z_i\coloneqq\PRG(s_i^0).
        \end{align}
    \end{itemize}
    The challenger sends $(y_1,...,y_n,z_1,...,z_n)$ to $\qA$.
    \item $\qA$ outputs $X'$.
    If $X'=(x,s^{x_1}_1,...,s^{x_n}_n,\ct_1,...,\ct_n)$, the output of the experiment is 1.
    Otherwise, the output of the experiment is 0.
\end{enumerate}
\begin{claim}
\label{clm:hyb3}
    For any QPT adversary $\qA$,
    \begin{align}
        |\Pr[1\gets\mathsf{Hyb}_3] - \Pr[1\gets\mathsf{Hyb}_4]|\le\negl(\secp).
    \end{align}
\end{claim}
\begin{proof}
For the sake of contradiction, we assume that there exists a QPT adversary $\qA$ and a polynomial $p$ such that
\begin{align}
    |\Pr[1\gets\mathsf{Hyb}_3] - \Pr[1\gets\mathsf{Hyb}_4]|\ge\frac{1}{p(\secp)}
\end{align}
for infinitely many $\secp$.
Since the view of $\qA$ in $\mathsf{Hyb}_4$ is independent of $x$, we have $\Pr[1\gets\mathsf{Hyb}_4]\le 2^{-n}$.

We construct a QPT adversary $\qD$ against the security of the hinting PRG.
\begin{enumerate}
    \item $\qD$ takes $(\pp, a_0^\beta, \{a_{i,b}^\beta\}_{i\in[n],b\in\bit})$ as input, where $\pp\gets\HPRG.\Setup(1^\secp)$, $x\gets\bit^{n}$, $\beta\gets\bit$, $a_0^0\gets\bit^{\ell}$, $a_0^1\coloneqq\HPRG.\Eval(\pp,x,0)$, $a_{i,0}^0\gets\bit^{\ell}$, $a_{i,1}^0\gets\bit^{\ell}$, $a_{i,x_i}^1\coloneqq\HPRG.\Eval(\pp,x,i)$, and $a_{i,x_i\oplus 1}^1\gets\bit^{\ell}$ for each $i\in[n]$.
    \item $\qD$ generates $(\ek_\PKE,\qdk_\PKE,\vk_\PKE)\gets\PKE.\qKG(1^\secp)$ and, for every $i\in[n]$, samples $s_i^0,s_i^1\gets\bit^m$ and sets $r_i\coloneqq\PRG(s_i^0)\oplus\PRG(s_i^1)$, $\ct_i^b\coloneqq\PKE.\Enc(\ek_\PKE,s_i^b;a_{i,b}^\beta)$ for $b\in\bit$, $y_i\coloneqq\begin{pmatrix} \ct_i^0\\ \ct_i^1 \end{pmatrix}$, and $z_i\coloneqq\PRG(s_i^0)$.
    It sends $(\ek_\PKE,\pp,r_1,\ldots,r_n,\qdk_\PKE)$ to $\qA$ and obtains $\cert$.
    If $\PKE.\Vrfy(\vk_\PKE,\cert)=\bot$, it outputs $0$.
    Otherwise, it sends $(y_1,\ldots,y_n,z_1,\ldots,z_n)$ to $\qA$ and obtains
    $X'=(\widetilde{x},\widetilde{s}_1,\ldots,\widetilde{s}_n,\widetilde{\ct}_1,\ldots,\widetilde{\ct}_n)$.
    If $X'$ cannot be parsed in this form, it outputs $0$.
    \item $\qD$ outputs $1$ if $a_0^\beta=\HPRG.\Eval(\pp,\widetilde{x},0)$ and, for every $i\in[n]$, $a_{i,\widetilde{x}_i}^\beta=\HPRG.\Eval(\pp,\widetilde{x},i)$, $\widetilde{s}_i=s_i^{\widetilde{x}_i}$, and $\widetilde{\ct}_i=\ct_i^{\widetilde{x}_i\oplus 1}$.
    Otherwise, it outputs $0$.
\end{enumerate}

When $\beta=1$, the view given to $\qA$ is distributed exactly as in $\mathsf{Hyb}_3$.
Moreover, whenever $\qA$ wins $\mathsf{Hyb}_3$, its output satisfies all the checks above.
Therefore,
\begin{align}
    \Pr[1\gets\mathsf{Hyb}_3]
    \le \Pr[1\gets\qD(\pp,a_0^1,\{a_{i,b}^1\}_{i\in[n],b\in\bit})].
\end{align}
When $\beta=0$, for every fixed $\widetilde{x}\in\bit^n$, the values $a_0^0$ and $\{a_{i,\widetilde{x}_i}^0\}_{i\in[n]}$ are sampled uniformly at random.
Acceptance implies that there exists some $\widetilde{x}\in\bit^n$ satisfying all the HPRG equalities in the check, regardless of how $\qA$ chooses its output.
Taking a union bound over all such $\widetilde{x}$, we obtain
\begin{align}
    \Pr[1\gets\qD(\pp,a_0^0,\{a_{i,b}^0\}_{i\in[n],b\in\bit})]
    \le 2^n\cdot 2^{-(n+1)\ell}=2^{n-(n+1)\ell}=\negl(\secp),
\end{align}
where the last equality uses $\ell=\omega(\log\secp)$, which we can assume without loss of generality by padding the encryption randomness.
Therefore, for infinitely many $\secp$,
\begin{align}
    &\Pr[1\gets\qD(\pp,a_0^1,\{a_{i,b}^1\}_{i\in[n],b\in\bit})]
    -\Pr[1\gets\qD(\pp,a_0^0,\{a_{i,b}^0\}_{i\in[n],b\in\bit})]\\
    &\ge \Pr[1\gets\mathsf{Hyb}_3]-\negl(\secp)\\
    &\ge |\Pr[1\gets\mathsf{Hyb}_3]-\Pr[1\gets\mathsf{Hyb}_4]|
    -\Pr[1\gets\mathsf{Hyb}_4]-\negl(\secp)\\
    &\ge \frac{1}{p(\secp)}-2^{-n}-\negl(\secp),
\end{align}
which is non-negligible and contradicts the security of the hinting PRG.
\end{proof}

In $\mathsf{Hyb}_4$, the distribution of $(y_1,...,y_n,z_1,...,z_n)$ is independent of $x$.
Thus, for any QPT adversary $\qA$, we have
\begin{align}
    \Pr[1\gets\mathsf{Hyb}_4] \le \negl(\secp).
\end{align}
By combining this inequality with \cref{clm:hyb0,clm:hyb1,clm:hyb2,clm:hyb3}, we have
\begin{align}
    \Pr[1\gets\mathsf{Hyb}_0] \le \negl(\secp)
\end{align}
and complete the proof.

\end{proof}
\else
We give the proof in \cref{sec:Proof_TDFSKL}.
\fi

\ifnum\submission=0
\subsection{Trapdoor Functions with Secure Key Leasing with Domain Sampler}
\label{sec:TDFSKL_DS}

In the previous section, we have constructed a TDF-SKL from PRCT secure PKE-SKL (and hinting PRGs).
However, via a slightly different approach, we can construct a TDF-SKL with \textit{domain sampler} from the weaker assumptions, namely IND-CPA secure PKE (and hinting PRGs).
TDF-SKL with domain sampler are defined similarly to TDF-SKL, except that the syntax additionally includes a domain-sampling algorithm that samples a bit string from the domain of the trapdoor function.
As the security for TDF-SKL with domain sampler, we consider \textit{OW-VRA security}, where the adversary obtains the verification key after submitting the deletion certificate and one-wayness is defined with respect to the distribution induced by the domain-sampling algorithm.

The definition of TDF-SKL with domain sampler is as follows:
\begin{definition}[TDF-SKL with domain sampler]\label{def:TDF-SKL_DS}
Let $n$ be a polynomial.
A TDF-SKL with domain sampler is a tuple of algorithms $\mathsf{TDFSKL}=(\mathpzc{KG},\Samp,\Eval,\mathpzc{Inv},\mathpzc{Del},\Vrfy)$.
\begin{itemize}
    \item $\mathpzc{KG}(1^\secp)\to(\ek,\mathpzc{td},\vk)$: The key generation algorithm takes as input a security parameter $1^\secp$ and outputs a classical evaluation key $\ek$, a quantum trapdoor $\mathpzc{td}$, and a classical deletion verification key $\vk$.
    \item $\Samp(\ek)\to x$: The domain sampling algorithm takes as input an evaluation key $\ek$ and outputs $x\in\bit^{n(\secp)}$. 
    \item $\Eval(\ek,x)\to y$: The evaluation algorithm takes an evaluation key $\ek$ and a string $x\in\bit^{n(\secp)}$ as input and outputs $y$. This algorithm is deterministic.
    \item $\mathpzc{Inv}(\mathpzc{td},y)\to(x',\mathpzc{td}')$: The inversion algorithm takes a quantum trapdoor $\mathpzc{td}$ and a string $y$ as input and outputs $x'$ and a resulting trapdoor $\mathpzc{td}'$.
    \item $\mathpzc{Del}(\mathpzc{td})\to\cert$: The trapdoor deletion algorithm takes a quantum trapdoor $\mathpzc{td}$ as input and outputs a classical certificate $\cert$.
    \item $\Vrfy(\vk,\cert)\to\top/\bot$: The deletion verification algorithm takes a deletion verification key $\vk$ and a certificate $\cert$ as input and outputs $\top/\bot$. This algorithm is deterministic.
\end{itemize}
$\mathsf{TDFSKL}$ is required to satisfy the following conditions:
\begin{itemize}
    \item \textbf{Almost-all-keys inversion correctness:} 
    \begin{align}
        \Pr_{(\ek,\mathpzc{td},\vk)\gets\mathpzc{KG}(1^\secp)} \left[ \forall x\in\bit^{n(\secp)}, 
        \Pr_{(x',\mathpzc{td}')\gets\mathpzc{Inv}(\mathpzc{td},\Eval(\ek,x))}[x'=x]\ge 1-\negl(\secp) \right] \ge 1-\negl(\secp).
    \end{align}
    \item \textbf{Deletion verification correctness:} 
    \begin{align}
        \Pr \left[ \Vrfy(\vk,\cert)=\top : \begin{gathered} (\ek,\mathpzc{td},\vk)\gets\mathpzc{KG}(1^\secp) \\ \cert\gets\mathpzc{Del}(\mathpzc{td})
        \end{gathered} \right] \ge 1-\negl(\secp).
    \end{align}
    \item \textbf{OW-VRA security:}
    Consider the experiment $\mathsf{Exp}^{\mathsf{OW-VRA}}_{\mathsf{TDFSKL},\qA}(\secp)$ between the challenger and an adversary $\qA$:
    \begin{enumerate}
        \item The challenger generates $(\ek,\mathpzc{td},\vk)\gets\mathpzc{KG}(1^\secp)$ and sends $(\ek,\mathpzc{td})$ to $\qA$.
        \item $\qA$ sends $\cert$ to the challenger.
        \item If $\Vrfy(\vk,\cert)=\bot$, then the output of the experiment is 0. 
        Otherwise, the challenger generates $x\gets\Samp(\ek), y\coloneqq\Eval(\ek,x)$ and sends $(y,\vk)$ to $\qA$.
        \item $\cA$ sends $x'$ to the challenger.
        \item If $x'=x$, the output of the experiment is 1.
        Otherwise, the output of the experiment is 0.
    \end{enumerate}
    Then, for any QPT adversary $\qA$, 
    \begin{align}
        \Pr[1\gets\mathsf{Exp}^{\mathsf{OW-VRA}}_{\mathsf{TDFSKL},\qA}(\secp)] \le\negl(\secp).
    \end{align}
\end{itemize}
\end{definition}

We show the following theorem.
\begin{theorem}
    If IND-VRA secure PKE-SKL and hinting PRGs exist, then TDF-SKL with domain sampler exist.
\end{theorem}
Our construction is as follows:
\paragraph{Construction.}
Let $\mathsf{PKESKL}=(\PKE.\qKG,\PKE.\Enc,\PKE.\qDec,\PKE.\qDel,\PKE.\Vrfy)$ be an IND-VRA secure PKE-SKL scheme with message length $m=m(\secp)$, randomness length $\ell=\ell(\secp)$, and ciphertext length $c=c(\secp)$.
Let $(\HPRG.\Setup,\HPRG.\Eval)$ be a hinting PRG with input length $n=n(\secp)$ and output length $\ell$.
Let $\PRG$ be a PRG mapping $m$-bit strings to $3m$-bit strings.
We construct a TDF-SKL with domain sampler $\mathsf{TDFSKL}=(\qKG,\Samp,\Eval,\qInv,\qDel,\Vrfy)$ as follows:
\begin{itemize}
    \item $\qKG(1^\secp)\to(\ek,\qtd,\vk)$: 
    \begin{enumerate}
        \item Run $(\ek_\PKE,\qdk_\PKE,\vk_\PKE)\gets\PKE.\qKG(1^\secp)$, $\pp\gets\HPRG.\Setup(1^\secp)$.
        For $i\in[n]$, sample $r_i\gets\bit^{3m}$.
        \item Return $\ek\coloneqq(\ek_\PKE,\pp,r_1,...,r_n)$, $\qtd\coloneqq\qdk_\PKE$, and $\vk\coloneqq\vk_\PKE$.
    \end{enumerate}
    \item $\Samp(\ek)\to X$:
    \begin{enumerate}
        \item Parse $\ek=(\ek_\PKE,\pp,r_1,...,r_n)$.
        \item Sample $x\gets\bit^{n}$. For $i\in[n]$, sample $s_i\gets\bit^{m}$ and $\ct_i\gets\PKE.\Enc(\ek_\PKE,0^{m})$.
        \item Return $X\coloneqq(x,s_1,...,s_n,\ct_1,...,\ct_n)$.
    \end{enumerate}
    \item $\Eval(\ek,X)\to Y$:
    \begin{enumerate}
        \item Parse $\ek=(\ek_\PKE,\pp,r_1,...,r_n)$ and $X=(x,s_1,...,s_n,\ct_1,...,\ct_n)$.
        \item For $i\in[n]$,
        \begin{itemize}
            \item if $x_i=0$, let 
            \begin{align}
                y_i\coloneqq 
                \begin{pmatrix}
                    \PKE.\Enc(\ek_\PKE,s_i;\HPRG.\Eval(\pp,x,i)) \\ 
                    \ct_i
                \end{pmatrix}, 
                \quad z_i\coloneqq\PRG(s_i).
            \end{align}
            \item if $x_i=1$, let
            \begin{align}
                y_i\coloneqq
                \begin{pmatrix}
                    \ct_i \\
                    \PKE.\Enc(\ek_\PKE,s_i;\HPRG.\Eval(\pp,x,i)) 
                \end{pmatrix}, 
                \quad z_i\coloneqq\PRG(s_i) \oplus r_i.
            \end{align}
        \end{itemize}
        \item Return $Y\coloneqq(y_1,...,y_n,z_1,...,z_n)$.
    \end{enumerate}
    \item $\qInv(\qtd,Y)\to (X',\qtd')$:
    \begin{enumerate}
        \item Parse $\qtd=\qdk_\PKE$ and $Y=(y_1,...,y_n,z_1,...,z_n)$, where for each $i\in[n]$, $y_i=\begin{pmatrix} a_i \\ b_i \end{pmatrix}$.
        Initialize the decryption-key register as $K_0\coloneqq\qdk_\PKE$.
        \item For each $i\in[n]$, perform the following operation coherently, using $K_{i-1}$ as the decryption-key register:
        \begin{enumerate}
            \item Apply a unitary implementation of $\PKE.\qDec$ to $(K_{i-1},a_i)$ and obtain a message register containing $s_i^{(a)}$.
            \item Compute $d_i$ that is $0$ if $z_i=\PRG(s_i^{(a)})$ and is $1$ otherwise.
            \item Controlled on $d_i=1$, uncompute $\PKE.\qDec$, and then coherently apply $\PKE.\qDec$ to $b_i$ to obtain a message register containing $s_i^{(b)}$.
            \item Coherently compute
            \begin{align}
                (x_i,s_i,\ct_i)\coloneqq
                \begin{cases}
                    (0,s_i^{(a)},b_i) & \text{if } d_i=0,\\
                    (1,s_i^{(b)},a_i) & \text{if } d_i=1.
                \end{cases}
            \end{align}
            \item Copy $(x_i,s_i,\ct_i)$ to the output registers and uncompute the previous steps.
            Let $K_i$ be the resulting decryption-key register.
        \end{enumerate}
        \item Measure the output registers in the computational basis to obtain $X'\coloneqq(x,s_1,...,s_n,\ct_1,...,\ct_n)$.
        Return $(X',\qtd')$, where $\qtd'\coloneqq K_n$.
    \end{enumerate}
    \item $\qDel(\qtd)\to\cert$: Run $\cert\gets\PKE.\qDel(\qtd)$.
    \item $\Vrfy(\vk,\cert)\to\top/\bot$: Run $\top/\bot\gets\PKE.\Vrfy(\vk,\cert)$.
\end{itemize}

Then, the deletion verification correctness of $\mathsf{TDFSKL}$ immediately follows from that of $\mathsf{PKESKL}$.
Moreover, almost-all-keys inversion correctness follows from the same proof as in \cref{lem:TDFSKL_correct}.
Thus, it suffices to show OW-VRA security.

\begin{lemma}\label{lem:TDF-SKL_DS_security}
    $\mathsf{TDFSKL}$ satisfies OW-VRA security.
\end{lemma}

\ifnum\submission=0
\begin{proof}[Proof of \cref{lem:TDF-SKL_DS_security}]
We consider the following sequence of hybrids.

\noindent
$\mathsf{Hyb}_0$: This is the original security experiment between the challenger and an adversary $\qA$ that works as follows:
\begin{enumerate}
    \item The challenger generates $(\ek_\PKE,\qdk_\PKE,\vk_\PKE)\gets\PKE.\qKG(1^\secp)$, $\pp\gets\HPRG.\Setup(1^\secp)$ and $r_i\gets\bit^{3m}$ for each $i\in[n]$.
    The challenger sends $(\ek_\PKE,\pp,r_1,...,r_n,\qdk_\PKE)$ to $\qA$.
    \item $\qA$ outputs $\cert$.
    \item If $\PKE.\Vrfy(\vk_\PKE,\cert)=\bot$, the output of the experiment is 0.
    Otherwise, the challenger generates $x\gets\bit^n$, $s_i\gets\bit^{m}$, and $\ct_i\gets\PKE.\Enc(\ek_\PKE,0^m)$ for each $i\in[n]$.
    For each $i\in[n]$,
    \begin{itemize}
        \item if $x_i=0$, let
        \begin{align}
            y_i\coloneqq 
            \begin{pmatrix}
                \PKE.\Enc(\ek_\PKE,s_i;\HPRG.\Eval(\pp,x,i)) \\ 
                \ct_i
            \end{pmatrix}, 
            \quad z_i\coloneqq\PRG(s_i).
        \end{align}
        \item if $x_i=1$, let
        \begin{align}
            y_i\coloneqq 
            \begin{pmatrix}
                \ct_i \\
                \PKE.\Enc(\ek_\PKE,s_i;\HPRG.\Eval(\pp,x,i)) 
            \end{pmatrix}, 
            \quad z_i\coloneqq\PRG(s_i) \oplus r_i.
        \end{align}
    \end{itemize}
    The challenger sends $(y_1,...,y_n,z_1,...,z_n,\vk_\PKE)$ to $\qA$.
    \item $\qA$ outputs $X'$.
    If $X'=(x,s_1,...,s_n,\ct_1,...,\ct_n)$, the output of the experiment is 1.
    Otherwise, the output of the experiment is 0.
\end{enumerate}
Our goal is to show that for any QPT adversary $\qA$,
\begin{align}
    \Pr[1\gets\mathsf{Hyb}_0] \le \negl(\secp).
\end{align}

\noindent
$\mathsf{Hyb}_1$: This is identical to $\mathsf{Hyb}_0$ except that $(x,s_1,...,s_n)$ is sampled before $\cA$ outputs $\cert$ and for each $i\in[n]$, $r_i$ is replaced with the XOR of a uniformly random string and the output of $\PRG$.
The experiment works as follows:
\begin{enumerate}
    \item The challenger generates $(\ek_\PKE,\qdk_\PKE,\vk_\PKE)\gets\PKE.\qKG(1^\secp)$, $\pp\gets\HPRG.\Setup(1^\secp)$, and $x\gets\bit^n$.
    For each $i\in[n]$, the challenger generates $s_i^0\gets\bit^{m}$, $s_i^1\gets\bit^{m}$, $u_i\gets\bit^{3m}$, and $r_i\coloneqq\PRG(s_i^{x_i})\oplus u_i$.
    The challenger sends $(\ek_\PKE,\pp,r_1,...,r_n,\qdk_\PKE)$ to $\qA$.
    \item $\qA$ outputs $\cert$.
    \item If $\PKE.\Vrfy(\vk_\PKE,\cert)=\bot$, the output of the experiment is 0.
    Otherwise, the challenger generates $\ct_i\gets\PKE.\Enc(\ek_\PKE,0^m)$ for each $i\in[n]$.
    For each $i\in[n]$,
    \begin{itemize}
        \item if $x_i=0$, let
        \begin{align}
            y_i\coloneqq
            \begin{pmatrix}
                \PKE.\Enc(\ek_\PKE,s^{0}_i;\HPRG.\Eval(\pp,x,i)) \\ 
                \ct_i
            \end{pmatrix}, 
            \quad z_i\coloneqq\PRG(s^0_i).
        \end{align}
        \item if $x_i=1$, let
        \begin{align}
            y_i\coloneqq
            \begin{pmatrix}
                \ct_i \\
                \PKE.\Enc(\ek_\PKE,s^1_i;\HPRG.\Eval(\pp,x,i)) 
            \end{pmatrix}, 
            \quad z_i\coloneqq\PRG(s^1_i) \oplus r_i.
        \end{align}
    \end{itemize}
    The challenger sends $(y_1,...,y_n,z_1,...,z_n,\vk_\PKE)$ to $\qA$.
    \item $\qA$ outputs $X'$.
    If $X'=(x,s^{x_1}_1,...,s^{x_n}_n,\ct_1,...,\ct_n)$, the output of the experiment is 1.
    Otherwise, the output of the experiment is 0.
\end{enumerate}

\begin{claim}
For any adversary $\qA$,
\begin{align}
\label{eq:hyb0_vra}
    \Pr[1\gets\mathsf{Hyb}_0] = \Pr[1\gets\mathsf{Hyb}_1]
\end{align}
\end{claim}
\begin{proof}
    The claim follows because the distribution of $(x,s^{x_1}_1,...,s^{x_n}_n,r_1,...,r_n)$ in $\mathsf{Hyb}_1$ is identical to the distribution of $(x,s_1,...,s_n,r_1,...,r_n)$ in $\mathsf{Hyb}_0$, and reordering the sampling procedure as in $\mathsf{Hyb}_1$ does not affect the output of the experiment.
\end{proof}

\noindent
$\mathsf{Hyb}_2$: This is identical to $\mathsf{Hyb}_1$ except that $r_i$ is replaced with $\PRG(s_i^0)\oplus\PRG(s_i^1)$ for each $i\in[n]$. 
The experiment works as follows:
\begin{enumerate}
    \item The challenger generates $(\ek_\PKE,\qdk_\PKE,\vk_\PKE)\gets\PKE.\qKG(1^\secp)$, $\pp\gets\HPRG.\Setup(1^\secp)$ and $x\gets\bit^n$.
    For each $i\in[n]$, the challenger generates $s_i^0\gets\bit^{m}$ and $s_i^1\gets\bit^{m}$, and $r_i\coloneqq\PRG(s_i^0)\oplus\PRG(s_i^1)$.
    The challenger sends $(\ek_\PKE,\pp,r_1,...,r_n,\qdk_\PKE)$ to $\qA$.
    \item $\qA$ outputs $\cert$.
    \item If $\PKE.\Vrfy(\vk_\PKE,\cert)=\bot$, the output of the experiment is 0.
    Otherwise, the challenger generates $\ct_i\gets\PKE.\Enc(\ek_\PKE,0^m)$ for each $i\in[n]$.
    For each $i\in[n]$,
    \begin{itemize}
        \item if $x_i=0$, let
        \begin{align}
            y_i\coloneqq
            \begin{pmatrix}
                \PKE.\Enc(\ek_\PKE,s^0_i;\HPRG.\Eval(\pp,x,i)) \\ 
                \ct_i
            \end{pmatrix}, 
            \quad z_i\coloneqq\PRG(s^0_i).
        \end{align}
        \item if $x_i=1$, let
        \begin{align}
            y_i\coloneqq
            \begin{pmatrix}
                \ct_i \\
                \PKE.\Enc(\ek_\PKE,s^1_i;\HPRG.\Eval(\pp,x,i)) 
            \end{pmatrix}, 
            \quad z_i\coloneqq\PRG(s^1_i) \oplus r_i=\PRG(s_i^0).
        \end{align}
    \end{itemize}
    The challenger sends $(y_1,...,y_n,z_1,...,z_n,\vk_\PKE)$ to $\qA$.
    \item $\qA$ outputs $X'$.
    If $X'=(x,s^{x_1}_1,...,s^{x_n}_n,\ct_1,...,\ct_n)$, the output of the experiment is 1.
    Otherwise, the output of the experiment is 0.
\end{enumerate}
\begin{claim}
    For any QPT adversary $\qA$, 
    \begin{align}
    \label{eq:hyb1_vra}
        |\Pr[1\gets\mathsf{Hyb}_1]-\Pr[1\gets\mathsf{Hyb}_2]|\le\negl(\secp).
    \end{align}
\end{claim}
\begin{proof}
The claim follows from the security of $\PRG$.
For each $i\in\{0,...,n\}$, consider the experiment $G_i$ that is identical to $\mathsf{Hyb}_1$ except that for all $j\le i$, $r_j$ is replaced with $\PRG(s_j^0)\oplus\PRG(s_j^1)$.
Then, $G_0$ is identical to $\mathsf{Hyb}_1$ and $G_n$ is identical to $\mathsf{Hyb}_2$.
For the sake of contradiction, we assume that there exists a QPT adversary $\qA$ and a polynomial $p$ such that
\begin{align}
    |\Pr[1\gets\mathsf{Hyb}_1] - \Pr[1\gets\mathsf{Hyb}_2]| \ge \frac{1}{p(\secp)}
\end{align}
holds for infinitely many $\secp$.
By using $\qA$, we construct a QPT adversary $\qB$ that breaks the security of $\PRG$ as follows:
\begin{enumerate}
    \item $\qB$ takes $a$ as input, where $a=\PRG(s)$ for $s\gets\bit^{m}$ or $a\gets\bit^{3m}$.
    \item $\qB$ samples $j\gets[n]$.
    $\qB$ simulates $\mathsf{Hyb}_1$, where $\qB$ generates $(r_1,...,r_n)$ as follows:
    For each $i\in[n]$, $\qB$ generates $s_i^0\gets\bit^{m}$, $s_i^1\gets\bit^{m}$, and $u_i\gets\bit^{3m}$ and lets
    \begin{align}
        r_i \coloneqq
        \begin{cases}
            \PRG(s_i^0)\oplus \PRG(s_i^1) & \text{if } i<j \\ 
            a \oplus \PRG(s_i^{x_i}) & \text{if } i=j \\
            u_i\oplus \PRG(s_i^{x_i}) & \text{if } i>j.
        \end{cases}
    \end{align}
\end{enumerate}
Then, 
\begin{align}
    \left|\Pr_{s\gets\bit^{m}}[1\gets\qB(\PRG(s))] - \Pr_{a\gets\bit^{3m}}[1\gets\qB(a)] \right| 
    &= \frac{1}{n} \left|\sum_{j\in[n]}(\Pr[1\gets G_{j-1}] - \Pr[1\gets G_{j}]) \right| \\
    &= \frac{1}{n} \left|\Pr[1\gets \mathsf{Hyb}_1] - \Pr[1\gets \mathsf{Hyb}_2] \right| \\
    &\ge \frac{1}{np(\secp)}
\end{align}
for infinitely many $\secp$ and $\qB$ breaks the security of $\PRG$.
\end{proof}

\noindent
$\mathsf{Hyb}_3$: This is identical to $\mathsf{Hyb}_2$ except that $\ct_i$ is replaced with $\PKE.\Enc(\ek_\PKE,s_i^{x_i\oplus 1})$ for each $i\in[n]$. The experiment works as follows:
\begin{enumerate}
    \item The challenger generates $(\ek_\PKE,\qdk_\PKE,\vk_\PKE)\gets\PKE.\qKG(1^\secp)$, $\pp\gets\HPRG.\Setup(1^\secp)$ and $x\gets\bit^n$.
    For each $i\in[n]$, the challenger generates $s_i^0\gets\bit^{m}$, $s_i^1\gets\bit^{m}$, and $r_i\coloneqq\PRG(s_i^0)\oplus\PRG(s_i^1)$.
    The challenger sends $(\ek_\PKE,\pp,r_1,...,r_n,\qdk_\PKE)$ to $\qA$.
    \item $\qA$ outputs $\cert$.
    \item If $\PKE.\Vrfy(\vk_\PKE,\cert)=\bot$, the output of the experiment is 0.
    Otherwise, for each $i\in[n]$,
    \begin{itemize}
        \item if $x_i=0$, the challenger generates $\ct_i\gets\PKE.\Enc(\ek_\PKE,s_i^1)$ and lets
        \begin{align}
            y_i\coloneqq 
            \begin{pmatrix}
                \PKE.\Enc(\ek_\PKE,s_i^0;\HPRG.\Eval(\pp,x,i)) \\ 
                \ct_i
            \end{pmatrix}, 
            \quad z_i\coloneqq\PRG(s_i^0).
        \end{align}
        \item if $x_i=1$, the challenger generates $\ct_i\gets\PKE.\Enc(\ek_\PKE,s_i^0)$ and lets
        \begin{align}
            y_i\coloneqq 
            \begin{pmatrix}
                \ct_i \\
                \PKE.\Enc(\ek_\PKE,s_i^1;\HPRG.\Eval(\pp,x,i)) 
            \end{pmatrix}, 
            \quad z_i\coloneqq\PRG(s_i^0).
        \end{align}
    \end{itemize}
    The challenger sends $(y_1,...,y_n,z_1,...,z_n,\vk_\PKE)$ to $\qA$.
    \item $\qA$ outputs $X'$.
    If $X'=(x,s^{x_1}_1,...,s^{x_n}_n,\ct_1,...,\ct_n)$, the output of the experiment is 1.
    Otherwise, the output of the experiment is 0.
\end{enumerate}
\begin{claim}
    For any QPT adversary $\qA$,
    \begin{align}
    \label{eq:hyb2_vra}
        |\Pr[1\gets\mathsf{Hyb}_2] - \Pr[1\gets\mathsf{Hyb}_3]| \le \negl(\secp).
    \end{align}
\end{claim}
\begin{proof}
The claim follows from the IND-VRA security of $\mathsf{PKESKL}$.
For each $i\in\{0,...,n\}$, consider the experiment $H_i$ that is identical to $\mathsf{Hyb}_2$ except that for all $j\le i$, $\ct_j$ is replaced with $\PKE.\Enc(\ek_\PKE,s_j^{x_j\oplus 1})$.
Then, $H_0$ is identical to $\mathsf{Hyb}_2$ and $H_n$ is identical to $\mathsf{Hyb}_3$.
For the sake of contradiction, we assume that there exists a QPT adversary $\qA$ and a polynomial $p$ such that
\begin{align}
    |\Pr[1\gets\mathsf{Hyb}_2] - \Pr[1\gets\mathsf{Hyb}_3]| \ge \frac{1}{p(\secp)}
\end{align}
holds for infinitely many $\secp$.
By using $\qA$, we construct a QPT adversary $\qC$ that breaks the IND-VRA security of $\mathsf{PKESKL}$. 
$\qC$ behaves during the IND-VRA security game $\mathsf{Exp}^{\mathsf{IND-VRA}}_{\mathsf{PKESKL},\qC}(\secp,b)$ as follows:
\begin{enumerate}
    \item The challenger generates $(\ek_\PKE,\qdk_\PKE,\vk_\PKE)\gets\PKE.\qKG(1^\secp)$ and sends $(\ek_\PKE,\qdk_\PKE)$ to $\qC$.
    \item $\qC$ generates $\pp\gets\HPRG.\Setup(1^\secp)$, $x\gets\bit^n$, $s_i^0\gets\bit^{m}$, $s_i^1\gets\bit^{m}$, and $r_i=\PRG(s_i^0)\oplus\PRG(s_i^1)$. 
    $\qC$ runs $\cert\gets\qA(\ek_\PKE,\pp,r_1,...,r_n,\qdk_\PKE)$. 
    $\qC$ samples $j\gets[n]$ and sets $m_0\coloneqq 0^m$ and $m_1\coloneqq s_{j}^{x_j\oplus 1}$.
    $\qC$ sends $\cert$ and $(m_0,m_1)$ to the challenger.
    \item If $\PKE.\Vrfy(\vk_\PKE,\cert)=\bot$, then the output of the experiment is 0.
    Otherwise, the challenger generates $\ct^*_b\gets\PKE.\Enc(\ek_\PKE,m_b)$ and sends $(\ct_b^*,\vk_\PKE)$ to $\qC$.
    \item $\qC$ simulates $\mathsf{Hyb}_2$, where $\qC$ generates $(\ct_1,...,\ct_n)$ as follows:
    \begin{align}
        \ct_i \coloneqq
        \begin{cases}
            \PKE.\Enc(\ek_\PKE,s_i^{x_i\oplus 1}) & \text{if } i<j \\ 
            \ct_b^* & \text{if } i=j \\
            \PKE.\Enc(\ek_\PKE,0^m) & \text{if } i>j.
        \end{cases}
    \end{align}
\end{enumerate}
Then, 
\begin{align}
    |\Pr [1\gets \mathsf{Exp}^{\mathsf{IND-VRA}}_{\mathsf{PKESKL},\qC}(\secp,0)] - \Pr[1\gets \mathsf{Exp}^{\mathsf{IND-VRA}}_{\mathsf{PKESKL},\qC}(\secp,1)] | 
    &= \frac{1}{n} \left|\sum_{j\in[n]}(\Pr[1\gets H_{j-1}] - \Pr[1\gets H_{j}]) \right| \\
    &= \frac{1}{n} \left|\Pr[1\gets H_{0}] - \Pr[1\gets H_{n}] \right| \\
    &= \frac{1}{n} \left|\Pr[1\gets \mathsf{Hyb}_2] - \Pr[1\gets \mathsf{Hyb}_3] \right| \\
    &\ge \frac{1}{np(\secp)}
\end{align}
for infinitely many $\secp$ and $\qC$ breaks the IND-VRA security of $\mathsf{PKESKL}$.
\end{proof}

\noindent
$\mathsf{Hyb}_4$: This is identical to $\mathsf{Hyb}_3$ except that for each $i\in[n]$, $\HPRG.\Eval(\pp,x,i)$ is replaced with a random string. The experiment works as follows:
\begin{enumerate}
    \item The challenger generates $(\ek_\PKE,\qdk_\PKE,\vk_\PKE)\gets\PKE.\qKG(1^\secp)$, $\pp\gets\HPRG.\Setup(1^\secp)$ and $x\gets\bit^n$.
    For $i\in[n]$, the challenger generates $s_i^0\gets\bit^{m}$, $s_i^1\gets\bit^{m}$, and $r_i\coloneqq\PRG(s_i^0)\oplus\PRG(s_i^1)$.
    The challenger sends $(\ek_\PKE,\pp,r_1,...,r_n,\qdk_\PKE)$ to $\qA$.
    \item $\qA$ outputs $\cert$.
    \item If $\PKE.\Vrfy(\vk_\PKE,\cert)=\bot$, the output of the experiment is 0.
    Otherwise, for $i\in[n]$, the challenger generates $v_i\gets\bit^{\ell}$.
    For each $i\in[n]$,
    \begin{itemize}
        \item if $x_i=0$, the challenger generates $\ct_i\gets\PKE.\Enc(\ek_\PKE,s_i^1)$ and lets
        \begin{align}
            y_i\coloneqq
            \begin{pmatrix}
                \PKE.\Enc(\ek_\PKE,s_i^0;v_i) \\ 
                \ct_i
            \end{pmatrix}, 
            \quad z_i\coloneqq\PRG(s_i^0).
        \end{align}
        \item if $x_i=1$, the challenger generates $\ct_i\gets\PKE.\Enc(\ek_\PKE,s_i^0)$ and lets
        \begin{align}
            y_i \coloneqq 
            \begin{pmatrix}
                \ct_i \\
                \PKE.\Enc(\ek_\PKE,s_i^1;v_i) 
            \end{pmatrix}, 
            \quad z_i \coloneqq \PRG(s_i^0).
        \end{align}
    \end{itemize}
    The challenger sends $(y_1,...,y_n,z_1,...,z_n,\vk_\PKE)$ to $\qA$.
    \item $\qA$ outputs $X'$.
    If $X'=(x,s^{x_1}_1,...,s^{x_n}_n,\ct_1,...,\ct_n)$, the output of the experiment is 1.
    Otherwise, the output of the experiment is 0.
\end{enumerate}
\begin{claim}
    For any QPT adversary $\qA$,
    \begin{align}
    \label{eq:hyb3_vra}
        |\Pr[1\gets\mathsf{Hyb}_3] - \Pr[1\gets\mathsf{Hyb}_4]|\le\negl(\secp).
    \end{align}
\end{claim}
\begin{proof}
For the sake of contradiction, we assume that there exists a QPT adversary $\qA$ and a polynomial $p$ such that
\begin{align}
    |\Pr[1\gets\mathsf{Hyb}_3] - \Pr[1\gets\mathsf{Hyb}_4]|\ge\frac{1}{p(\secp)}
\end{align}
for infinitely many $\secp$.
Since the view of $\qA$ in $\mathsf{Hyb}_4$ is independent of $x$, we have $\Pr[1\gets\mathsf{Hyb}_4]\le 2^{-n}$.

We construct a QPT adversary $\qD$ against the security of the hinting PRG.
\begin{enumerate}
    \item $\qD$ takes $(\pp, a_0^\beta, \{a_{i,b}^\beta\}_{i\in[n],b\in\bit})$ as input, where $\pp\gets\HPRG.\Setup(1^\secp)$, $x\gets\bit^{n}$, $\beta\gets\bit$, $a_0^0\gets\bit^{\ell}$, $a_0^1\coloneqq\HPRG.\Eval(\pp,x,0)$, $a_{i,0}^0\gets\bit^{\ell}$, $a_{i,1}^0\gets\bit^{\ell}$, $a_{i,x_i}^1\coloneqq\HPRG.\Eval(\pp,x,i)$, and $a_{i,x_i\oplus 1}^1\gets\bit^{\ell}$ for each $i\in[n]$.
    \item $\qD$ generates $(\ek_\PKE,\qdk_\PKE,\vk_\PKE)\gets\PKE.\qKG(1^\secp)$ and, for every $i\in[n]$, samples $s_i^0,s_i^1\gets\bit^m$ and sets $r_i\coloneqq\PRG(s_i^0)\oplus\PRG(s_i^1)$, $\ct_i^b\coloneqq\PKE.\Enc(\ek_\PKE,s_i^b;a_{i,b}^\beta)$ for $b\in\bit$, $y_i\coloneqq\begin{pmatrix}\ct_i^0\\ \ct_i^1\end{pmatrix}$, $z_i\coloneqq\PRG(s_i^0)$.
    It sends $(\ek_\PKE,\pp,r_1,\ldots,r_n,\qdk_\PKE)$ to $\qA$ and obtains $\cert$.
    If $\PKE.\Vrfy(\vk_\PKE,\cert)=\bot$, it outputs $0$.
    Otherwise, it sends $(y_1,\ldots,y_n,z_1,\ldots,z_n,\vk_\PKE)$ to $\qA$ and obtains
    $X'=(\widetilde{x},\widetilde{s}_1,\ldots,\widetilde{s}_n,\widetilde{\ct}_1,\ldots,\widetilde{\ct}_n)$.
    If $X'$ cannot be parsed in this form, it outputs $0$.
    \item $\qD$ outputs $1$ if $a_0^\beta=\HPRG.\Eval(\pp,\widetilde{x},0)$
    and, for every $i\in[n]$, $a_{i,\widetilde{x}_i}^\beta=\HPRG.\Eval(\pp,\widetilde{x},i)$, $\widetilde{s}_i=s_i^{\widetilde{x}_i}$, $\widetilde{\ct}_i=\ct_i^{\widetilde{x}_i\oplus 1}$.
    Otherwise, it outputs $0$.
\end{enumerate}

When $\beta=1$, the view given to $\qA$ is distributed exactly as in $\mathsf{Hyb}_3$.
Moreover, whenever $\qA$ wins $\mathsf{Hyb}_3$, its output satisfies all the checks above.
Therefore,
\begin{align}
    \Pr[1\gets\mathsf{Hyb}_3]
    \le \Pr[1\gets\qD(\pp,a_0^1,\{a_{i,b}^1\}_{i\in[n],b\in\bit})].
\end{align}
When $\beta=0$, for every fixed $\widetilde{x}\in\bit^n$, the values $a_0^0$ and $\{a_{i,\widetilde{x}_i}^0\}_{i\in[n]}$ are sampled uniformly at random.
Thus, we obtain
\begin{align}
    \Pr[1\gets\qD(\pp,a_0^0,\{a_{i,b}^0\}_{i\in[n],b\in\bit})]
    \le 2^n\cdot 2^{-(n+1)\ell}=2^{n-(n+1)\ell}=\negl(\secp),
\end{align}
where the last equality uses $\ell=\omega(\log\secp)$, which we can assume without loss of generality by padding the encryption randomness.
Therefore, for infinitely many $\secp$,
\begin{align}
    &\Pr[1\gets\qD(\pp,a_0^1,\{a_{i,b}^1\}_{i\in[n],b\in\bit})]
    -\Pr[1\gets\qD(\pp,a_0^0,\{a_{i,b}^0\}_{i\in[n],b\in\bit})]\\
    &\ge \Pr[1\gets\mathsf{Hyb}_3]-\negl(\secp)\\
    &\ge |\Pr[1\gets\mathsf{Hyb}_3]-\Pr[1\gets\mathsf{Hyb}_4]|
    -\Pr[1\gets\mathsf{Hyb}_4]-\negl(\secp)\\
    &\ge \frac{1}{p(\secp)}-2^{-n}-\negl(\secp),
\end{align}
which is non-negligible and contradicts the security of the hinting PRG.
\end{proof}

In $\mathsf{Hyb}_4$, the distribution of $(y_1,...,y_n,z_1,...,z_n)$ is independent of $x$.
Thus, for any QPT adversary $\qA$, we have
\begin{align}
    \Pr[1\gets\mathsf{Hyb}_4] \le \negl(\secp).
\end{align}
By combining this inequality with \cref{eq:hyb0_vra,eq:hyb1_vra,eq:hyb2_vra,eq:hyb3_vra}, we have
\begin{align}
    \Pr[1\gets\mathsf{Hyb}_0] \le \negl(\secp)
\end{align}
and complete the proof.

\end{proof}
\else
We give the proof in \cref{sec:Proof_TDFSKL_DS}.
\fi

\fi

\ifnum\submission=0
\subsection{Construction of IND-VRA Secure Robust PKE-SKL}
\label{sec:robustPKESKL}
As an application of TDF-SKL, we construct an IND-VRA secure robust PKE-SKL scheme from TDF-SKL with domain sampler via the quantum Goldreich-Levin lemma with quantum auxiliary input (\cref{lem:quantGL}).

We define robust PKE-SKL schemes as follows:
\begin{definition}[Robust PKE-SKL]
We say that a PKE-SKL scheme $(\mathpzc{KG},\Enc,\mathpzc{Dec},\mathpzc{Del},\Vrfy)$ is robust if 
\begin{align}
    \Pr_{(\ek,\qdk,\vk)\gets\qKG(1^\secp)} \left[ \forall \ct,~ \TD(\qdk,\qdk')\le\negl(\secp) \right] \ge 1-\negl(\secp),
\end{align}
where $(m',\qdk')\gets\qDec(\qdk,\ct)$.
\end{definition}

Our construction of IND-VRA secure robust PKE-SKL schemes is the following:
\paragraph{Construction.}
Let $\mathsf{TDFSKL}=(\mathsf{TDF.}\mathpzc{KG},\mathsf{TDF.}\Samp,\mathsf{TDF.}\Eval,\mathsf{TDF.}\mathpzc{Inv},\mathsf{TDF.}\mathpzc{Del},\mathsf{TDF.}\Vrfy)$ be a TDF-SKL with domain sampler for domain $\bit^{n(\secp)}$, where $n$ is some polynomial.
We construct an IND-VRA secure robust PKE-SKL scheme $\mathsf{PKESKL}=(\mathpzc{KG},\Enc,\mathpzc{Dec},\mathpzc{Del},\Vrfy)$ with message space $\cM=\bit$ as follows:
\begin{itemize}
    \item $\mathpzc{KG}(1^\secp)\to(\ek,\mathpzc{dk},\vk)$: 
    Run $(\ek_\TDF,\qtd_\TDF,\vk_\TDF)\gets\mathsf{TDF.}\mathpzc{KG}(1^\secp)$. 
    Return $\ek\coloneqq\ek_\TDF$, $\qdk\coloneqq(\ek_\TDF,\qtd_\TDF)$, and $\vk\coloneqq\vk_\TDF$.
    \item $\Enc(\ek,m)\to\ct$:
    On input an encryption key $\ek=\ek_\TDF$ and a message $m\in\bit$, generate $r\gets\bit^{n(\secp)}$, $x\gets\mathsf{TDF.}\Samp(\ek_\TDF)$, $y\coloneqq\mathsf{TDF.}\Eval(\ek_\TDF,x)$, and $b\coloneqq(x\cdot r)\oplus m$.
    Return $\ct\coloneqq(y,r,b)$.
    \item $\mathpzc{Dec}(\mathpzc{dk},\ct)\to(m',\mathpzc{dk}')$:
    On input a decryption key $\mathpzc{dk}=(\ek_\TDF,\qtd_\TDF)$ and a ciphertext $\ct$, proceed as follows:
    \begin{enumerate}
        \item If $\ct$ cannot be parsed as $(y,r,b)$ with $r\in\bit^{n(\secp)}$ and $b\in\bit$, output $(\bot,\mathpzc{dk})$.
        \item Run the following coherently: Apply the unitary implementation of $\mathsf{TDF.}\mathpzc{Inv}(\qtd_\TDF,y)$ to obtain a preimage register $x'$.
        Compute a result register $(x'\cdot r)\oplus b$ if $\mathsf{TDF.}\Eval(\ek_\TDF,x')=y$, and $\bot$ otherwise.
        \item Copy the result register into an ancilla register and uncompute the previous step.
        \item Measure the ancilla register in the computational basis to obtain $m'$.
        Let $\qtd_\TDF'$ be the resulting state of the original trapdoor register.
        Output $(m',\mathpzc{dk}')$, where $\mathpzc{dk}'\coloneqq(\ek_\TDF,\qtd_\TDF')$.
    \end{enumerate}
    \item $\mathpzc{Del}(\mathpzc{dk})\to\cert$: On input $\mathpzc{dk}=(\ek_\TDF,\qtd_\TDF)$, run $\cert\gets\mathsf{TDF.}\mathpzc{Del}(\qtd_\TDF)$.
    \item $\Vrfy(\vk,\cert)\to\top/\bot$: On input $\vk=\vk_\TDF$ and $\cert$, run $\top/\bot\gets\TDF.\Vrfy(\vk_\TDF,\cert)$.
\end{itemize}

The decryption correctness and the deletion verification correctness of $\mathsf{PKESKL}$ immediately follow from the almost-all-keys inversion correctness and the deletion verification correctness of $\mathsf{TDFSKL}$.
Thus, it suffices to show the IND-VRA security and the robustness of $\mathsf{PKESKL}$.

\begin{lemma}
\label{lem:robustPKESKL_security}
    $\mathsf{PKESKL}$ is IND-VRA secure.
\end{lemma}

The proof of this lemma is essentially the same as that of Lemma 3.12 in \cite{EC:AKNYY23}.
For completeness, we provide the proof in \cref{sec:robustPKESKL_security}.

\begin{lemma}
\label{lem:robustPKESKL_robust}
     $\mathsf{PKESKL}$ is robust.
\end{lemma}

\begin{proof}[Proof of \cref{lem:robustPKESKL_robust}]
By almost-all-keys inversion correctness, with overwhelming probability over key generation, the inversion succeeds with probability at least $1-\negl(\secp)$ for every input $x$.
Fix such a key tuple $(\ek_\TDF,\qtd_\TDF,\vk_\TDF)$ and any ciphertext $\ct=(y,r,b)$.

\begin{itemize}
    \item[(A)] There exists $x\in\bit^{n(\secp)}$ such that $y=\TDF.\Eval(\ek_\TDF,x)$:
By almost-all-keys inversion correctness of TDF-SKL,
\begin{align}
\label{eq:good_y}
    \Pr_{(x',\qtd'_\TDF)\gets\TDF.\qInv(\qtd_\TDF,y)}[x'=x]\ge 1-\negl(\secp).
\end{align}
Let $U$ be the coherent computation in Step~2 of $\qDec$.
Write $\tau=U(\qtd_\TDF\otimes|0\rangle\langle0|)U^\dagger$, where $|0\rangle$ denotes the initialized ancilla registers.
For $a\in\{0,1,\bot\}$, let $\Pi_a$ be the projector to the measurement result $a$.
By \cref{eq:good_y}, the result $a_*=(x\cdot r)\oplus b$ has probability at least $1-\negl(\secp)$.
Thus, by the gentle measurement lemma (\cref{lem:gentle}),
\begin{align}
    \left\|\tau-\sum_a\Pi_a\tau\Pi_a\right\|_1
    &\le\left\|\tau-\Pi_{a_*}\tau\Pi_{a_*}\right\|_1
        +\sum_{a\ne a_*}\operatorname{Tr}(\Pi_a\tau) \\
    &\le \negl(\secp).
\end{align}
Then,
\begin{align}
    \TD(\qdk,\qdk')
    =\TD(\qtd_\TDF,\qtd_\TDF') \le \frac{1}{2} \left\|\tau-\sum_a\Pi_a\tau\Pi_a\right\|_1 
    \le \negl(\secp).
\end{align}
\item[(B)] For any $x\in\bit^{n(\secp)}$, $y\neq\TDF.\Eval(\ek_\TDF,x)$:
Every preimage fails the verification in $\qDec$, and hence the measurement result is $\bot$ with probability 1.
Hence,
\begin{align}
    \TD(\qdk,\qdk')=\TD(\qtd_\TDF,\qtd'_\TDF)=0.
\end{align}
\end{itemize}
In both cases, we have $\TD(\qdk,\qdk')\le\negl(\secp)$ and we complete the proof.
\end{proof}
\fi

\section{Trapdoor Functions with Copy-Protection}
In this section, we introduce the definition of trapdoor functions with copy protection (TDF-CP) and give the construction of TDF-CP assuming the existence of iO and LWE.
Our construction follows the modular approach proposed by Ananth and Behera \cite{C:AnaBeh24}.
We first show that iO and OWFs imply puncturable TDFs (\cref{sec:puncTDF}).
Next, by combining puncturable TDFs with UPO, we construct TDF-CP in \cref{sec:TDF_copy_protection}.
As an application of TDF-CP, we construct a CPA$^+$ anti-piracy secure \textit{robust} SDE in \cref{sec:robustSDE}.
Robustness of SDE requires that the quantum decryption key remains reusable after decrypting any ciphertext.

\subsection{Construction of Puncturable TDFs}
\label{sec:puncTDF}
We begin by reviewing the definition of puncturable TDFs introduced by \cite{AC:CheWanZho18}. 
While \cite{AC:CheWanZho18} required one-wayness at a single punctured point, in this work, we generalize their definition to require one-wayness at two punctured points.
\begin{definition}[Puncturable TDFs \cite{AC:CheWanZho18}]
    Let $n,m$ be polynomials.
    A puncturable TDF mapping $\bit^{n(\secp)}$ to $\bit^{m(\secp)}$ is a tuple of algorithms $(\KG,\Eval,\Inv,\Puncture,\PuncInv)$ with the following syntax:
    \begin{itemize}
        \item $\KG(1^\secp)\to(\ek,\td)$: The key generation algorithm takes the security parameter $1^\secp$ as input and outputs an evaluation key $\ek$ and a trapdoor $\td$.
        \item $\Eval(\ek,x)\to y$: The evaluation algorithm takes $\ek$ and $x\in\bit^{n(\secp)}$ as input and outputs a string $y\in\bit^{m(\secp)}$. This algorithm is deterministic.
        \item $\Inv(\td,y)\to x'$: The inversion algorithm takes $\td$ and $y\in\bit^{m(\secp)}$ as input and outputs a string $x'$. This algorithm is deterministic.
        \item $\Puncture(\td,(y_1,y_2))\to \td_{y_1,y_2}$: The puncturing algorithm takes $\td$ and a pair of $m(\secp)$-bit strings $(y_1,y_2)$ as input and outputs a punctured trapdoor $\td_{y_1,y_2}$.
        \item $\PuncInv(\td_{y_1,y_2},y)\to x'$: The punctured inversion algorithm takes a punctured trapdoor $\td_{y_1,y_2}$ and a string $y\in\bit^{m(\secp)}$ as input and outputs $x'$. This algorithm is deterministic.
    \end{itemize}
    We require the following properties:
    \begin{itemize}
        \item \textbf{Almost-all-keys inversion correctness:}
        \begin{align}
            \Pr_{(\ek,\td)\gets\KG(1^\secp)} [\forall x\in\bit^{n(\secp)},~ \Inv(\td,\Eval(\ek,x))=x] \ge 1-\negl(\secp).
        \end{align}
        \item \textbf{Puncturing correctness:} For any pair of $m(\secp)$-bit strings $(y_1,y_2)$,
        \begin{align}
            \Pr_{\substack{(\ek,\td)\gets\KG(1^\secp) \\ \td_{y_1,y_2}\gets\Puncture(\td,(y_1,y_2))}} \left[
            \begin{lgathered}
                \forall y\in\bit^{m(\secp)},\\
                \PuncInv(\td_{y_1,y_2},y)=
                \begin{cases}
                \Inv(\td,y) & \text{if } y\notin \{y_1,y_2\}, \\
                \bot & \text{if } y\in \{y_1,y_2\}
                \end{cases}
            \end{lgathered} \right] \ge 1-\negl(\secp).
        \end{align}
        \item \textbf{One-wayness at punctured points:} For any QPT adversary $\qA$,
        \begin{align}
            \Pr\left[x'_1=x^*_1 \lor x'_2=x^*_2:
            \begin{gathered}
                (\ek,\td)\gets\KG(1^\secp) \\
                x^*_1\gets\bit^{n(\secp)}, y^*_1\coloneqq\Eval(\ek,x^*_1) \\
                x^*_2\gets\bit^{n(\secp)}, y^*_2\coloneqq\Eval(\ek,x^*_2) \\ 
                \td_{y^*_1,y^*_2}\gets\Puncture(\td,(y^*_1,y^*_2)) \\
                (x'_1,x'_2)\gets\qA(\ek,\td_{y^*_1,y^*_2},y^*_1,y^*_2)
            \end{gathered}\right] \le \negl(\secp).
        \end{align}
    \end{itemize} 
\end{definition}

\begin{remark}
\cite{AC:CheWanZho18} constructed puncturable TDFs (with one-wayness at a single punctured point) based on correlated-product TDFs \cite{TCC:RosSeg09}.
Moreover, \cite{TCC:RosSeg09} showed that correlated-product TDFs are implied by lossy TDFs \cite{STOC:PeiWat08}, which exist under the LWE assumption.
We believe that it is possible to extend their construction to satisfy one-wayness at two punctured points.
However, in this work, we adopt a more direct approach.
In other words, we construct puncturable TDFs with one-wayness at two punctured points assuming iO and OWFs.
\end{remark}

We construct a puncturable TDF from iO, puncturable PRFs and injective OWFs, following the construction of TDFs by Sahai and Waters \cite{STOC:SahWat14}.
Puncturable PRFs are implied by OWFs \cite{AC:BonWat13,CCS:KPTZ13,PKC:BoyGolIva14}, and injective OWFs can be constructed from iO and OWFs \cite{TCC:BitPanWic16}, and hence our construction of puncturable TDFs is based on iO and OWFs.
\paragraph{Construction.}
Let $f$ be an injective OWF mapping $\secp$ bits to $\ell(\secp)$ bits, equipped with a generation algorithm $\OWF.\Gen$.
Let $(\PPRF.\Gen,\PPRF.\Eval,\PPRF.\Puncture)$ be a puncturable PRF mapping $\ell(\secp)$ bits to $\secp$ bits.
Let $iO$ be an iO.
We construct a puncturable TDF $\mathsf{PTDF}=(\KG,\Eval,\Inv,\Puncture,\PuncInv)$ as follows:
\begin{itemize}
    \item $\KG(1^\secp)\to(\ek,\td)$: Generate $f\gets\OWF.\Gen(1^\secp)$ and $k\gets\PPRF.\Gen(1^\secp)$.
    Let $C$ be the circuit of \cref{fig:circuit_TDF}.
    Run $\hat{C}\gets iO(1^\secp,C)$.
    Return $\ek\coloneqq \hat{C}$, $\td\coloneqq (f,k)$.
    \item $\Eval(\ek,x)\to y$: Parse $\ek=\hat{C}$. Return $y\coloneqq \hat{C}(x)$.
    \item $\Inv(\td,y)\to x'$: Parse $\td=(f,k)$ and $y=(t,Y)$, where $t\in\bit^{\ell(\secp)}$ and $Y\in\bit^\secp$.
    Compute $x'\coloneqq Y\oplus \PPRF.\Eval(k,t)$. If $f(x')=t$, return $x'$. Otherwise, return $\bot$.
    \item $\Puncture(\td,(y_1,y_2))\to \td_{y_1,y_2}$: Parse $\td=(f,k)$, $y_1=(t_1,Y_1)$, and $y_2=(t_2,Y_2)$, where $t_1,t_2\in\bit^{\ell(\secp)}$ and $Y_1,Y_2\in\bit^\secp$.
    For each $j\in\{1,2\}$, compute $x_j\coloneqq Y_j\oplus\PPRF.\Eval(k,t_j)$, and let
    \begin{align}
        S\coloneqq\{t_j:j\in\{1,2\}\text{ and }f(x_j)=t_j\}.
    \end{align}
    Run $k_S\gets\PPRF.\Puncture(k,S)$ and return $\td_{y_1,y_2}\coloneqq(f,k_S,S,y_1,y_2)$.
    \item $\PuncInv(\td_{y_1,y_2},y)\to x'$: Parse $\td_{y_1,y_2}=(f,k_S,S,y_1,y_2)$ and $y=(t,Y)$, where $t\in\bit^{\ell(\secp)}$ and $Y\in\bit^\secp$.
    If $y\in\{y_1,y_2\}$ or $t\in S$, return $\bot$.
    Otherwise, compute $x'\coloneqq Y\oplus \PPRF.\Eval(k_S,t)$. If $f(x')=t$, return $x'$. Otherwise, return $\bot$.
\end{itemize}

\begin{figure}[h]
\centering
\fbox{
    \begin{minipage}{0.6\columnwidth}
        \begin{center} \textbf{Circuit $C$} \end{center}
        \textbf{Constant:} $f$, $k$. \\
        \textbf{Input:} $x\in\bit^\secp$. 
        \begin{enumerate}
            \item Compute $t=f(x)$.
            \item Output $(t,\PPRF.\Eval(k,t)\oplus x)$.
        \end{enumerate}
    \end{minipage}
}
\caption{The description of the circuit $C$}
\label{fig:circuit_TDF}
\end{figure}

\begin{figure}[H]
\centering
\fbox{
        \begin{minipage}{0.6\columnwidth}
            \begin{center} \textbf{Circuit $C^*$} \end{center}
            \textbf{Constant:} $f$, $k_{t^*_1,t^*_2}$, $t^*_1$, $t^*_2$, $y^*_1$, and $y^*_2$. \\
            \textbf{Input:} $x\in\bit^\secp$. 
            \begin{enumerate}
                \item Compute $t=f(x)$.
                \item If $t=t^*_1$, output $y^*_1$.
                \item Else if $t=t^*_2$, output $y^*_2$.
                \item Else, output $(t,\PPRF.\Eval(k_{t^*_1,t^*_2},t)\oplus x)$.
            \end{enumerate}
        \end{minipage}
}
\caption{The description of the circuit $C^*$}
\label{fig:circuit_TDF_punc}
\end{figure}

Almost-all-keys inversion correctness follows from the correctness of $iO$.
The puncturing correctness follows from the injectivity of $f$ and the puncturing correctness of $\mathsf{PPRF}$.
We show the one-wayness at punctured points.
\begin{lemma}
\label{lem:PTDF_OW}
    $\mathsf{PTDF}$ satisfies one-wayness at punctured points.
\end{lemma}

\ifnum\submission=0
\begin{proof}
We consider the following sequence of hybrids:

\noindent
$\mathsf{Hyb}_0$: This is identical to the original security game, except that the challenge images $y^*_1$ and $y^*_2$ are defined as the outputs of the original circuit $C$, rather than those of the obfuscated circuit $\hat{C}$. 
\begin{enumerate}
    \item The challenger generates $f\gets\OWF.\Gen(1^\secp)$, $k\gets\PPRF.\Gen(1^\secp)$ and $\hat{C}\gets iO(1^\secp,C)$.
    \item The challenger generates $x^*_1\gets\bit^\secp$, $t^*_1\coloneqq f(x^*_1)$, $y^*_1\coloneqq C(x^*_1)$, $x^*_2\gets\bit^\secp$, $t^*_2\coloneqq f(x^*_2)$, $y^*_2\coloneqq C(x^*_2)$, and $k_{t^*_1,t^*_2}\gets\PPRF.\Puncture(k,(t^*_1,t^*_2))$.
    Let $\td^*\coloneqq(f,k_{t^*_1,t^*_2},\{t^*_1,t^*_2\},y^*_1,y^*_2)$.
    \item $(x'_1,x'_2)\gets\qA(\hat{C},\td^*,y^*_1,y^*_2)$.
    \item Output 1 if $x'_1=x^*_1$ or $x'_2=x^*_2$.
\end{enumerate}
By the perfect correctness of $iO$, the output distributions in the original security game and $\mathsf{Hyb}_0$ are identical.
Hence, our goal is to show that for any QPT adversary $\qA$,
\begin{align}
    \Pr[1\gets\mathsf{Hyb}_0] \le \negl(\secp).
\end{align}

\noindent
$\mathsf{Hyb}_1$: This is identical to $\mathsf{Hyb}_0$ except that $\qA$ obtains the obfuscation of the circuit $C^*$ (\cref{fig:circuit_TDF_punc}) instead of that of $C$.
\begin{enumerate}
    \item The challenger generates $f\gets\OWF.\Gen(1^\secp)$, $k\gets\PPRF.\Gen(1^\secp)$ and $\hat{C}\gets iO(1^\secp,C)$.  
    \item The challenger generates $x^*_1\gets\bit^\secp$, $t^*_1\coloneqq f(x^*_1)$, $y^*_1\coloneqq C(x^*_1)$, $x^*_2\gets\bit^\secp$, $t^*_2\coloneqq f(x^*_2)$, $y^*_2\coloneqq C(x^*_2)$, $k_{t^*_1,t^*_2}\gets\PPRF.\Puncture(k,(t^*_1,t^*_2))$, and $\hat{C^*}\gets iO(1^\secp,C^*)$.
    Let $\td^*\coloneqq(f,k_{t^*_1,t^*_2},\{t^*_1,t^*_2\},y^*_1,y^*_2)$.
    \item $(x'_1,x'_2)\gets\qA(\hat{C^*},\td^*,y^*_1,y^*_2)$.
    \item Output 1 if $x'_1=x^*_1$ or $x'_2=x^*_2$.
\end{enumerate}

\begin{claim}
For any QPT adversary $\qA$, 
    \begin{align}
        |\Pr[1\gets\mathsf{Hyb}_0]-\Pr[1\gets\mathsf{Hyb}_1]|\le\negl(\secp).
    \end{align}
\end{claim}
\begin{proof}
We show this claim by using the security of $iO$.
Let us consider the QPT adversary $\qB$ against iO that behaves as follows:
\begin{enumerate}
    \item $\qB$ generates $f\gets\OWF.\Gen(1^\secp)$, $k\gets\PPRF.\Gen(1^\secp)$, $x^*_1\gets\bit^\secp$, $t^*_1\coloneqq f(x^*_1)$, $y^*_1\coloneqq C(x^*_1)$, $x^*_2\gets\bit^\secp$, $t^*_2\coloneqq f(x^*_2)$, $y^*_2\coloneqq C(x^*_2)$, and $k_{t^*_1,t^*_2}\gets\PPRF.\Puncture(k,(t^*_1,t^*_2))$.
    It sets $\td^*\coloneqq(f,k_{t^*_1,t^*_2},\{t^*_1,t^*_2\},y^*_1,y^*_2)$.
    $\qB$ sends the circuits $(C,C^*)$ to the challenger. 
    \item The challenger samples $b\gets\bit$ and runs $\hat{C}_b\gets iO(1^\secp,C_b)$, where $C_0\coloneqq C$ and $C_1\coloneqq C^*$.
    The challenger sends $\hat{C}_b$ to $\qB$.
    \item $\qB$ runs $(x'_1,x'_2)\gets\qA(\hat{C}_b,\td^*,y^*_1,y^*_2)$.
    $\qB$ outputs 1 if $x'_1=x^*_1$ or $x'_2=x^*_2$.
\end{enumerate}
Then, 
\begin{align}
    &\Pr[1\gets\mathsf{Hyb}_0] = \Pr[1\gets\qB\mid b=0] \\
    &\Pr[1\gets\mathsf{Hyb}_1] = \Pr[1\gets\qB\mid b=1].
\end{align}
Note that circuits $C$ and $C^*$ have the same functionality with overwhelming probability over the choice of $f\gets\OWF.\Gen(1^\secp)$ because of the injectivity of $f$.
Thus, by the security of iO, we have
\begin{align}
    |\Pr[1\gets\mathsf{Hyb}_0] - \Pr[1\gets\mathsf{Hyb}_1]| \le \negl(\secp).
\end{align}
\end{proof}

\noindent
$\mathsf{Hyb}_2$: This is identical to $\mathsf{Hyb}_1$ except that $y^*_1$ and $y^*_2$ are replaced with $(t^*_1,r_1\oplus x^*_1)$ and $(t^*_2,r_2\oplus x^*_2)$ respectively, where $r_1$ and $r_2$ are random strings.
\begin{enumerate}
    \item The challenger generates $f\gets\OWF.\Gen(1^\secp)$, $k\gets\PPRF.\Gen(1^\secp)$.
    \item The challenger generates $r_1\gets\bit^\secp$, $r_2\gets\bit^\secp$, $x^*_1\gets\bit^\secp$, $t^*_1\coloneqq f(x^*_1)$, $y^*_1\coloneqq (t^*_1,r_1\oplus x^*_1)$, $x^*_2\gets\bit^\secp$, $t^*_2\coloneqq f(x^*_2)$, $y^*_2\coloneqq (t^*_2,r_2\oplus x^*_2)$, $k_{t^*_1,t^*_2}\gets\PPRF.\Puncture(k,(t^*_1,t^*_2))$ and $\hat{C^*}\gets iO(1^\secp,C^*)$.
    Let $\td^*\coloneqq(f,k_{t^*_1,t^*_2},\{t^*_1,t^*_2\},y^*_1,y^*_2)$.
    \item $(x'_1,x'_2)\gets\qA(\hat{C^*},\td^*,y^*_1,y^*_2)$.
    \item Output 1 if $x'_1=x^*_1$ or $x'_2=x^*_2$.
\end{enumerate}
\begin{claim}
For any QPT adversary $\qA$,
    \begin{align}
        |\Pr[1\gets\mathsf{Hyb}_1]-\Pr[1\gets\mathsf{Hyb}_2]|\le\negl(\secp).
    \end{align}
\end{claim}
\begin{proof}
We show this claim by using the security of the puncturable PRF $\PPRF$.
Let us consider the QPT adversary $\qD$ against $\PPRF$ that behaves as follows:
\begin{enumerate}
    \item $\qD$ generates $f\gets\OWF.\Gen(1^\secp)$, $x^*_1\gets\bit^\secp$, $t^*_1\coloneqq f(x^*_1)$, $x^*_2\gets\bit^\secp$, and $t^*_2\coloneqq f(x^*_2)$.
    \item The challenger generates $k\gets\PPRF.\Gen(1^\secp)$ and $k_{t^*_1,t^*_2}\gets\PPRF.\Puncture(k,(t^*_1,t^*_2))$.
    The challenger samples $b\gets\bit$. 
    If $b=0$, the challenger sets $a_1\coloneqq \PPRF.\Eval(k,t^*_1)$ and $a_2\coloneqq \PPRF.\Eval(k,t^*_2)$.
    If $b=1$, the challenger samples $a_1\gets\bit^\secp$ and $a_2\gets\bit^\secp$.
    The challenger sends $(k_{t^*_1,t^*_2},a_1,a_2)$ to $\qD$.
    \item $\qD$ generates $y^*_1\coloneqq (t^*_1,a_1\oplus x^*_1)$, $y^*_2\coloneqq (t^*_2,a_2\oplus x^*_2)$, and $\hat{C^*}\gets iO (1^\secp,C^*)$.
    It sets $\td^*\coloneqq(f,k_{t^*_1,t^*_2},\{t^*_1,t^*_2\},y^*_1,y^*_2)$.
    \item $\qD$ runs $(x'_1,x'_2)\gets\qA(\hat{C^*},\td^*,y^*_1,y^*_2)$ and outputs 1 if $x'_1=x^*_1$ or $x'_2=x^*_2$.
\end{enumerate}
Then
\begin{align}
    &\Pr[1\gets\mathsf{Hyb}_1] = \Pr[1\gets\qD \mid b=0] \\
    &\Pr[1\gets\mathsf{Hyb}_2] = \Pr[1\gets\qD \mid b=1].
\end{align}
By the pseudorandomness at punctured points of puncturable PRFs, we have
\begin{align}
    |\Pr[1\gets\mathsf{Hyb}_1] - \Pr[1\gets\mathsf{Hyb}_2] |\le\negl(\secp).
\end{align}
\end{proof}

\noindent
$\mathsf{Hyb}_3$: This is identical to $\mathsf{Hyb}_2$ except that $y^*_1$ and $y^*_2$ are replaced with $(t^*_1,r_1)$ and $(t^*_2,r_2)$ respectively, where $r_1$ and $r_2$ are sampled uniformly at random.
\begin{enumerate}
    \item The challenger generates $f\gets\OWF.\Gen(1^\secp)$, $k\gets\PPRF.\Gen(1^\secp)$.
    \item The challenger generates $r_1\gets\bit^\secp$, $r_2\gets\bit^\secp$, $x^*_1\gets\bit^\secp$, $t^*_1\coloneqq f(x^*_1)$, $y^*_1\coloneqq (t^*_1,r_1)$, $x^*_2\gets\bit^\secp$, $t^*_2\coloneqq f(x^*_2)$, $y^*_2\coloneqq (t^*_2,r_2)$, $k_{t^*_1,t^*_2}\gets\PPRF.\Puncture(k,(t^*_1,t^*_2))$ and $\hat{C^*}\gets iO(1^\secp,C^*)$.
    Let $\td^*\coloneqq(f,k_{t^*_1,t^*_2},\{t^*_1,t^*_2\},y^*_1,y^*_2)$.
    \item $(x'_1,x'_2)\gets\qA(\hat{C^*},\td^*,y^*_1,y^*_2)$.
    \item Output 1 if $x'_1=x^*_1$ or $x'_2=x^*_2$.
\end{enumerate}
\begin{claim}
    \begin{align}
        \Pr[1\gets\mathsf{Hyb}_2] = \Pr[1\gets\mathsf{Hyb}_3]
    \end{align}
\end{claim}
\begin{proof}
    This claim holds because the distributions of $y^*_1$ and $y^*_2$ are the same in $\mathsf{Hyb}_2$ and $\mathsf{Hyb}_3$. 
\end{proof}

\begin{claim}
For any QPT adversary $\qA$,
    \begin{align}
        \Pr[1\gets\mathsf{Hyb}_3] \le \negl(\secp).
    \end{align}
\end{claim}
\begin{proof}
We show this claim by using the security of the injective OWF $f$.
Let us consider the QPT adversary $\qE$ against $f$ that behaves as follows:
\begin{enumerate}
    \item $\qE$ takes $(f,t)$ as input, where $f\gets\OWF.\Gen(1^\secp)$, $x\gets\bit^\secp$ and $t\coloneqq f(x)$.
    \item $\qE$ samples $b\gets\bit$ and sets $t^*_b\coloneqq t$
    \item $\qE$ generates $k\gets\PPRF.\Gen(1^\secp)$, $x^*_{1-b}\gets\bit^\secp$, $t^*_{1-b}\coloneqq f(x^*_{1-b})$, $k_{t^*_0,t^*_1}\gets\PPRF.\Puncture(k,(t^*_0,t^*_1))$, $r_0\gets\bit^\secp$, $r_1\gets\bit^\secp$, $y^*_0\coloneqq (t^*_0,r_0)$, $y^*_1\coloneqq (t^*_1,r_1)$, and $\hat{C^*}\gets iO(1^\secp,C^*)$.
    It sets $\td^*\coloneqq(f,k_{t^*_0,t^*_1},\{t^*_0,t^*_1\},y^*_0,y^*_1)$.
    \item $\qE$ runs $(x'_0,x'_1)\gets\qA(\hat{C^*},\td^*,y^*_0,y^*_1)$ and outputs $x'_b$.
\end{enumerate}
Then
\begin{align}
    \Pr[1\gets\mathsf{Hyb}_3] 
    &= \Pr[x'_0=x^*_0 \lor x'_1=x^*_1 : (x'_0,x'_1)\gets\qA(\hat{C^*},\td^*,y^*_0,y^*_1)] \\
    &\le \sum_{b\in\bit} \Pr[x'_b=x^*_b : (x'_0,x'_1)\gets\qA(\hat{C^*},\td^*,y^*_0,y^*_1)] \\
    &= \Pr[x\gets\qE(f,t) \mid b=0] + \Pr[x\gets\qE(f,t) \mid b=1] \\ 
    &\le \negl(\secp).
\end{align}
In the last inequality, we used the one-wayness of the injective OWF $f$. 
\end{proof}

Therefore, we finally obtain $\Pr[1\gets\mathsf{Hyb}_0]\le\negl(\secp)$ for any QPT adversary $\qA$ and complete the proof.
\end{proof}
\else
We give the proof in \cref{sec:Proof_PTDF}.
\fi

\subsection{Construction of TDFs with Copy Protection}
\label{sec:TDF_copy_protection}
We define TDFs with copy protection as follows.
\begin{definition}[TDFs with Copy Protection]
    Let $n,m$ be polynomials.
    A TDF with copy protection (TDF-CP) mapping $n(\secp)$ bits to $m(\secp)$ bits is a tuple of algorithms $(\Setup,\qKG,\Eval,\qInv)$ with the following syntax:
    \begin{itemize}
        \item $\Setup(1^\secp)\to(\ek,\td)$: The setup algorithm takes a security parameter $1^\secp$ as input and outputs a classical evaluation key $\ek$ and a classical trapdoor $\td$.
        \item $\qKG(\td)\to\qtd$: The quantum trapdoor generation algorithm takes $\td$ as input and outputs a quantum trapdoor $\qtd$.
        \item $\Eval(\ek,x)\to y$: The evaluation algorithm takes $\ek$ and $x\in\bit^{n(\secp)}$ as input and outputs $y\in\bit^{m(\secp)}$. This algorithm is deterministic.
        \item $\qInv(\qtd,y)\to (x',\qtd')$: The inversion algorithm takes $\qtd$ and $y$ as input and outputs $x'$ and a resulting trapdoor $\qtd'$.
    \end{itemize}
    We require the following properties:
    \begin{itemize}
        \item \textbf{Almost-all-keys correctness:}
        \begin{align}
            \Pr_{\substack{(\ek,\td)\gets\Setup(1^\secp) \\ \qtd\gets\qKG(\td)}} [\forall x\in\bit^{n(\secp)},~ \Pr[x\gets\qInv(\qtd,\Eval(\ek,x))]\ge 1-\negl(\secp)] \ge 1-\negl(\secp).
        \end{align}
        \item \textbf{Search anti-piracy:} Consider the experiment $\mathsf{Exp}_{\qA_{\TDF}}^{\mathsf{SearchAntiPiracy}}(\secp)$ between the challenger and an adversary $\qA_{\TDF}=(\qA,\qB,\qC)$ that works as follows:
        \begin{enumerate}
            \item The challenger runs $(\ek,\td)\gets\Setup(1^\secp)$, $\qtd\gets\qKG(\td)$, $x_\qB\gets\bit^{n(\secp)}$, $x_\qC\gets\bit^{n(\secp)}$ and sends $(\ek,\qtd)$ to $\qA$.
            \item $\qA$ generates a bipartite state $\mathpzc{q}$ over registers $\regR_\qB$ and $\regR_\qC$. $\qA$ sends $\regR_\qB$ to $\qB$ and $\regR_\qC$ to $\qC$.
            \item The challenger generates $y_\qB\coloneqq\Eval(\ek,x_\qB)$ and $y_\qC\coloneqq\Eval(\ek,x_\qC)$ and sends $y_\qB$ and $y_\qC$ to $\qB$ and $\qC$, respectively.
            \item $\qB$ and $\qC$ output $x_\qB'$ and $x_\qC'$, respectively. The challenger outputs 1 if $x_\qB'=x_\qB$ and $x_\qC'=x_\qC$.
        \end{enumerate}
        Then, for any QPT adversary $\qA_{\TDF}=(\qA,\qB,\qC)$,
        \begin{align}
            \Pr[1\gets\mathsf{Exp}_{\qA_{\TDF}}^{\mathsf{SearchAntiPiracy}}(\secp)] \le\negl(\secp).
        \end{align}
    \end{itemize}
\end{definition}

We construct TDF-CP from puncturable TDFs and UPO, following the modular approach of \cite{C:AnaBeh24}.

\paragraph{Construction.}
Let $(\PTDF.\KG,\PTDF.\Eval,\PTDF.\Inv,\PTDF.\Puncture,\PTDF.\PuncInv)$ be a puncturable TDF and let $(\UPO.\qObf,\UPO.\qEval)$ be a UPO.
We construct a TDF-CP $\mathsf{TDFCP}=(\Setup,\qKG,\Eval,\qInv)$ as follows:
\begin{itemize}
    \item $\Setup(1^\secp)\to(\ek,\td)$: Run $(\ek,\td)\gets\PTDF.\KG(1^\secp)$.
    \item $\qKG(\td)\to\qtd$: Run $\qtd\gets\UPO.\qObf(1^\secp,\PTDF.\Inv(\td,\cdot))$.
    \item $\Eval(\ek,x)\to y$: Run $y\gets\PTDF.\Eval(\ek,x)$.
    \item $\qInv(\qtd,y)\to (x,\qtd')$: Run $(x,\qtd')\gets\UPO.\qEval(\qtd,y)$. 
\end{itemize}

\begin{lemma}
\label{lem:TDFCP_correct}
    $\mathsf{TDFCP}$ satisfies almost-all-keys correctness.
\end{lemma}
\begin{proof}[Proof of \cref{lem:TDFCP_correct}]
Define the set of good PTDF keys by
\begin{align}
    G_{\PTDF}
    \coloneqq
    \left\{
        (\ek,\td):
        \forall x\in\bit^{n(\secp)},~
        \PTDF.\Inv(\td,\PTDF.\Eval(\ek,x))=x
    \right\}.
\end{align}
By the almost-all-keys correctness of $\PTDF$,
\begin{align}
    \Pr_{(\ek,\td)\gets\PTDF.\KG(1^\secp)}
    [(\ek,\td)\in G_{\PTDF}]
    \ge 1-\negl(\secp).
\end{align}
For every fixed $(\ek,\td)$, let
\begin{align}
    G_{\UPO}(\td)
    \coloneqq
    \left\{
        \qtd:
        \forall y,~
        \Pr_{\substack{
            (x',\qtd')\gets
            \UPO.\qEval(\qtd,y)
        }}
        [x'=\PTDF.\Inv(\td,y)]
        \ge 1-\negl(\secp)
    \right\}.
\end{align}
By the correctness of $\UPO$, for every fixed $\td$,
\begin{align}
    \Pr_{\qtd\gets
        \UPO.\qObf(1^\secp,\PTDF.\Inv(\td,\cdot))}
    [\qtd\in G_{\UPO}(\td)]
    \ge 1-\negl(\secp).
\end{align}
Therefore, by a union bound,
\begin{align}
    \Pr_{\substack{
        (\ek,\td)\gets\PTDF.\KG(1^\secp)\\
        \qtd\gets
        \UPO.\qObf(1^\secp,\PTDF.\Inv(\td,\cdot))
    }}
    \left[
        (\ek,\td)\in G_{\PTDF}
        \land
        \qtd\in G_{\UPO}(\td)
    \right]
    \ge
    1-\negl(\secp).
\end{align}
Fix $(\ek,\td,\qtd)$ that satisfy $(\ek,\td)\in G_\PTDF$ and $\qtd\in G_\UPO(\td)$.
Then, for every $x\in\bit^{n(\secp)}$,
\begin{align}
    \Pr_{(x',\qtd')\gets\qInv(\qtd,\Eval(\ek,x))}[x'=x]
    &= \Pr_{(x',\qtd')\gets\UPO.\qEval(\qtd,\PTDF.\Eval(\ek,x))}
    [x'=x] \\
    &\ge 1-\negl(\secp),
\end{align}
where the inequality follows from $\PTDF.\Inv(\td,\PTDF.\Eval(\ek,x))=x$.
We finally obtain
\begin{align}
    \Pr_{\substack{
        (\ek,\td)\gets\Setup(1^\secp)\\
        \qtd\gets\qKG(\td)
    }}
    \left[
        \forall x\in\bit^{n(\secp)},~
        \Pr_{\substack{
            (x',\qtd')\gets
            \qInv(
                \qtd,\Eval(\ek,x))
        }}
        [x'=x]
        \ge 1-\negl(\secp)
    \right]
    \ge 1-\negl(\secp).
\end{align}
\end{proof}

\begin{lemma}
\label{lem:TDFCP_security}
    $\mathsf{TDFCP}$ is search anti-piracy secure.
\end{lemma}
\begin{proof}[Proof of \cref{lem:TDFCP_security}]
We show search anti-piracy security by using \cref{lem:modular}.
To this end, we first observe that we obtain a puncturable secure circuit $(\mathsf{Circ},\Puncture)$ from puncturable TDF $\PTDF$ as follows:
\begin{itemize}
    \item $\mathsf{Circ}\coloneqq\{ \PTDF.\Inv(\td,\cdot) \}_{\td}$.
    \item $\Puncture(\PTDF.\Inv(\td,\cdot),(y_1,y_2)):$ Run $\td_{y_1,y_2}\gets\PTDF.\Puncture(\td,(y_1,y_2))$.
    Return $\hat{C}\coloneqq\PTDF.\PuncInv(\td_{y_1,y_2},\cdot)$.
\end{itemize}
Define $\Gen_{\mathsf{Circ}}$ and the family $\{\cD_{\ek}\}_{\ek}$ as follows:
\begin{itemize}
    \item $\Gen_{\mathsf{Circ}}(1^\secp)$ samples $(\ek,\td)\gets\PTDF.\KG(1^\secp)$ and outputs $(\ek,\PTDF.\Inv(\td,\cdot))$.
    \item For every $\ek$ in the support of the first output of $\PTDF.\KG$, the distribution $\cD_{\ek}$ samples $x\gets\bit^{n(\secp)}$ and outputs $y\coloneqq\PTDF.\Eval(\ek,x)$.
\end{itemize}
The one-wayness at punctured points of $\PTDF$ implies that $(\mathsf{Circ},\Puncture)$ satisfies 2-point $n(\secp)$-bit $(\Gen_{\mathsf{Circ}},\{\cD_{\ek}\}_{\ek})$-unpredictability-style puncturing security.
Moreover, for the puncturable TDF constructed above and every fixed $\ek$, the function $\PTDF.\Eval(\ek,\cdot)$ is injective, and hence $\cD_{\ek}$ has min-entropy $n(\secp)=\secp$.
Therefore, the existence of UPO for high-min-entropy product distributions gives $(\cD_{\ek}\times\cD_{\ek})$-UPO security for every $\ek$.
The experiment in \cref{lem:modular} is now identical to the search anti-piracy experiment of $\mathsf{TDFCP}$.
By \cref{lem:modular}, we obtain search anti-piracy security of $\mathsf{TDFCP}$.

\end{proof}

\subsection{Construction of Robust SDE Schemes}
\label{sec:robustSDE}

We define robustness of SDE schemes as follows:
\begin{definition}[Robust SDE Schemes]
    We say that an SDE scheme $\mathsf{SDE}=(\Setup,\qKG,\Enc,\qDec)$ is robust if 
    \begin{align}
        \Pr_{(\ek,\dk)\gets\Setup(1^\secp), \qdk\gets\qKG(\dk)} \left[\forall\ct,~ \TD(\qdk,\qdk')\le\negl(\secp) \right] \ge 1-\negl(\secp),
    \end{align}
    where $(m',\qdk')\gets\qDec(\qdk,\ct)$.
\end{definition}

We construct a CPA$^+$ secure robust SDE scheme from TDF-CP.

\paragraph{Construction.}
Let $(\TDF.\Setup,\TDF.\qKG,\TDF.\Eval,\TDF.\qInv)$ be a TDF-CP.
We construct a CPA$^+$ secure robust SDE scheme $\mathsf{SDE}=(\Setup,\qKG,\Enc,\qDec)$ with the message space $\bit$ as follows:
\begin{itemize}
    \item $\Setup(1^\secp)\to(\pk,\sk)$: Run $(\ek,\td)\gets\TDF.\Setup(1^\secp)$. Output $\pk\coloneqq \ek$ and $\sk\coloneqq (\ek,\td)$.
    \item $\qKG(\sk)\to\qsk$: On input $\sk=(\ek,\td)$, run $\qtd\gets\TDF.\qKG(\td)$. Output $\qsk\coloneqq (\ek,\qtd)$.
    \item $\Enc(\pk,m)\to\ct$: Sample $x\gets\bit^{n(\secp)}$ and $r\gets\bit^{n(\secp)}$. Let $y\coloneqq\TDF.\Eval(\pk,x)$ and $b\coloneqq (x\cdot r)\oplus m$. Output $\ct\coloneqq (y,r,b)$.
    \item $\qDec(\qsk,\ct)\to (m',\qsk')$: On input a decryption key $\qsk=(\ek,\qtd)$ and a ciphertext $\ct$, proceed as follows:
    \begin{enumerate}
        \item If $\ct$ cannot be parsed as $(y,r,b)$ with $r\in\bit^{n(\secp)}$ and $b\in\bit$, output $(\bot,\qsk)$.
        \item Run the following coherently: Run $\TDF.\qInv(\qtd,y)$ to obtain a preimage register containing $x'$.
        Compute a result register containing $(x'\cdot r)\oplus b$ if $\TDF.\Eval(\ek,x')=y$, and $\bot$ otherwise.
        \item Copy the result register in the computational basis into an ancilla register and uncompute the previous step.
        \item Measure the ancilla register in the computational basis to obtain $m'$.
        Let $\qtd'$ be the resulting state of the original trapdoor register.
        Output $(m',\qsk')$, where $\qsk'\coloneqq(\ek,\qtd')$.
    \end{enumerate}
\end{itemize}
The correctness of $\mathsf{SDE}$ follows from the almost-all-keys correctness of TDF-CP.
We can show the robustness of $\mathsf{SDE}$ in essentially the same way as in \cref{lem:robustPKESKL_robust}.
Moreover, the CPA$^+$ security follows as a corollary of Theorem 10.1 of \cite{TCC:KitYam25}, which showed that search anti-piracy secure SDE imply CPA$^+$ anti-piracy secure SDE. Note that TDF-CP are the special cases of search anti-piracy secure SDE.

\ifnum\anonymous=1
\else
{\bf Acknowledgements.}
YS is supported by JST SPRING, Grant Number JPMJSP2110.
\fi

\ifnum\submission=0
\bibliographystyle{alpha} 
\else
\bibliographystyle{splncs04}
\fi
\bibliography{abbrev3,crypto,reference,text}

\appendix
\section{On TDF with Domain Sampler}\label{sec:TDF_sampler}
In \cite{C:HohKopWat20}, the authors make an observation (attributed to Omkant Pandey) that a certain formulation of TDF with domain sampler is equivalent to PKE.

Specifically, they argue as follows. Let $\PKE.(\KG,\Enc,\Dec)$ be a OW-CPA secure PKE scheme with message space $\mathcal{M}$. Consider the following TDF with domain sampler:

\begin{description}
\item[$\TDF.\KG(1^\secp)$:] This is identical to $\PKE.\KG$, where the encryption key is treated as an evaluation key $\ek$ and the decryption key is treated as a trapdoor $\td$.
\item[$\TDF.\Samp(\ek)$:] Choose $m \gets \mathcal{M}$, compute $\ct \gets \PKE.\Enc(\ek,m)$, and output $(\ct,m)$.
\item[$\TDF.\Eval(\ek,(\ct,m))$:] Output $\ct$.
\item[$\TDF.\Inv(\td,\ct)$:] Compute $m \gets \PKE.\Dec(\td,\ct)$ and output $(\ct,m)$.
\end{description}

One-wayness of the above TDF with respect to the sampler $\TDF.\Samp$ clearly follows from the OW-CPA security of $\PKE$. Hence, they argue that a TDF with domain sampler is simply PKE in disguise.

However, we observe that inversion correctness of the above TDF is guaranteed only for $(\ct,m)$ in the support of the domain sampler. For example, let $\ct_m$ be an encryption of $m$, let $m' \ne m$, and suppose we evaluate on $(\ct_m,m')$ and then invert. The result is $(\ct_m,m)$, which is not equal to the original input $(\ct_m,m')$.

This observation clarifies that the weakness of the above construction is not merely that one-wayness is defined only with respect to the domain sampler, but also that correctness is guaranteed only for inputs in the support of the domain sampler. 
In particular, if we regard the input space as $\{0,1\}^n$, where $n$ denotes the bit-length of $(\ct,m)$ sampled by the sampler, then correctness is ensured only on a structured subset of this full input space.

Therefore, even if one relaxes one-wayness so that it is required to hold only with respect to the domain sampler, the above trivial construction from PKE still fails to satisfy  correctness on all inputs in the full domain $\{0,1\}^n$. In this sense, the notion of TDF with domain sampler remains non-trivial.

Our definition of TDF-SKL with domain sampler (\Cref{def:TDF-SKL_DS}) requires correctness on all inputs in $\{0,1\}^n$.  Thus, there appears to be no trivial construction from PKE-SKL.
In fact, correctness on all inputs in $\{0,1\}^n$ is essential for our application to robust PKE-SKL. If correctness were guaranteed only for inputs within the support of the domain sampler, then the resulting construction would fail to achieve robustness.

\ifnum\submission=0

\section{Proof of \cref{lem:robustPKESKL_security}}
\label{sec:robustPKESKL_security}
For the sake of contradiction, assume that $\mathsf{PKESKL}$ is not IND-VRA secure.
Then there exist a QPT adversary $\qA$ and a polynomial $p$ such that
\begin{align}
    \Pr_{b\gets\bit} [b\gets\mathsf{Exp}^{\mathsf{IND-VRA}}_{\mathsf{PKESKL},\qA}(\secp,b)] \ge \frac{1}{2} + \frac{1}{p(\secp)}
\end{align}
for infinitely many $\secp\in\N$.
Let $\Lambda$ be the set of such $\secp$.
We divide $\qA$ into the following two algorithms $\qA_0$ and $\qA_1$:
\begin{itemize}
    \item $\qA_0$: It takes $(\ek,\qdk)$ as input and outputs a certificate $\cert$ and quantum auxiliary information $\mathpzc{aux}$.
    \item $\qA_1$: It takes $(y,r,b)$, $\vk$, and $\mathpzc{aux}$ as input and outputs $b'$. 
\end{itemize}
Then, for all $\secp\in\Lambda$,
\begin{align}
    &\frac{1}{2}+\frac{1}{p(\secp)} \\
    &\le \Pr_{b\gets\bit} [b\gets\mathsf{Exp}^{\mathsf{IND-VRA}}_{\mathsf{PKESKL},\qA}(\secp,b)] \\ 
    &= \Pr_{b\gets\bit} [b\gets\mathsf{Exp}^{\mathsf{IND-VRA}}_{\mathsf{PKESKL},\qA}(\secp,b) \mid \Vrfy(\vk,\cert)=\top] \Pr[\Vrfy(\vk,\cert)=\top] \\
    &\quad + \Pr_{b\gets\bit} [b\gets\mathsf{Exp}^{\mathsf{IND-VRA}}_{\mathsf{PKESKL},\qA}(\secp,b) \mid \Vrfy(\vk,\cert)=\bot] (1-\Pr[\Vrfy(\vk,\cert)=\top]) \\
    &= \Pr_{b\gets\bit} [b\gets\mathsf{Exp}^{\mathsf{IND-VRA}}_{\mathsf{PKESKL},\qA}(\secp,b) \mid \Vrfy(\vk,\cert)=\top] \Pr[\Vrfy(\vk,\cert)=\top] \\
    &\quad + \frac{1}{2}(1-\Pr[\Vrfy(\vk,\cert)=\top]). 
\end{align}
Thus, for all $\secp\in\Lambda$,
\begin{align}
\label{eq:condition}
    \Pr_{b\gets\bit} [b\gets\mathsf{Exp}^{\mathsf{IND-VRA}}_{\mathsf{PKESKL},\qA}(\secp,b) \mid \Vrfy(\vk,\cert)=\top] \ge \frac{1}{2} + \frac{1}{p(\secp)\Pr[\Vrfy(\vk,\cert)=\top]}.
\end{align}
We define the following set $G$:
\begin{align}
    G\coloneqq \left\{ (\mathpzc{aux},\vk,x,y) : 
    \Pr_{\substack{r\gets\bit^{n(\secp)} \\ b\gets\bit \\ b'\coloneqq (x\cdot r)\oplus b}}[b\gets\qA_1(\mathpzc{aux},\vk,y,r,b')] \ge \frac{1}{2} + \frac{1}{2p(\secp)}
    \right\}.
\end{align}
By \cref{eq:condition} and an averaging argument,
\begin{align}
    \Pr [(\mathpzc{aux},\vk,x,y)\in G \mid \Vrfy(\vk,\cert)=\top] \ge \frac{1}{2p(\secp)\Pr[\Vrfy(\vk,\cert)=\top]}
\end{align}
holds for all $\secp\in\Lambda$, where $(\ek,\qdk,\vk)\gets\qKG(1^\secp)$, $(\cert,\mathpzc{aux})\gets\qA_0(\ek,\qdk)$, $x\gets\TDF.\Samp(\ek)$, and $y\coloneqq\TDF.\Eval(\ek,x)$.
Thus,
\begin{align}
\label{eq:vrfy_and_succ}
    \Pr \left[ \Vrfy(\vk,\cert)=\top \land \Pr_{\substack{r\gets\bit^{n(\secp)} \\ b\gets\bit \\ b'\coloneqq (x\cdot r)\oplus b}}[b\gets\qA_1(\mathpzc{aux},\vk,y,r,b')] \ge \frac{1}{2} + \frac{1}{2p(\secp)} \right] \ge \frac{1}{2p(\secp)}
\end{align}
for all $\secp\in\Lambda$, where the outer probability is taken over $(\ek,\qdk,\vk)\gets\qKG(1^\secp)$, $(\cert,\mathpzc{aux})\gets\qA_0(\ek,\qdk)$, $x\gets\TDF.\Samp(\ek)$, and $y\coloneqq\TDF.\Eval(\ek,x)$.


By using $\qA$, we construct a QPT adversary $\qB$ that breaks the OW-VRA security of $\mathsf{TDFSKL}$.
Let $\mathsf{Exp}^{\mathsf{OW-VRA}}_{\mathsf{TDFSKL},\qB}(\secp)$ be the OW-VRA security game between the challenger and the adversary $\qB$ that works as follows:
\begin{enumerate}
    \item The challenger runs $(\ek_\TDF,\qtd_\TDF,\vk_\TDF)\gets\TDF.\qKG(1^\secp)$ and sends $(\ek_\TDF,\qtd_\TDF)$ to $\qB$.
    \item $\qB$ runs $(\cert,\mathpzc{aux})\gets\qA_0(\ek_\TDF,\qtd_\TDF)$ and sends $\cert$ to the challenger.
    \item If $\TDF.\Vrfy(\vk_\TDF,\cert)=\bot$, then the output of the experiment is 0.
    Otherwise, the challenger generates $x\gets\TDF.\Samp(\ek_\TDF)$, $y\coloneqq \TDF.\Eval(\ek_\TDF,x)$ and sends $(y,\vk_\TDF)$ to $\qB$.
    \item $\qB$ does the following:
    \begin{enumerate}
        \item Generate $r\gets\bit^{n(\secp)}$.
        \item Let $\qB^*$ be the algorithm in \cref{fig:algo_B*}.
        \item Run $x'\gets\mathpzc{Ext}([\qB^*],\mathpzc{aux},\vk_\TDF,y)$.
    \end{enumerate}
    \item The challenger outputs 1 if $x'=x$.
\end{enumerate}

\begin{figure}[H]
\centering
\fbox{
        \begin{minipage}{0.6\columnwidth}
            \begin{center} \textbf{Algorithm $\qB^*$} \end{center}
            \textbf{Input:} $\mathpzc{aux}$, $\vk_\TDF$, $y$, $r$.
            \begin{enumerate}
                \item Sample $b\gets\bit$.
                \item Run $\coin'\gets\qA_1(\mathpzc{aux},\vk_\TDF,y,r,b)$.
                \item Output $b\oplus \coin'$.
            \end{enumerate}
        \end{minipage}
}
\caption{The description of the algorithm $\qB^*$}
\label{fig:algo_B*}
\end{figure}

To show that $\qB$ outputs $x$ with high probability, we first observe that $\qB^*$ outputs $x\cdot r$ with high probability.
In fact, by \cref{eq:vrfy_and_succ},
\begin{align}
    & \Pr \left[ \TDF.\Vrfy(\vk_\TDF,\cert)=\top \land \Pr_{r\gets\bit^{n(\secp)}}\left[x\cdot r\gets\qB^*(\mathpzc{aux},\vk_\TDF,y,r) \right] \ge \frac{1}{2}+\frac{1}{2p(\secp)} \right] \\
    &= \Pr \left[ \TDF.\Vrfy(\vk_\TDF,\cert)=\top \land \Pr_{\substack{r\gets\bit^{n(\secp)} \\ b\gets\bit \\ b'\coloneqq (x\cdot r)\oplus b}}[ b\gets\qA_1(\mathpzc{aux},\vk_\TDF,y,r,b')] \ge \frac{1}{2} + \frac{1}{2p(\secp)} \right] \\
    &\ge \frac{1}{2p(\secp)} 
\end{align}
holds for all $\secp\in\Lambda$, where the outer probability is taken over $(\ek_\TDF,\qtd_\TDF,\vk_\TDF)\gets\TDF.\qKG(1^\secp)$, $(\cert,\mathpzc{aux})\gets\qA_0(\ek_\TDF,\qtd_\TDF)$, $x\gets\TDF.\Samp(\ek_\TDF)$, and $y\coloneqq\TDF.\Eval(\ek_\TDF,x)$.
By \cref{lem:quantGL},
\begin{align}
    \Pr \left[ \TDF.\Vrfy(\vk_\TDF,\cert)=\top \land \Pr[ x\gets\mathpzc{Ext}([\qB^*],\mathpzc{aux},\vk_\TDF,y) ] \ge \frac{1}{p(\secp)^2} \right] \ge \frac{1}{2p(\secp)}
\end{align}
for all $\secp\in\Lambda$, where $(\ek_\TDF,\qtd_\TDF,\vk_\TDF)\gets\TDF.\qKG(1^\secp)$, $(\cert,\mathpzc{aux})\gets\qA_0(\ek_\TDF,\qtd_\TDF)$, $x\gets\TDF.\Samp(\ek_\TDF)$, and $y\coloneqq\TDF.\Eval(\ek_\TDF,x)$.
Then,
\begin{align}
    \Pr \left[ \TDF.\Vrfy(\vk_\TDF,\cert)=\top \land x\gets\mathpzc{Ext}([\qB^*],\mathpzc{aux},\vk_\TDF,y) \right] \ge \frac{1}{2p(\secp)^3}
\end{align}
for all $\secp\in\Lambda$, where $(\ek_\TDF,\qtd_\TDF,\vk_\TDF)\gets\TDF.\qKG(1^\secp)$, $(\cert,\mathpzc{aux})\gets\qA_0(\ek_\TDF,\qtd_\TDF)$, $x\gets\TDF.\Samp(\ek_\TDF)$, and $y\coloneqq\TDF.\Eval(\ek_\TDF,x)$.
Therefore $\qB$ breaks the OW-VRA security and we complete the proof. 

\else

\section{Proof of \cref{lem:TDF-SKL_security}}
\label{sec:Proof_TDFSKL}
We consider the following sequence of hybrids.

\noindent
$\mathsf{Hyb}_0$: This is the original security experiment between the challenger and an adversary $\qA$ that works as follows:
\begin{enumerate}
    \item The challenger generates $(\ek_\PKE,\qdk_\PKE,\vk_\PKE)\gets\PKE.\qKG(1^\secp)$, $\pp\gets\HPRG.\Setup(1^\secp)$ and $r_i\gets\bit^{3m}$ for each $i\in[n]$.
    The challenger sends $(\ek_\PKE,\pp,r_1,...,r_n,\qdk_\PKE)$ to $\qA$.
    \item $\qA$ outputs $\cert$.
    \item If $\PKE.\Vrfy(\vk_\PKE,\cert)=\bot$, the output of the experiment is 0.
    Otherwise, the challenger generates $x\gets\bit^n$, $s_i\gets\bit^{m}$, and $\ct_i\gets\bit^c$ for each $i\in[n]$.
    For each $i\in[n]$,
    \begin{itemize}
        \item if $x_i=0$, let
        \begin{align}
            y_i \coloneqq 
            \begin{pmatrix}
                \PKE.\Enc(\ek_\PKE,s_i;\HPRG.\Eval(\pp,x,i)) \\ 
                \ct_i
            \end{pmatrix}, 
            \quad z_i\coloneqq\PRG(s_i).
        \end{align}
        \item if $x_i=1$, let
        \begin{align}
            y_i \coloneqq
            \begin{pmatrix}
                \ct_i \\
                \PKE.\Enc(\ek_\PKE,s_i;\HPRG.\Eval(\pp,x,i)) 
            \end{pmatrix}, 
            \quad z_i\coloneqq\PRG(s_i) \oplus r_i.
        \end{align}
    \end{itemize}
    The challenger sends $(y_1,...,y_n,z_1,...,z_n)$ to $\qA$.
    \item $\qA$ outputs $X'$.
    If $X'=(x,s_1,...,s_n,\ct_1,...,\ct_n)$, the output of the experiment is 1.
    Otherwise, the output of the experiment is 0.
\end{enumerate}
Our goal is to show that for any QPT adversary $\qA$,
\begin{align}
    \Pr[1\gets\mathsf{Hyb}_0] \le \negl(\secp).
\end{align}

\noindent
$\mathsf{Hyb}_1$: This is identical to $\mathsf{Hyb}_0$ except that $(x,s_1,...,s_n)$ is sampled before $\cA$ outputs $\cert$ and for each $i\in[n]$, $r_i$ is replaced with the XOR of a uniformly random string and the output of $\PRG$.
The experiment works as follows:
\begin{enumerate}
    \item The challenger generates $(\ek_\PKE,\qdk_\PKE,\vk_\PKE)\gets\PKE.\qKG(1^\secp)$, $\pp\gets\HPRG.\Setup(1^\secp)$, and $x\gets\bit^n$.
    For each $i\in[n]$, the challenger generates $s_i^0\gets\bit^{m}$, $s_i^1\gets\bit^{m}$, $u_i\gets\bit^{3m}$, and $r_i\coloneqq\PRG(s_i^{x_i})\oplus u_i$.
    The challenger sends $(\ek_\PKE,\pp,r_1,...,r_n,\qdk_\PKE)$ to $\qA$.
    \item $\qA$ outputs $\cert$.
    \item If $\PKE.\Vrfy(\vk_\PKE,\cert)=\bot$, the output of the experiment is 0.
    Otherwise, the challenger generates $\ct_i\gets\bit^c$ for each $i\in[n]$.
    For each $i\in[n]$,
    \begin{itemize}
        \item if $x_i=0$, let
        \begin{align}
            y_i \coloneqq
            \begin{pmatrix}
                \PKE.\Enc(\ek_\PKE,s^{0}_i;\HPRG.\Eval(\pp,x,i)) \\ 
                \ct_i
            \end{pmatrix}, 
            \quad z_i\coloneqq\PRG(s^0_i).
        \end{align}
        \item if $x_i=1$, let
        \begin{align}
            y_i \coloneqq 
            \begin{pmatrix}
                \ct_i \\
                \PKE.\Enc(\ek_\PKE,s^1_i;\HPRG.\Eval(\pp,x,i)) 
            \end{pmatrix}, 
            \quad z_i\coloneqq\PRG(s^1_i) \oplus r_i.
        \end{align}
    \end{itemize}
    The challenger sends $(y_1,...,y_n,z_1,...,z_n)$ to $\qA$.
    \item $\qA$ outputs $X'$.
    If $X'=(x,s^{x_1}_1,...,s^{x_n}_n,\ct_1,...,\ct_n)$, the output of the experiment is 1.
    Otherwise, the output of the experiment is 0.
\end{enumerate}

\begin{claim}
For any adversary $\qA$,
\begin{align}
\label{eq:hyb0}
    \Pr[1\gets\mathsf{Hyb}_0] = \Pr[1\gets\mathsf{Hyb}_1]
\end{align}
\end{claim}
\begin{proof}
    The claim follows because the distribution of $(x,s^{x_1}_1,...,s^{x_n}_n,r_1,...,r_n)$ in $\mathsf{Hyb}_1$ is identical to the distribution of $(x,s_1,...,s_n,r_1,...,r_n)$ in $\mathsf{Hyb}_0$, and reordering the sampling procedure as in $\mathsf{Hyb}_1$ does not affect the output of the experiment.
\end{proof}

\noindent
$\mathsf{Hyb}_2$: This is identical to $\mathsf{Hyb}_1$ except that $r_i$ is replaced with $\PRG(s_i^0)\oplus\PRG(s_i^1)$ for each $i\in[n]$. 
The experiment works as follows:
\begin{enumerate}
    \item The challenger generates $(\ek_\PKE,\qdk_\PKE,\vk_\PKE)\gets\PKE.\qKG(1^\secp)$, $\pp\gets\HPRG.\Setup(1^\secp)$ and $x\gets\bit^n$.
    For each $i\in[n]$, the challenger generates $s_i^0\gets\bit^{m}$ and $s_i^1\gets\bit^{m}$, and $r_i\coloneqq\PRG(s_i^0)\oplus\PRG(s_i^1)$.
    The challenger sends $(\ek_\PKE,\pp,r_1,...,r_n\qdk_\PKE)$ to $\qA$.
    \item $\qA$ outputs $\cert$.
    \item If $\PKE.\Vrfy(\vk_\PKE,\cert)=\bot$, the output of the experiment is 0.
    Otherwise, the challenger generates $\ct_i\gets\bit^c$ for each $i\in[n]$.
    For each $i\in[n]$,
    \begin{itemize}
        \item if $x_i=0$, let
        \begin{align}
            y_i \coloneqq
            \begin{pmatrix}
                \PKE.\Enc(\ek_\PKE,s^0_i;\HPRG.\Eval(\pp,x,i)) \\ 
                \ct_i
            \end{pmatrix}, 
            \quad z_i\coloneqq\PRG(s^0_i).
        \end{align}
        \item if $x_i=1$, let
        \begin{align}
            y_i \coloneqq 
            \begin{pmatrix}
                \ct_i \\
                \PKE.\Enc(\ek_\PKE,s^1_i;\HPRG.\Eval(\pp,x,i)) 
            \end{pmatrix}, 
            \quad z_i\coloneqq\PRG(s^1_i) \oplus r_i=\PRG(s_i^0).
        \end{align}
    \end{itemize}
    The challenger sends $(y_1,...,y_n,z_1,...,z_n)$ to $\qA$.
    \item $\qA$ outputs $X'$.
    If $X'=(x,s^{x_i}_1,...,s^{x_i}_n,\ct_1,...,\ct_n)$, the output of the experiment is 1.
    Otherwise, the output of the experiment is 0.
\end{enumerate}
\begin{claim}
    For any QPT adversary $\qA$, 
    \begin{align}
    \label{eq:hyb1}
        |\Pr[1\gets\mathsf{Hyb}_1]-\Pr[1\gets\mathsf{Hyb}_2]|\le\negl(\secp).
    \end{align}
\end{claim}
\begin{proof}
The claim follows from the security of $\PRG$.
For each $i\in\{0,...,n\}$, consider the experiment $G_i$, that is identical to $\mathsf{Hyb}_1$ except that for all $j\le i$, $r_j$ is replaced with $\PRG(s_j^0)\oplus\PRG(s_j^1)$.
Then, $G_0$ is identical to $\mathsf{Hyb}_1$ and $G_n$ is identical to $\mathsf{Hyb}_2$.
For the sake of contradiction , we assume that there exists a QPT adversary $\qA$ and a polynomial $p$ such that
\begin{align}
    |\Pr[1\gets\mathsf{Hyb}_1] - \Pr[1\gets\mathsf{Hyb}_2]| \ge \frac{1}{p(\secp)}
\end{align}
holds for infinitely many $\secp$.
Then, there exists $i^*\in[n]$ such that
\begin{align}
    |\Pr[1\gets G_{i^*-1}] - \Pr[1\gets G_{i^*}]| \ge \frac{1}{np(\secp)}
\end{align}
for infinitely many $\secp$.
By using $\qA$, we construct a QPT adversary $\qB$ that breaks the security of $\PRG$ as follows:
\begin{enumerate}
    \item $\qB$ takes $a$ as input, where $a=\PRG(s)$ for $s\gets\bit^{m}$ or $a\gets\bit^{3m}$.
    \item $\qB$ sample $j\gets[n]$.
    $\qB$ simulates $\mathsf{Hyb}_1$, where $\qB$ generates $(r_1,...,r_n)$ as follows:
    For each $i\in[n]$, $\qB$ generates $s_i^0\gets\bit^{m}$, $s_i^1\gets\bit^{m}$, and $u_i\gets\bit^{3m}$ and lets
    \begin{align}
        r_i \coloneqq
        \begin{cases}
            \PRG(s_i^0)\oplus \PRG(s_i^1) & \text{if } i<j \\ 
            a \oplus \PRG(s_i^{x_i}) & \text{if } i=j \\
            u_i\oplus \PRG(s_i^{x_i}) & \text{if } i>j.
        \end{cases}
    \end{align}
\end{enumerate}
Consider the case in which $j=i^*$, that occurs with probability $1/n$.
If $a=\PRG(s)$ for $s\gets\bit^{m}$, then $\qB$ simulates $G_{i^*}$.
If $a\gets\bit^{3m}$, then $\qB$ simulates $G_{i^*-1}$. 
Thus, 
\begin{align}
    &\left|\Pr_{s\gets\bit^{m}}[1\gets\qB(\PRG(s))] - \Pr_{a\gets\bit^{3m}}[1\gets\qB(a)] \right| \\
    &= \frac{1}{n} |\Pr[1\gets G_{i^*-1}] - \Pr[1\gets G_{i^*}] | \\
    &\ge \frac{1}{n^2p(\secp)}
\end{align}
for infinitely many $\secp$ and $\qB$ breaks the security of $\PRG$.
\end{proof}

\noindent
$\mathsf{Hyb}_3$: This is identical to $\mathsf{Hyb}_2$ except that $\ct_i$ is replaced with $\PKE.\Enc(\ek_\PKE,s_i^{x_i\oplus 1})$ for each $i\in[n]$. The experiment works as follows:
\begin{enumerate}
    \item The challenger generates $(\ek_\PKE,\qdk_\PKE,\vk_\PKE)\gets\PKE.\qKG(1^\secp)$, $\pp\gets\HPRG.\Setup(1^\secp)$ and $x\gets\bit^n$.
    For each $i\in[n]$, the challenger generates $s_i^0\gets\bit^{m}$, $s_i^1\gets\bit^{m}$, and $r_i\coloneqq\PRG(s_i^0)\oplus\PRG(s_i^1)$.
    The challenger sends $(\ek_\PKE,\pp,r_1,...,r_n,\qdk_\PKE)$ to $\qA$.
    \item $\qA$ outputs $\cert$.
    \item If $\PKE.\Vrfy(\vk_\PKE,\cert)=\bot$, the output of the experiment is 0.
    Otherwise, for each $i\in[n]$,
    \begin{itemize}
        \item if $x_i=0$, the challenger generates $\ct_i\gets\PKE.\Enc(\ek_\PKE,s_i^1)$ and lets
        \begin{align}
            y_i \coloneqq
            \begin{pmatrix}
                \PKE.\Enc(\ek_\PKE,s_i^0;\HPRG.\Eval(\pp,x,i)) \\ 
                \ct_i
            \end{pmatrix}, 
            \quad z_i\coloneqq\PRG(s_i^0).
        \end{align}
        \item if $x_i=1$, the challenger generates $\ct_i\gets\PKE.\Enc(\ek_\PKE,s_i^0)$ and lets
        \begin{align}
            y_i \coloneqq
            \begin{pmatrix}
                \ct_i \\
                \PKE.\Enc(\ek_\PKE,s_i^1;\HPRG.\Eval(\pp,x,i)) 
            \end{pmatrix}, 
            \quad z_i\coloneqq\PRG(s_i^0).
        \end{align}
    \end{itemize}
    The challenger sends $(y_1,...,y_n,z_1,...,z_n)$ to $\qA$.
    \item $\qA$ outputs $X'$.
    If $X'=(x,s^{x_1}_1,...,s^{x_n}_n,\ct_1,...,\ct_n)$, the output of the experiment is 1.
    Otherwise, the output of the experiment is 0.
\end{enumerate}
\begin{claim}
    For any QPT adversary $\qA$,
    \begin{align}
    \label{eq:hyb2}
        |\Pr[1\gets\mathsf{Hyb}_2] - \Pr[1\gets\mathsf{Hyb}_3]| \le \negl(\secp).
    \end{align}
\end{claim}
\begin{proof}
The claim follows from the PRCT security of $\mathsf{PKESKL}$.
For each $i\in\{0,...,n\}$, consider the experiment $H_i$, that is identical to $\mathsf{Hyb}_2$ except that for all $j\le i$, $\ct_j$ is replaced with $\PKE.\Enc(\ek_\PKE,s_i^{x_i\oplus 1})$.
Then, $H_0$ is identical to $\mathsf{Hyb}_2$ and $H_n$ is identical to $\mathsf{Hyb}_3$.
For the sake of contradiction , we assume that there exists a QPT adversary $\qA$ and a polynomial $p$ such that
\begin{align}
    |\Pr[1\gets\mathsf{Hyb}_2] - \Pr[1\gets\mathsf{Hyb}_3]| \ge \frac{1}{p(\secp)}
\end{align}
holds for infinitely many $\secp$.
Then, there exists $i^*\in[n]$ such that
\begin{align}
    |\Pr[1\gets H_{i^*-1}] - \Pr[1\gets H_{i^*}]| \ge \frac{1}{np(\secp)}
\end{align}
for infinitely many $\secp$.
By using $\qA$, we construct a QPT adversary $\qC$ that breaks the PRCT security of $\mathsf{PKESKL}$. 
$\qC$ behaves during the PRCT security game $\mathsf{Exp}^{\mathsf{PRCT}}_{\mathsf{PKESKL},\qC}(\secp,b)$ as follows:
\begin{enumerate}
    \item The challenger generates $(\ek_\PKE,\qdk_\PKE)\gets\PKE.\qKG(1^\secp)$ and sends $(\ek_\PKE,\qdk_\PKE)$ to $\qC$.
    \item $\qC$ generates $\pp\gets\HPRG.\Setup(1^\secp)$, $x\gets\bit^n$, $s_i^0\gets\bit^{m}$, $s_i^1\gets\bit^{m}$, and $r_i=\PRG(s_i^0)\oplus\PRG(s_i^1)$. 
    $\qC$ runs $\cert\gets\qA(\ek_\PKE,\pp,r_1,...,r_n,\qdk_\PKE)$. 
    $\qC$ samples $j\gets[n]$ and sets $m\coloneqq s_{j}^{x_j\oplus 1}$.
    $\qC$ sends $\cert$ and $m$ to the challenger.
    \item If $\PKE.\Vrfy(\vk_\PKE,\cert)=\bot$, then the output of the experiment is 0.
    Otherwise, the challenger generates $\ct_b^*$, where $\ct_0\gets\bit^c$ and $\ct^*_1\gets\PKE.\Enc(\ek_\PKE,m)$ and sends $\ct_b^*$ to $\qC$.
    \item $\qC$ simulates $\mathsf{Hyb}_2$, where $\qC$ generates $(\ct_1,...,\ct_n)$ as follows:
    \begin{align}
        \ct_i \coloneqq 
        \begin{cases}
            \PKE.\Enc(\ek_\PKE,s_i^{x_i\oplus 1}) & \text{if } i<j \\ 
            \ct_b^* & \text{if } i=j \\
            \ct_i\gets\bit^c & \text{if } i>j.
        \end{cases}
    \end{align}
\end{enumerate}
Consider the case in which $j=i^*$, that occurs with probability $1/n$.
Then 
\begin{align}
    \Pr[1\gets H_{i^*-1}] &= \Pr[1\gets \mathsf{Exp}^{\mathsf{PRCT}}_{\mathsf{PKESKL},\qC}(\secp,0)] \\
    \Pr[1\gets H_{i^*}] &= \Pr[1\gets \mathsf{Exp}^{\mathsf{PRCT}}_{\mathsf{PKESKL},\qC}(\secp,1)].
\end{align}
Thus, 
\begin{align}
    &|\Pr [1\gets \mathsf{Exp}^{\mathsf{PRCT}}_{\mathsf{PKESKL},\qC}(\secp,0)] - \Pr[1\gets \mathsf{Exp}^{\mathsf{PRCT}}_{\mathsf{PKESKL},\qC}(\secp,1)] | \\
    &= \frac{1}{n} |\Pr[1\gets H_{i^*-1}] - \Pr[1\gets H_{i^*}] | \\
    &\ge \frac{1}{n^2p(\secp)}
\end{align}
for infinitely many $\secp$ and $\qC$ breaks the PRCT security of $\mathsf{PKESKL}$.
\end{proof}

\noindent
$\mathsf{Hyb}_4$: This is identical to $\mathsf{Hyb}_3$ except that for each $i\in[n]$, $\HPRG.\Eval(\pp,x,i)$ is replaced with a random string. The experiment works as follows:
\begin{enumerate}
    \item The challenger generates $(\ek_\PKE,\qdk_\PKE,\vk_\PKE)\gets\PKE.\qKG(1^\secp)$, $\pp\gets\HPRG.\Setup(1^\secp)$ and $x\gets\bit^n$.
    For $i\in[n]$, the challenger generates $s_i^0\gets\bit^{m}$, $s_i^1\gets\bit^{m}$, and $r_i\coloneqq\PRG(s_i^0)\oplus\PRG(s_i^1)$.
    The challenger sends $(\ek_\PKE,\pp,r_1,...,r_n,\qdk_\PKE)$ to $\qA$.
    \item $\qA$ outputs $\cert$.
    \item If $\PKE.\Vrfy(\vk_\PKE,\cert)=\bot$, the output of the experiment is 0.
    Otherwise, For $i\in[n]$, the challenger generates $v_i\gets\bit^{\ell}$.
    For each $i\in[n]$,
    \begin{itemize}
        \item if $x_i=0$, the challenger generates $\ct_i\gets\PKE.\Enc(\ek_\PKE,s_i^1)$ and lets
        \begin{align}
            y_i \coloneqq
            \begin{pmatrix}
                \PKE.\Enc(\ek_\PKE,s_i^0;v_i) \\ 
                \ct_i
            \end{pmatrix}, 
            \quad z_i\coloneqq\PRG(s_i^0).
        \end{align}
        \item if $x_i=1$, the challenger generates $\ct_i\gets\PKE.\Enc(\ek_\PKE,s_i^0)$ and lets
        \begin{align}
            y_i \coloneqq
            \begin{pmatrix}
                \ct_i \\
                \PKE.\Enc(\ek_\PKE,s_i^1;v_i) 
            \end{pmatrix}, 
            \quad z_i\coloneqq\PRG(s_i^0).
        \end{align}
    \end{itemize}
    The challenger sends $(y_1,...,y_n,z_1,...,z_n)$ to $\qA$.
    \item $\qA$ outputs $X'$.
    If $X'=(x,s^{x_1}_1,...,s^{x_n}_n,\ct_1,...,\ct_n)$, the output of the experiment is 1.
    Otherwise, the output of the experiment is 0.
\end{enumerate}
\begin{claim}
    For any QPT adversary $\qA$,
    \begin{align}
    \label{eq:hyb3}
        |\Pr[1\gets\mathsf{Hyb}_3] - \Pr[1\gets\mathsf{Hyb}_4]|\le\negl(\secp).
    \end{align}
\end{claim}
\begin{proof}
For the sake of contradiction, we assume that there exists a QPT adversary $\qA$ and a polynomial $p$ such that
\begin{align}
    |\Pr[1\gets\mathsf{Hyb}_3] - \Pr[1\gets\mathsf{Hyb}_4]|\ge\frac{1}{p(\secp)}
\end{align}
for infinitely many $\secp$.
Since the view of $\qA$ in $\mathsf{Hyb}_4$ is independent of $x$, we have $\Pr[1\gets\mathsf{Hyb}_4]\le 2^{-n}$.

We construct a QPT adversary $\qD$ against the security of the hinting PRG.
\begin{enumerate}
    \item $\qD$ takes $(\pp, a_0^\beta, \{a_{i,b}^\beta\}_{i\in[n],b\in\bit})$ as input, where $\pp\gets\HPRG.\Setup(1^\secp)$, $x\gets\bit^{n}$, $\beta\gets\bit$, $a_0^0\gets\bit^{\ell}$, $a_0^1\coloneqq\HPRG.\Eval(\pp,x,0)$, $a_{i,0}^0\gets\bit^{\ell}$, $a_{i,1}^0\gets\bit^{\ell}$, $a_{i,x_i}^1\coloneqq\HPRG.\Eval(\pp,x,i)$, and $a_{i,x_i\oplus 1}^1\gets\bit^{\ell}$ for each $i\in[n]$.
    \item $\qD$ generates $(\ek_\PKE,\qdk_\PKE,\vk_\PKE)\gets\PKE.\qKG(1^\secp)$ and, for every $i\in[n]$, samples $s_i^0,s_i^1\gets\bit^m$ and sets
    \begin{align}
        r_i&\coloneqq\PRG(s_i^0)\oplus\PRG(s_i^1),\\
        \ct_i^b&\coloneqq\PKE.\Enc(\ek_\PKE,s_i^b;a_{i,b}^\beta)\quad\text{for }b\in\bit,\\
        y_i&\coloneqq
        \begin{pmatrix}
            \ct_i^0\\
            \ct_i^1
        \end{pmatrix},
        \qquad z_i\coloneqq\PRG(s_i^0).
    \end{align}
    It sends $(\ek_\PKE,\pp,r_1,\ldots,r_n,\qdk_\PKE)$ to $\qA$ and obtains $\cert$.
    If $\PKE.\Vrfy(\vk_\PKE,\cert)=\bot$, it outputs $0$.
    Otherwise, it sends $(y_1,\ldots,y_n,z_1,\ldots,z_n)$ to $\qA$ and obtains
    $X'=(\widetilde{x},\widetilde{s}_1,\ldots,\widetilde{s}_n,\widetilde{\ct}_1,\ldots,\widetilde{\ct}_n)$.
    If $X'$ cannot be parsed in this form, it outputs $0$.
    \item $\qD$ outputs $1$ if
    \begin{align}
        a_0^\beta=\HPRG.\Eval(\pp,\widetilde{x},0)
    \end{align}
    and, for every $i\in[n]$,
    \begin{align}
        a_{i,\widetilde{x}_i}^\beta
        &=\HPRG.\Eval(\pp,\widetilde{x},i), &
        \widetilde{s}_i&=s_i^{\widetilde{x}_i}, &
        \widetilde{\ct}_i&=\ct_i^{\widetilde{x}_i\oplus 1}.
    \end{align}
    Otherwise, it outputs $0$.
\end{enumerate}

When $\beta=1$, the view given to $\qA$ is distributed exactly as in $\mathsf{Hyb}_3$.
Moreover, whenever $\qA$ wins $\mathsf{Hyb}_3$, its output satisfies all the checks above.
Consequently,
\begin{align}
    \Pr[1\gets\mathsf{Hyb}_3]
    \le \Pr[1\gets\qD(\pp,a_0^1,\{a_{i,b}^1\}_{i\in[n],b\in\bit})].
\end{align}
When $\beta=0$, for every fixed $\widetilde{x}\in\bit^n$, the values $a_0^0$ and $\{a_{i,\widetilde{x}_i}^0\}_{i\in[n]}$ are independent and uniform.
Acceptance implies that there exists some $\widetilde{x}\in\bit^n$ satisfying all the HPRG equalities in the check, regardless of how $\qA$ chooses its output.
Taking a union bound over all such $\widetilde{x}$, we obtain
\begin{align}
    \Pr[1\gets\qD(\pp,a_0^0,\{a_{i,b}^0\}_{i\in[n],b\in\bit})]
    \le 2^n\cdot 2^{-(n+1)\ell}=2^{n-(n+1)\ell}=\negl(\secp),
\end{align}
where the last equality uses $\ell=\omega(\log\secp)$, which we may assume by padding the encryption randomness.
Therefore, for infinitely many $\secp$,
\begin{align}
    &\Pr[1\gets\qD(\pp,a_0^1,\{a_{i,b}^1\}_{i\in[n],b\in\bit})]
    -\Pr[1\gets\qD(\pp,a_0^0,\{a_{i,b}^0\}_{i\in[n],b\in\bit})]\\
    &\ge \Pr[1\gets\mathsf{Hyb}_3]-\negl(\secp)\\
    &\ge |\Pr[1\gets\mathsf{Hyb}_3]-\Pr[1\gets\mathsf{Hyb}_4]|
    -\Pr[1\gets\mathsf{Hyb}_4]-\negl(\secp)\\
    &\ge \frac{1}{p(\secp)}-2^{-n}-\negl(\secp),
\end{align}
which is non-negligible and contradicts the security of the hinting PRG.
\end{proof}

In $\mathsf{Hyb}_4$, the distribution of $(y_1,...,y_n,z_1,...,z_n)$ is independent of $x$.
Thus, for any QPT adversary $\qA$, we have
\begin{align}
    \Pr[1\gets\mathsf{Hyb}_4] \le \negl(\secp).
\end{align}
By combining this inequality with \cref{eq:hyb0,eq:hyb1,eq:hyb2,eq:hyb3}, we have
\begin{align}
    \Pr[1\gets\mathsf{Hyb}_0] \le \negl(\secp)
\end{align}
and complete the proof.

\section{Trapdoor Functions with Secure Key Leasing with Domain Sampler}
\label{sec:TDFSKL_DS}

In \cref{sec:TDFSKL}, we have constructed a TDF-SKL from PRCT secure PKE-SKL (and hinting PRGs).
However, via a slightly different approach, we can construct a TDF-SKL with \textit{domain sampler} from the weaker assumptions, namely IND-CPA secure PKE (and hinting PRGs).
TDF-SKL with domain sampler are defined similarly to TDF-SKL, except that the syntax additionally includes a domain-sampling algorithm that samples a bit string from the domain of the trapdoor function.
As the security for TDF-SKL with domain sampler, we consider \textit{OW-VRA security}, where the adversary obtains the verification key after submitting the deletion certificate and one-wayness is defined with respect to the distribution induced by the domain-sampling algorithm.

The definition of TDF-SKL with domain sampler is as follows:
\begin{definition}[TDF-SKL with domain sampler]\label{def:TDF-SKL_DS}
Let $n$ be a polynomial.
A TDF-SKL with domain sampler is a tuple of algorithms $\mathsf{TDFSKL}=(\mathpzc{KG},\Samp,\Eval,\mathpzc{Inv},\mathpzc{Del},\Vrfy)$.
\begin{itemize}
    \item $\mathpzc{KG}(1^\secp)\to(\ek,\mathpzc{td},\vk)$: The key generation algorithm takes as input a security parameter $1^\secp$ and outputs a classical evaluation key $\ek$, a quantum trapdoor $\mathpzc{td}$, and a classical deletion verification key $\vk$.
    \item $\Samp(\ek)\to x$: The domain sampling algorithm takes as input an evaluation key $\ek$ and outputs $x\in\bit^{n(\secp)}$. 
    \item $\Eval(\ek,x)\to y$: The evaluation algorithm takes an evaluation key $\ek$ and a string $x\in\bit^{n(\secp)}$ as input and outputs $y$. This algorithm is deterministic.
    \item $\mathpzc{Inv}(\mathpzc{td},y)\to(x',\mathpzc{td}')$: The inversion algorithm takes a quantum trapdoor $\mathpzc{td}$ and a string $y$ as input and outputs $x'$ and a resulting trapdoor $\mathpzc{td}'$.
    \item $\mathpzc{Del}(\mathpzc{td})\to\cert$: The trapdoor deletion algorithm takes a quantum trapdoor $\mathpzc{td}$ as input and outputs a classical certificate $\cert$.
    \item $\Vrfy(\vk,\cert)\to\top/\bot$: The deletion verification algorithm takes a deletion verification key $\vk$ and a certificate $\cert$ as input and outputs $\top/\bot$. This algorithm is deterministic.
\end{itemize}
$\mathsf{TDFSKL}$ is required to satisfy the following conditions:
\begin{itemize}
    \item \textbf{Almost-all-keys inversion correctness:} 
    \begin{align}
        \Pr_{(\ek,\mathpzc{td},\vk)\gets\mathpzc{KG}(1^\secp)} \left[ \forall x\in\bit^{n(\secp)}, 
        \Pr_{(x',\mathpzc{td}')\gets\mathpzc{Inv}(\mathpzc{td},\Eval(\ek,x))}[x'=x]\ge 1-\negl(\secp) \right] \ge 1-\negl(\secp).
    \end{align}
    \item \textbf{Deletion verification correctness:} 
    \begin{align}
        \Pr \left[ \Vrfy(\vk,\cert)=\top : \begin{gathered} (\ek,\mathpzc{td},\vk)\gets\mathpzc{KG}(1^\secp) \\ \cert\gets\mathpzc{Del}(\mathpzc{td})
        \end{gathered} \right] \ge 1-\negl(\secp).
    \end{align}
    \item \textbf{OW-VRA security:}
    Consider the experiment $\mathsf{Exp}^{\mathsf{OW-VRA}}_{\mathsf{TDFSKL},\qA}(\secp)$ between the challenger and an adversary $\qA$:
    \begin{enumerate}
        \item The challenger generates $(\ek,\mathpzc{td},\vk)\gets\mathpzc{KG}(1^\secp)$ and sends $(\ek,\mathpzc{td})$ to $\qA$.
        \item $\qA$ sends $\cert$ to the challenger.
        \item If $\Vrfy(\vk,\cert)=\bot$, then the output of the experiment is 0. 
        Otherwise, the challenger generates $x\gets\Samp(\ek), y\coloneqq\Eval(\ek,x)$ and sends $(y,\vk)$ to $\qA$.
        \item $\cA$ sends $x'$ to the challenger.
        \item If $x'=x$, the output of the experiment is 1.
        Otherwise, the output of the experiment is 0.
    \end{enumerate}
    Then, for any QPT adversary $\qA$, 
    \begin{align}
        \Pr[1\gets\mathsf{Exp}^{\mathsf{OW-VRA}}_{\mathsf{TDFSKL},\qA}(\secp)] \le\negl(\secp).
    \end{align}
\end{itemize}
\end{definition}

We show the following theorem.
\begin{theorem}
    If IND-VRA secure PKE-SKL and hinting PRGs exist, then TDF-SKL with domain sampler exist.
\end{theorem}
Our construction is as follows:
\paragraph{Construction.}
Let $\mathsf{PKESKL}=(\PKE.\qKG,\PKE.\Enc,\PKE.\qDec,\PKE.\qDel,\PKE.\Vrfy)$ be an IND-VRA secure PKE-SKL scheme with message length $m=m(\secp)$, randomness length $\ell=\ell(\secp)$, and ciphertext length $c=c(\secp)$.
Let $(\HPRG.\Setup,\HPRG.\Eval)$ be a hinting PRG with input length $n=n(\secp)$ and output length $\ell$.
Let $\PRG$ be a PRG mapping $m$-bit strings to $3m$-bit strings.
We construct a TDF-SKL with domain sampler $\mathsf{TDFSKL}=(\qKG,\Samp,\Eval,\qInv,\qDel,\Vrfy)$ as follows:
\begin{itemize}
    \item $\qKG(1^\secp)\to(\ek,\qtd,\vk)$: 
    \begin{enumerate}
        \item Run $(\ek_\PKE,\qdk_\PKE,\vk_\PKE)\gets\PKE.\qKG(1^\secp)$, $\pp\gets\HPRG.\Setup(1^\secp)$.
        For $i\in[n]$, sample $r_i\gets\bit^{3m}$.
        \item Return $\ek\coloneqq(\ek_\PKE,\pp,r_1,...,r_n)$, $\qtd\coloneqq\qdk_\PKE$, and $\vk\coloneqq\vk_\PKE$.
    \end{enumerate}
    \item $\Samp(\ek)\to X$:
    \begin{enumerate}
        \item Parse $\ek=(\ek_\PKE,\pp,r_1,...,r_n)$.
        \item Sample $x\gets\bit^{n}$. For $i\in[n]$, sample $s_i\gets\bit^{m}$ and $\ct_i\gets\PKE.\Enc(\ek_\PKE,0^{m})$.
        \item Return $X\coloneqq(x,s_1,...,s_n,\ct_1,...,\ct_n)$.
    \end{enumerate}
    \item $\Eval(\ek,X)\to Y$:
    \begin{enumerate}
        \item Parse $\ek=(\ek_\PKE,\pp,r_1,...,r_n)$ and $X=(x,s_1,...,s_n,\ct_1,...,\ct_n)$.
        \item For $i\in[n]$,
        \begin{itemize}
            \item if $x_i=0$, let 
            \begin{align}
                y_i\coloneqq 
                \begin{pmatrix}
                    \PKE.\Enc(\ek_\PKE,s_i;\HPRG.\Eval(\pp,x,i)) \\ 
                    \ct_i
                \end{pmatrix}, 
                \quad z_i\coloneqq\PRG(s_i).
            \end{align}
            \item if $x_i=1$, let
            \begin{align}
                y_i\coloneqq
                \begin{pmatrix}
                    \ct_i \\
                    \PKE.\Enc(\ek_\PKE,s_i;\HPRG.\Eval(\pp,x,i)) 
                \end{pmatrix}, 
                \quad z_i\coloneqq\PRG(s_i) \oplus r_i.
            \end{align}
        \end{itemize}
        \item Return $Y\coloneqq(y_1,...,y_n,z_1,...,z_n)$.
    \end{enumerate}
    \item $\qInv(\qtd,Y)\to (X',\qtd')$:
    \begin{enumerate}
        \item Parse $\qtd=\qdk_\PKE$ and $Y=(y_1,...,y_n,z_1,...,z_n)$, where for each $i\in[n]$, $y_i=\begin{pmatrix} a_i \\ b_i \end{pmatrix}$.
        \item For $i\in[n]$, run $(s_i',\qdk_\PKE')\gets\PKE.\qDec(\qdk_\PKE,a_i)$ coherently.
        \begin{itemize}
            \item If $z_i=\PRG(s_i')$, let $x_i\coloneqq 0$, $s_i\coloneqq s_i'$, $\ct_i\coloneqq b_i$, and $\qtd'=\qdk_\PKE'$.
            \item If $z_i\neq\PRG(s_i')$, uncompute $\PKE.\qDec$, and run $(s_i'',\qdk_\PKE'')\gets\PKE.\Dec(\qdk_\PKE,b_i)$. Let $x_i\coloneqq 1$, $s_i\coloneqq s_i''$, $\ct_i\coloneqq a_i$, and $\qtd'=\qdk_\PKE''$.
        \end{itemize}
        \item Return $X'\coloneqq (x,s_1,...,s_n,\ct_1,...\ct_n)$ and $\qtd'$.
    \end{enumerate}
    \item $\qDel(\qtd)\to\cert$: Run $\cert\gets\PKE.\qDel(\qtd)$.
    \item $\Vrfy(\vk,\cert)\to\top/\bot$: Run $\top/\bot\gets\PKE.\Vrfy(\vk,\cert)$.
\end{itemize}

Then, the deletion verification correctness of $\mathsf{TDFSKL}$ immediately follows from that of $\mathsf{PKESKL}$.
Moreover, almost-all-keys inversion correctness can be shown via essentially the same proof as in \cref{lem:TDFSKL_correct}.
Thus, it suffices to show OW-VRA security.

\begin{lemma}\label{lem:TDF-SKL_DS_security}
    $\mathsf{TDFSKL}$ satisfies OW-VRA security.
\end{lemma}

\begin{proof}[Proof of \cref{lem:TDF-SKL_DS_security}]
We consider the following sequence of hybrids.

\noindent
$\mathsf{Hyb}_0$: This is the original security experiment between the challenger and an adversary $\qA$ that works as follows:
\begin{enumerate}
    \item The challenger generates $(\ek_\PKE,\qdk_\PKE,\vk_\PKE)\gets\PKE.\qKG(1^\secp)$, $\pp\gets\HPRG.\Setup(1^\secp)$ and $r_i\gets\bit^{3m}$ for each $i\in[n]$.
    The challenger sends $(\ek_\PKE,\pp,r_1,...,r_n,\qdk_\PKE)$ to $\qA$.
    \item $\qA$ outputs $\cert$.
    \item If $\PKE.\Vrfy(\vk_\PKE,\cert)=\bot$, the output of the experiment is 0.
    Otherwise, the challenger generates $x\gets\bit^n$, $s_i\gets\bit^{m}$, and $\ct_i\gets\PKE.\Enc(\ek_\PKE,0)$ for each $i\in[n]$.
    For each $i\in[n]$,
    \begin{itemize}
        \item if $x_i=0$, let
        \begin{align}
            y_i\coloneqq 
            \begin{pmatrix}
                \PKE.\Enc(\ek_\PKE,s_i;\HPRG.\Eval(\pp,x,i)) \\ 
                \ct_i
            \end{pmatrix}, 
            \quad z_i\coloneqq\PRG(s_i).
        \end{align}
        \item if $x_i=1$, let
        \begin{align}
            y_i\coloneqq 
            \begin{pmatrix}
                \ct_i \\
                \PKE.\Enc(\ek_\PKE,s_i;\HPRG.\Eval(\pp,x,i)) 
            \end{pmatrix}, 
            \quad z_i\coloneqq\PRG(s_i) \oplus r_i.
        \end{align}
    \end{itemize}
    The challenger sends $(y_1,...,y_n,z_1,...,z_n,\vk_\PKE)$ to $\qA$.
    \item $\qA$ outputs $X'$.
    If $X'=(x,s_1,...,s_n,\ct_1,...,\ct_n)$, the output of the experiment is 1.
    Otherwise, the output of the experiment is 0.
\end{enumerate}
Our goal is to show that for any QPT adversary $\qA$,
\begin{align}
    \Pr[1\gets\mathsf{Hyb}_0] \le \negl(\secp).
\end{align}

\noindent
$\mathsf{Hyb}_1$: This is identical to $\mathsf{Hyb}_0$ except that $(x,s_1,...,s_n)$ is sampled before $\cA$ outputs $\cert$ and for each $i\in[n]$, $r_i$ is replaced with the XOR of a uniformly random string and the output of $\PRG$.
The experiment works as follows:
\begin{enumerate}
    \item The challenger generates $(\ek_\PKE,\qdk_\PKE,\vk_\PKE)\gets\PKE.\qKG(1^\secp)$, $\pp\gets\HPRG.\Setup(1^\secp)$, and $x\gets\bit^n$.
    For each $i\in[n]$, the challenger generates $s_i^0\gets\bit^{m}$, $s_i^1\gets\bit^{m}$, $u_i\gets\bit^{3m}$, and $r_i\coloneqq\PRG(s_i^{x_i})\oplus u_i$.
    The challenger sends $(\ek_\PKE,\pp,r_1,...,r_n,\qdk_\PKE)$ to $\qA$.
    \item $\qA$ outputs $\cert$.
    \item If $\PKE.\Vrfy(\vk_\PKE,\cert)=\bot$, the output of the experiment is 0.
    Otherwise, the challenger generates $\ct_i\gets\PKE.\Enc(\ek_\PKE,0)$ for each $i\in[n]$.
    For each $i\in[n]$,
    \begin{itemize}
        \item if $x_i=0$, let
        \begin{align}
            y_i\coloneqq
            \begin{pmatrix}
                \PKE.\Enc(\ek_\PKE,s^{0}_i;\HPRG.\Eval(\pp,x,i)) \\ 
                \ct_i
            \end{pmatrix}, 
            \quad z_i\coloneqq\PRG(s^0_i).
        \end{align}
        \item if $x_i=1$, let
        \begin{align}
            y_i\coloneqq
            \begin{pmatrix}
                \ct_i \\
                \PKE.\Enc(\ek_\PKE,s^1_i;\HPRG.\Eval(\pp,x,i)) 
            \end{pmatrix}, 
            \quad z_i\coloneqq\PRG(s^1_i) \oplus r_i.
        \end{align}
    \end{itemize}
    The challenger sends $(y_1,...,y_n,z_1,...,z_n,\vk_\PKE)$ to $\qA$.
    \item $\qA$ outputs $X'$.
    If $X'=(x,s^{x_1}_1,...,s^{x_n}_n,\ct_1,...,\ct_n)$, the output of the experiment is 1.
    Otherwise, the output of the experiment is 0.
\end{enumerate}

\begin{claim}
For any adversary $\qA$,
\begin{align}
\label{eq:hyb0_vra}
    \Pr[1\gets\mathsf{Hyb}_0] = \Pr[1\gets\mathsf{Hyb}_1]
\end{align}
\end{claim}
\begin{proof}
    The claim follows because the distribution of $(x,s^{x_1}_1,...,s^{x_n}_n,r_1,...,r_n)$ in $\mathsf{Hyb}_1$ is identical to the distribution of $(x,s_1,...,s_n,r_1,...,r_n)$ in $\mathsf{Hyb}_0$, and reordering the sampling procedure as in $\mathsf{Hyb}_1$ does not affect the output of the experiment.
\end{proof}

\noindent
$\mathsf{Hyb}_2$: This is identical to $\mathsf{Hyb}_1$ except that $r_i$ is replaced with $\PRG(s_i^0)\oplus\PRG(s_i^1)$ for each $i\in[n]$. 
The experiment works as follows:
\begin{enumerate}
    \item The challenger generates $(\ek_\PKE,\qdk_\PKE,\vk_\PKE)\gets\PKE.\qKG(1^\secp)$, $\pp\gets\HPRG.\Setup(1^\secp)$ and $x\gets\bit^n$.
    For each $i\in[n]$, the challenger generates $s_i^0\gets\bit^{m}$ and $s_i^1\gets\bit^{m}$, and $r_i\coloneqq\PRG(s_i^0)\oplus\PRG(s_i^1)$.
    The challenger sends $(\ek_\PKE,\pp,r_1,...,r_n,\qdk_\PKE)$ to $\qA$.
    \item $\qA$ outputs $\cert$.
    \item If $\PKE.\Vrfy(\vk_\PKE,\cert)=\bot$, the output of the experiment is 0.
    Otherwise, the challenger generates $\ct_i\gets\PKE.\Enc(\ek_\PKE,0)$ for each $i\in[n]$.
    For each $i\in[n]$,
    \begin{itemize}
        \item if $x_i=0$, let
        \begin{align}
            y_i\coloneqq
            \begin{pmatrix}
                \PKE.\Enc(\ek_\PKE,s^0_i;\HPRG.\Eval(\pp,x,i)) \\ 
                \ct_i
            \end{pmatrix}, 
            \quad z_i\coloneqq\PRG(s^0_i).
        \end{align}
        \item if $x_i=1$, let
        \begin{align}
            y_i\coloneqq
            \begin{pmatrix}
                \ct_i \\
                \PKE.\Enc(\ek_\PKE,s^1_i;\HPRG.\Eval(\pp,x,i)) 
            \end{pmatrix}, 
            \quad z_i\coloneqq\PRG(s^1_i) \oplus r_i=\PRG(s_i^0).
        \end{align}
    \end{itemize}
    The challenger sends $(y_1,...,y_n,z_1,...,z_n,\vk_\PKE)$ to $\qA$.
    \item $\qA$ outputs $X'$.
    If $X'=(x,s^{x_i}_1,...,s^{x_i}_n,\ct_1,...,\ct_n)$, the output of the experiment is 1.
    Otherwise, the output of the experiment is 0.
\end{enumerate}
\begin{claim}
    For any QPT adversary $\qA$, 
    \begin{align}
    \label{eq:hyb1_vra}
        |\Pr[1\gets\mathsf{Hyb}_1]-\Pr[1\gets\mathsf{Hyb}_2]|\le\negl(\secp).
    \end{align}
\end{claim}
\begin{proof}
The claim follows from the security of $\PRG$.
For each $i\in\{0,...,n\}$, consider the experiment $G_i$, that is identical to $\mathsf{Hyb}_1$ except that for all $j\le i$, $r_j$ is replaced with $\PRG(s_j^0)\oplus\PRG(s_j^1)$.
Then, $G_0$ is identical to $\mathsf{Hyb}_1$ and $G_n$ is identical to $\mathsf{Hyb}_2$.
For the sake of contradiction , we assume that there exists a QPT adversary $\qA$ and a polynomial $p$ such that
\begin{align}
    |\Pr[1\gets\mathsf{Hyb}_1] - \Pr[1\gets\mathsf{Hyb}_2]| \ge \frac{1}{p(\secp)}
\end{align}
holds for infinitely many $\secp$.
Then, there exists $i^*\in[n]$ such that
\begin{align}
    |\Pr[1\gets G_{i^*-1}] - \Pr[1\gets G_{i^*}]| \ge \frac{1}{np(\secp)}
\end{align}
for infinitely many $\secp$.
By using $\qA$, we construct a QPT adversary $\qB$ that breaks the security of $\PRG$ as follows:
\begin{enumerate}
    \item $\qB$ takes $a$ as input, where $a=\PRG(s)$ for $s\gets\bit^{m}$ or $a\gets\bit^{3m}$.
    \item $\qB$ sample $j\gets[n]$.
    $\qB$ simulates $\mathsf{Hyb}_1$, where $\qB$ generates $(r_1,...,r_n)$ as follows:
    For each $i\in[n]$, $\qB$ generates $s_i^0\gets\bit^{m}$, $s_i^1\gets\bit^{m}$, and $u_i\gets\bit^{3m}$ and lets
    \begin{align}
        r_i \coloneqq
        \begin{cases}
            \PRG(s_i^0)\oplus \PRG(s_i^1) & \text{if } i<j \\ 
            a \oplus \PRG(s_i^{x_i}) & \text{if } i=j \\
            u_i\oplus \PRG(s_i^{x_i}) & \text{if } i>j.
        \end{cases}
    \end{align}
\end{enumerate}
Consider the case in which $j=i^*$, that occurs with probability $1/n$.
If $a=\PRG(s)$ for $s\gets\bit^{m}$, then $\qB$ simulates $G_{i^*}$.
If $a\gets\bit^{3m}$, then $\qB$ simulates $G_{i^*-1}$. 
Thus, 
\begin{align}
    &\left|\Pr_{s\gets\bit^{m}}[1\gets\qB(\PRG(s))] - \Pr_{a\gets\bit^{3m}}[1\gets\qB(a)] \right| \\
    &= \frac{1}{n} |\Pr[1\gets G_{i^*-1}] - \Pr[1\gets G_{i^*}] | \\
    &\ge \frac{1}{n^2p(\secp)}
\end{align}
for infinitely many $\secp$ and $\qB$ breaks the security of $\PRG$.
\end{proof}

\noindent
$\mathsf{Hyb}_3$: This is identical to $\mathsf{Hyb}_2$ except that $\ct_i$ is replaced with $\PKE.\Enc(\ek_\PKE,s_i^{x_i\oplus 1})$ for each $i\in[n]$. The experiment works as follows:
\begin{enumerate}
    \item The challenger generates $(\ek_\PKE,\qdk_\PKE,\vk_\PKE)\gets\PKE.\qKG(1^\secp)$, $\pp\gets\HPRG.\Setup(1^\secp)$ and $x\gets\bit^n$.
    For each $i\in[n]$, the challenger generates $s_i^0\gets\bit^{m}$, $s_i^1\gets\bit^{m}$, and $r_i\coloneqq\PRG(s_i^0)\oplus\PRG(s_i^1)$.
    The challenger sends $(\ek_\PKE,\pp,r_1,...,r_n,\qdk_\PKE)$ to $\qA$.
    \item $\qA$ outputs $\cert$.
    \item If $\PKE.\Vrfy(\vk_\PKE,\cert)=\bot$, the output of the experiment is 0.
    Otherwise, for each $i\in[n]$,
    \begin{itemize}
        \item if $x_i=0$, the challenger generates $\ct_i\gets\PKE.\Enc(\ek_\PKE,s_i^1)$ and lets
        \begin{align}
            y_i\coloneqq 
            \begin{pmatrix}
                \PKE.\Enc(\ek_\PKE,s_i^0;\HPRG.\Eval(\pp,x,i)) \\ 
                \ct_i
            \end{pmatrix}, 
            \quad z_i\coloneqq\PRG(s_i^0).
        \end{align}
        \item if $x_i=1$, the challenger generates $\ct_i\gets\PKE.\Enc(\ek_\PKE,s_i^0)$ and lets
        \begin{align}
            y_i\coloneqq 
            \begin{pmatrix}
                \ct_i \\
                \PKE.\Enc(\ek_\PKE,s_i^1;\HPRG.\Eval(\pp,x,i)) 
            \end{pmatrix}, 
            \quad z_i\coloneqq\PRG(s_i^0).
        \end{align}
    \end{itemize}
    The challenger sends $(y_1,...,y_n,z_1,...,z_n,\vk_\PKE)$ to $\qA$.
    \item $\qA$ outputs $X'$.
    If $X'=(x,s^{x_1}_1,...,s^{x_n}_n,\ct_1,...,\ct_n)$, the output of the experiment is 1.
    Otherwise, the output of the experiment is 0.
\end{enumerate}
\begin{claim}
    For any QPT adversary $\qA$,
    \begin{align}
    \label{eq:hyb2_vra}
        |\Pr[1\gets\mathsf{Hyb}_2] - \Pr[1\gets\mathsf{Hyb}_3]| \le \negl(\secp).
    \end{align}
\end{claim}
\begin{proof}
The claim follows from the IND-VRA security of $\mathsf{PKESKL}$.
For each $i\in\{0,...,n\}$, consider the experiment $H_i$, that is identical to $\mathsf{Hyb}_2$ except that for all $j\le i$, $\ct_j$ is replaced with $\PKE.\Enc(\ek_\PKE,s_i^{x_i\oplus 1})$.
Then, $H_0$ is identical to $\mathsf{Hyb}_2$ and $H_n$ is identical to $\mathsf{Hyb}_3$.
For the sake of contradiction , we assume that there exists a QPT adversary $\qA$ and a polynomial $p$ such that
\begin{align}
    |\Pr[1\gets\mathsf{Hyb}_2] - \Pr[1\gets\mathsf{Hyb}_3]| \ge \frac{1}{p(\secp)}
\end{align}
holds for infinitely many $\secp$.
Then, there exists $i^*\in[n]$ such that
\begin{align}
    |\Pr[1\gets H_{i^*-1}] - \Pr[1\gets H_{i^*}]| \ge \frac{1}{np(\secp)}
\end{align}
for infinitely many $\secp$.
By using $\qA$, we construct a QPT adversary $\qC$ that breaks the IND-VRA security of $\mathsf{PKESKL}$. 
$\qC$ behaves during the IND-VRA security game $\mathsf{Exp}^{\mathsf{IND-VRA}}_{\mathsf{PKESKL},\qC}(\secp,b)$ as follows:
\begin{enumerate}
    \item The challenger generates $(\ek_\PKE,\qdk_\PKE,\vk_\PKE)\gets\PKE.\qKG(1^\secp)$ and sends $(\ek_\PKE,\qdk_\PKE)$ to $\qC$.
    \item $\qC$ generates $\pp\gets\HPRG.\Setup(1^\secp)$, $x\gets\bit^n$, $s_i^0\gets\bit^{m}$, $s_i^1\gets\bit^{m}$, and $r_i=\PRG(s_i^0)\oplus\PRG(s_i^1)$. 
    $\qC$ runs $\cert\gets\qA(\ek_\PKE,\pp,r_1,...,r_n,\qdk_\PKE)$. 
    $\qC$ samples $j\gets[n]$ and sets $m_0\coloneqq 0$ and $m_1\coloneqq s_{j}^{x_j\oplus 1}$.
    $\qC$ sends $\cert$ and $(m_0,m_1)$ to the challenger.
    \item If $\PKE.\Vrfy(\vk_\PKE,\cert)=\bot$, then the output of the experiment is 0.
    Otherwise, the challenger generates $\ct^*_b\gets\PKE.\Enc(\ek_\PKE,m_b)$ and sends $\ct_b^*$ to $\qC$.
    \item $\qC$ simulates $\mathsf{Hyb}_2$, where $\qC$ generates $(\ct_1,...,\ct_n)$ as follows:
    \begin{align}
        \ct_i \coloneqq
        \begin{cases}
            \PKE.\Enc(\ek_\PKE,s_i^{x_i\oplus 1}) & \text{if } i<j \\ 
            \ct_b^* & \text{if } i=j \\
            \PKE.\Enc(\ek_\PKE,0) & \text{if } i>j.
        \end{cases}
    \end{align}
\end{enumerate}
Consider the case in which $j=i^*$, that occurs with probability $1/n$.
Then 
\begin{align}
    \Pr[1\gets H_{i^*-1}] &= \Pr[1\gets \mathsf{Exp}^{\mathsf{IND-VRA}}_{\mathsf{PKESKL},\qC}(\secp,0)] \\
    \Pr[1\gets H_{i^*}] &= \Pr[1\gets \mathsf{Exp}^{\mathsf{IND-VRA}}_{\mathsf{PKESKL},\qC}(\secp,1)].
\end{align}
Thus, 
\begin{align}
    &|\Pr [1\gets \mathsf{Exp}^{\mathsf{IND-VRA}}_{\mathsf{PKESKL},\qC}(\secp,0)] - \Pr[1\gets \mathsf{Exp}^{\mathsf{IND-VRA}}_{\mathsf{PKESKL},\qC}(\secp,1)] | \\
    &= \frac{1}{n} |\Pr[1\gets H_{i^*-1}] - \Pr[1\gets H_{i^*}] | \\
    &\ge \frac{1}{n^2p(\secp)}
\end{align}
for infinitely many $\secp$ and $\qC$ breaks the IND-VRA security of $\mathsf{PKESKL}$.
\end{proof}

\noindent
$\mathsf{Hyb}_4$: This is identical to $\mathsf{Hyb}_3$ except that for each $i\in[n]$, $\HPRG.\Eval(\pp,x,i)$ is replaced with a random string. The experiment works as follows:
\begin{enumerate}
    \item The challenger generates $(\ek_\PKE,\qdk_\PKE,\vk_\PKE)\gets\PKE.\qKG(1^\secp)$, $\pp\gets\HPRG.\Setup(1^\secp)$ and $x\gets\bit^n$.
    For $i\in[n]$, the challenger generates $s_i^0\gets\bit^{m}$, $s_i^1\gets\bit^{m}$, and $r_i\coloneqq\PRG(s_i^0)\oplus\PRG(s_i^1)$.
    The challenger sends $(\ek_\PKE,\pp,r_1,...,r_n,\qdk_\PKE)$ to $\qA$.
    \item $\qA$ outputs $\cert$.
    \item If $\PKE.\Vrfy(\vk_\PKE,\cert)=\bot$, the output of the experiment is 0.
    Otherwise, For $i\in[n]$, the challenger generates $v_i\gets\bit^{\ell}$.
    For each $i\in[n]$,
    \begin{itemize}
        \item if $x_i=0$, the challenger generates $\ct_i\gets\PKE.\Enc(\ek_\PKE,s_i^1)$ and lets
        \begin{align}
            y_i\coloneqq
            \begin{pmatrix}
                \PKE.\Enc(\ek_\PKE,s_i^0;v_i) \\ 
                \ct_i
            \end{pmatrix}, 
            \quad z_i\coloneqq\PRG(s_i^0).
        \end{align}
        \item if $x_i=1$, the challenger generates $\ct_i\gets\PKE.\Enc(\ek_\PKE,s_i^0)$ and lets
        \begin{align}
            y_i \coloneqq 
            \begin{pmatrix}
                \ct_i \\
                \PKE.\Enc(\ek_\PKE,s_i^1;v_i) 
            \end{pmatrix}, 
            \quad z_i \coloneqq \PRG(s_i^0).
        \end{align}
    \end{itemize}
    The challenger sends $(y_1,...,y_n,z_1,...,z_n,\vk_\PKE)$ to $\qA$.
    \item $\qA$ outputs $X'$.
    If $X'=(x,s^{x_1}_1,...,s^{x_n}_n,\ct_1,...,\ct_n)$, the output of the experiment is 1.
    Otherwise, the output of the experiment is 0.
\end{enumerate}
\begin{claim}
    For any QPT adversary $\qA$,
    \begin{align}
    \label{eq:hyb3_vra}
        |\Pr[1\gets\mathsf{Hyb}_3] - \Pr[1\gets\mathsf{Hyb}_4]|\le\negl(\secp).
    \end{align}
\end{claim}
\begin{proof}
The security follows from the security of hinting PRGs.
For the sake of contradiction, we assume that there exists a QPT adversary $\qA$ and a polynomial $p$ such that
\begin{align}
    |\Pr[1\gets\mathsf{Hyb}_3] - \Pr[1\gets\mathsf{Hyb}_4]|\ge\frac{1}{p(\secp)}
\end{align}
for infinitely many $\secp$.
By using $\qA$, we construct a QPT adversary $\qD$ that breaks the security of hinting PRGs as follows:
\begin{enumerate}
    \item $\qD$ takes $(\pp, a_0^\beta, \{a_{i,b}^\beta\}_{i\in[n],b\in\bit})$ as input, where $\pp\gets\HPRG.\Setup(1^\secp)$, $x\gets\bit^{n}$, $\beta\gets\bit$, $a_0^0\gets\bit^{\ell}$, $a_0^1\coloneqq\HPRG.\Eval(\pp,x,0)$, $a_{i,0}^0\gets\bit^{\ell}$, $a_{i,1}^0\gets\bit^{\ell}$, $a_{i,x_i}^1\coloneqq\HPRG.\Eval(\pp,x,i)$, and $a_{i,x_i\oplus 1}^1\gets\bit^{\ell}$ for each $i\in[n]$.
    \item $\qD$ simulates $\mathsf{Hyb}_3$, where $\qD$ generates $(y_1,...,y_n)$ as follows:
    For each $i\in[n]$, $\qD$ generates $s_i^0\gets\bit^{m}$, $s_i^1\gets\bit^{m}$, and
    \begin{align}
        y_i\coloneqq
        \begin{pmatrix}
            \PKE.\Enc(\ek_\PKE,s_i^0;a_{i,0}^\beta) \\
            \PKE.\Enc(\ek_\PKE,s_i^1;a_{i,1}^\beta)
        \end{pmatrix}.
    \end{align}
\end{enumerate}
Then, 
\begin{align}
    &\Pr[1\gets\mathsf{Hyb}_3] = \Pr[1\gets\cD(\pp, a_0^0, \{a_{i,b}^0\}_{i\in[n],b\in\bit})] \\
    &\Pr[1\gets\mathsf{Hyb}_4] = \Pr[1\gets\cD(\pp, a_0^1, \{a_{i,b}^1\}_{i\in[n],b\in\bit})] 
\end{align}
Therefore we have,
\begin{align}
    |\Pr[1\gets\cD(\pp, a_0^0, \{a_{i,b}^0\}_{i\in[n],b\in\bit})] - \Pr[1\gets\cD(\pp, a_0^1, \{a_{i,b}^1\}_{i\in[n],b\in\bit})] | \ge \frac{1}{\poly(\secp)}
\end{align}
for infinitely many $\secp$.
\end{proof}

In $\mathsf{Hyb}_4$, the distribution of $(y_1,...,y_n,z_1,...,z_n)$ is independent of $x$.
Thus, for any QPT adversary $\qA$, we have
\begin{align}
    \Pr[1\gets\mathsf{Hyb}_4] \le \negl(\secp).
\end{align}
By combining this inequality with \cref{eq:hyb0_vra,eq:hyb1_vra,eq:hyb2_vra,eq:hyb3_vra}, we have
\begin{align}
    \Pr[1\gets\mathsf{Hyb}_0] \le \negl(\secp)
\end{align}
and complete the proof.

\end{proof}

\section{Construction of IND-VRA Secure Robust PKE-SKL}
\label{sec:robustPKESKL}
As an application of TDF-SKL, we construct an IND-VRA secure robust PKE-SKL scheme from TDF-SKL with domain sampler via quantum Goldreich-Levin with quantum auxiliary input (\cref{lem:quantGL}).

We define robust PKE-SKL schemes as follows:
\begin{definition}[Robust PKE-SKL]
We say that a PKE-SKL scheme $(\mathpzc{KG},\Enc,\mathpzc{Dec},\mathpzc{Del},\Vrfy)$ is robust if for any $\ct$, $\TD(\qdk,\qdk')\le\negl(\secp)$, where $(\ek,\qdk,\vk)\gets\qKG(1^\secp)$, $(m',\qdk')\gets\qDec(\qdk,\ct)$.
\end{definition}

Our construction of IND-VRA secure robust PKE-SKL schemes is the following:
\paragraph{Construction.}
Let $\mathsf{TDFSKL}=(\mathsf{TDF.}\mathpzc{KG},\mathsf{TDF.}\Samp,\mathsf{TDF.}\Eval,\mathsf{TDF.}\mathpzc{Inv},\mathsf{TDF.}\mathpzc{Del},\mathsf{TDF.}\Vrfy)$ be a TDF-SKL with domain sampler for domain $\bit^{n(\secp)}$, where $n$ is some polynomial.
We construct an IND-VRA secure robust PKE-SKL scheme $\mathsf{PKESKL}=(\mathpzc{KG},\Enc,\mathpzc{Dec},\mathpzc{Del},\Vrfy)$ with message space $\cM=\bit$ as follows:
\begin{itemize}
    \item $\mathpzc{KG}(1^\secp)\to(\ek,\mathpzc{dk},\vk)$: 
    Run $(\ek_\TDF,\qtd_\TDF,\vk_\TDF)\gets\mathsf{TDF.}\mathpzc{KG}(1^\secp)$. 
    Return $\ek\coloneqq\ek_\TDF$, $\qdk\coloneqq(\ek_\TDF,\qtd_\TDF)$, and $\vk\coloneqq\vk_\TDF$.
    \item $\Enc(\ek,m)\to\ct$:
    On input an encryption key $\ek=\ek_\TDF$ and a message $m\in\bit$, generate $r\gets\bit^{n(\secp)}$, $x\gets\mathsf{TDF.}\Samp(\ek_\TDF)$, $y\coloneqq\mathsf{TDF.}\Eval(\ek_\TDF,x)$, and $b\coloneqq(x\cdot r)\oplus m$.
    Return $\ct\coloneqq(y,r,b)$.
    \item $\mathpzc{Dec}(\mathpzc{dk},\ct)\to(m',\mathpzc{dk}')$:
    On input a decryption key $\mathpzc{dk}=(\ek_\TDF,\qtd_\TDF)$ and a ciphertext $\ct=(y,r,b)$, run $(x',\qtd_\TDF')\gets\mathsf{TDF.}\mathpzc{Inv}(\qtd_\TDF,y)$.
    If $\mathsf{TDF.}\Eval(\ek_\TDF,x')=y$, then let $m'\coloneqq(x'\cdot r)\oplus b$, $\mathpzc{dk}'\coloneqq(\ek_\TDF,\qtd_\TDF')$.
    Otherwise, let  $m'\coloneqq\bot$, $\mathpzc{dk}'\coloneqq(\ek_\TDF,\qtd_\TDF')$.
    Output $(m',\mathpzc{dk}')$.
    \item $\mathpzc{Del}(\mathpzc{dk})\to\cert$: On input $\mathpzc{dk}=(\ek_\TDF,\qtd_\TDF)$, run $\cert\gets\mathsf{TDF.}\mathpzc{Del}(\qtd_\TDF)$.
    \item $\Vrfy(\vk,\cert)\to\top/\bot$: On input $\vk=\vk_\TDF$ and $\cert$, run $\top/\bot\gets\TDF.\Vrfy(\vk_\TDF,\cert)$.
\end{itemize}

The decryption correctness and the deletion verification correctness of $\mathsf{PKESKL}$ immediately follows from the almost-all-keys inversion correctness and the deletion verification correctness of $\mathsf{TDFSKL}$.
Thus, it suffices to show the IND-VRA security and the robustness of $\mathsf{PKESKL}$.

\begin{lemma}
\label{lem:robustPKESKL_security}
    $\mathsf{PKESKL}$ is IND-VRA secure.
\end{lemma}

The proof of this lemma is essentially the same as that of Lemma 3.12 in \cite{EC:AKNYY23}.
\begin{proof}[Proof of \cref{lem:robustPKESKL_security}]
\label{sec:robustPKESKL_security}
For the sake of contradiction, assume that $\mathsf{PKESKL}$ is not IND-VRA secure.
Then there exist a QPT adversary $\qA$ and a polynomial $p$ such that
\begin{align}
    \Pr_{b\gets\bit} [b\gets\mathsf{Exp}^{\mathsf{IND-VRA}}_{\mathsf{PKESKL},\qA}(\secp,b)] \ge \frac{1}{2} + \frac{1}{p(\secp)}
\end{align}
for infinitely many $\secp\in\N$.
Let $\Lambda$ be the set of such $\secp$.
We divide $\qA$ into the following two algorithms $\qA_0$ and $\qA_1$:
\begin{itemize}
    \item $\qA_0$: It takes $(\ek,\qdk)$ as input and outputs a certificate $\cert$ and quantum auxiliary information $\mathpzc{aux}$.
    \item $\qA_1$: It takes $(y,r,b)$, $\vk$, and $\mathpzc{aux}$ as input and outputs $b'$. 
\end{itemize}
Then, for all $\secp\in\Lambda$,
\begin{align}
    &\frac{1}{2}+\frac{1}{p(\secp)} \\
    &\le \Pr_{b\gets\bit} [b\gets\mathsf{Exp}^{\mathsf{IND-VRA}}_{\mathsf{PKESKL},\qA}(\secp,b)] \\ 
    &= \Pr_{b\gets\bit} [b\gets\mathsf{Exp}^{\mathsf{IND-VRA}}_{\mathsf{PKESKL},\qA}(\secp,b) \mid \Vrfy(\vk,\cert)=\top] \Pr[\Vrfy(\vk,\cert)=\top] \\
    &\quad + \Pr_{b\gets\bit} [b\gets\mathsf{Exp}^{\mathsf{IND-VRA}}_{\mathsf{PKESKL},\qA}(\secp,b) \mid \Vrfy(\vk,\cert)=\bot] (1-\Pr[\Vrfy(\vk,\cert)=\top]) \\
    &= \Pr_{b\gets\bit} [b\gets\mathsf{Exp}^{\mathsf{IND-VRA}}_{\mathsf{PKESKL},\qA}(\secp,b) \mid \Vrfy(\vk,\cert)=\top] \Pr[\Vrfy(\vk,\cert)=\top] \\
    &\quad + \frac{1}{2}(1-\Pr[\Vrfy(\vk,\cert)=\top]). \\
\end{align}
Thus, for all $\secp\in\Lambda$,
\begin{align}
\label{eq:condition}
    \Pr_{b\gets\bit} [b\gets\mathsf{Exp}^{\mathsf{IND-VRA}}_{\mathsf{PKESKL},\qA}(\secp,b) \mid \Vrfy(\vk,\cert)=\top] \ge \frac{1}{2} + \frac{1}{p(\secp)\Pr[\Vrfy(\vk,\cert)=\top]}.
\end{align}
We define the following set $G$:
\begin{align}
    G\coloneqq \left\{ (\mathpzc{aux},\vk,x,y) : 
    \Pr_{\substack{r\gets\bit^{n(\secp)} \\ b\gets\bit \\ b'\coloneqq (x\cdot r)\oplus b}}[b\gets\qA_1(\mathpzc{aux},\vk,y,r,b')] \ge \frac{1}{2} + \frac{1}{2p(\secp)}
    \right\}.
\end{align}
By \cref{eq:condition} and an averaging argument,
\begin{align}
    \Pr [(\mathpzc{aux},\vk,x,y)\in G \mid \Vrfy(\vk,\cert)=\top] \ge \frac{1}{2p(\secp)\Pr[\Vrfy(\vk,\cert)=\top]}
\end{align}
holds for all $\secp\in\Lambda$, where $(\ek,\qdk,\vk)\gets\qKG(1^\secp)$, $(\cert,\mathpzc{aux})\gets\qA_0(\ek,\qdk)$, $x\gets\TDF.\Samp(\ek)$, and $y\coloneqq\TDF.\Eval(\ek,x)$.
Thus,
\begin{align}
\label{eq:vrfy_and_succ}
    \Pr \left[ \Vrfy(\vk,\cert)=\top \land \Pr_{\substack{r\gets\bit^{n(\secp)} \\ b\gets\bit \\ b'\coloneqq (x\cdot r)\oplus b}}[b\gets\qA_1(\mathpzc{aux},\vk,y,r,b')] \ge \frac{1}{2} + \frac{1}{2p(\secp)} \right] \ge \frac{1}{2p(\secp)}
\end{align}
for all $\secp\in\Lambda$, where the outer probability is taken over $(\ek,\qdk,\vk)\gets\qKG(1^\secp)$, $(\cert,\mathpzc{aux})\gets\qA_0(\ek,\qdk)$, $x\gets\TDF.\Samp(\ek)$, and $y\coloneqq\TDF.\Eval(\ek,x)$.


By using $\qA$, we construct a QPT adversary $\qB$ that breaks the OW-VRA security of $\mathsf{TDFSKL}$.
Let $\mathsf{Exp}^{\mathsf{OW-VRA}}_{\mathsf{TDFSKL},\qB}(\secp)$ be the OW-VRA security game between the challenger and the adversary $\qB$ that works as follows:
\begin{enumerate}
    \item The challenger runs $(\ek_\TDF,\qtd_\TDF,\vk_\TDF)\gets\TDF.\qKG(1^\secp)$ and sends $(\ek_\TDF,\qtd_\TDF)$ to $\qB$.
    \item $\qB$ runs $(\cert,\mathpzc{aux})\gets\qA_0(\ek_\TDF,\qtd_\TDF)$ and sends $\cert$ to the challenger.
    \item If $\TDF.\Vrfy(\vk_\TDF,\cert)=\bot$, then the output of the experiment is 0.
    Otherwise, the challenger generates $x\gets\TDF.\Samp(\ek_\TDF)$, $y\coloneqq \TDF.\Eval(\ek_\TDF,x)$ and sends $(y,\vk_\TDF)$ to $\qB$.
    \item $\qB$ does the following:
    \begin{enumerate}
        \item Generate $r\gets\bit^{n(\secp)}$.
        \item Let $\qB^*$ be the algorithm in \cref{fig:algo_B*}.
        \item Run $x'\gets\mathpzc{Ext}([\qB^*],\mathpzc{aux},\vk_\TDF,y)$.
    \end{enumerate}
    \item The challenger outputs 1 if $x'=x$.
\end{enumerate}

\begin{figure}[H]
\centering
\fbox{
        \begin{minipage}{0.6\columnwidth}
            \begin{center} \textbf{Algorithm $\qB^*$} \end{center}
            \textbf{Input:} $\mathpzc{aux}$, $\vk_\TDF$, $y$, $r$.
            \begin{enumerate}
                \item Sample $b\gets\bit$.
                \item Run $\coin'\gets\qA_1(\mathpzc{aux},\vk_\TDF,y,r,b)$.
                \item Output $b\oplus \coin'$.
            \end{enumerate}
        \end{minipage}
}
\caption{The description of the algorithm $\qB^*$}
\label{fig:algo_B*}
\end{figure}

To show that $\qB$ outputs $x$ with high probability, we first observe that $\qB^*$ outputs $x\cdot r$ with high probability.
In fact, by \cref{eq:vrfy_and_succ},
\begin{align}
    & \Pr \left[ \TDF.\Vrfy(\vk_\TDF,\cert)=\top \land \Pr_{r\gets\bit^{n(\secp)}}\left[x\cdot r\gets\qB^*(\mathpzc{aux},\vk_\TDF,y,r) \right] \ge \frac{1}{2}+\frac{1}{2p(\secp)} \right] \\
    &= \Pr \left[ \TDF.\Vrfy(\vk_\TDF,\cert)=\top \land \Pr_{\substack{r\gets\bit^{n(\secp)} \\ b\gets\bit \\ b'\coloneqq (x\cdot r)\oplus b}}[ b\gets\qA_1(\mathpzc{aux},\vk_\TDF,y,r,b')] \ge \frac{1}{2} + \frac{1}{2p(\secp)} \right] \\
    &\ge \frac{1}{2p(\secp)} 
\end{align}
holds for all $\secp\in\Lambda$, where the outer probability is taken over $(\ek_\TDF,\qtd_\TDF,\vk_\TDF)\gets\TDF.\qKG(1^\secp)$, $(\cert,\mathpzc{aux})\gets\qA_0(\ek_\TDF,\qtd_\TDF)$, $x\gets\TDF.\Samp(\ek_\TDF)$, and $y\coloneqq\TDF.\Eval(\ek_\TDF,x)$.
By \cref{lem:quantGL},
\begin{align}
    \Pr \left[ \TDF.\Vrfy(\vk_\TDF,\cert)=\top \land \Pr[ x\gets\mathpzc{Ext}([\qB^*],\mathpzc{aux},\vk_\TDF,y) ] \ge \frac{1}{p(\secp)^2} \right] \ge \frac{1}{2p(\secp)}
\end{align}
for all $\secp\in\Lambda$, where $(\ek_\TDF,\qtd_\TDF,\vk_\TDF)\gets\TDF.\qKG(1^\secp)$, $(\cert,\mathpzc{aux})\gets\qA_0(\ek_\TDF,\qtd_\TDF)$, $x\gets\TDF.\Samp(\ek_\TDF)$, and $y\coloneqq\TDF.\Eval(\ek_\TDF,x)$.
Then,
\begin{align}
    \Pr \left[ \TDF.\Vrfy(\vk_\TDF,\cert)=\top \land x\gets\mathpzc{Ext}([\qB^*],\mathpzc{aux},\vk_\TDF,y) \right] \ge \frac{1}{2p(\secp)^3}
\end{align}
for all $\secp\in\Lambda$, where $(\ek_\TDF,\qtd_\TDF,\vk_\TDF)\gets\TDF.\qKG(1^\secp)$, $(\cert,\mathpzc{aux})\gets\qA_0(\ek_\TDF,\qtd_\TDF)$, $x\gets\TDF.\Samp(\ek_\TDF)$, and $y\coloneqq\TDF.\Eval(\ek_\TDF,x)$.
Therefore $\qB$ breaks the OW-VRA security and we complete the proof. 
\end{proof}

\begin{lemma}
\label{lem:robustPKESKL_robust}
     $\mathsf{PKESKL}$ is robust.
\end{lemma}

\begin{proof}[Proof of \cref{lem:robustPKESKL_robust}]
We consider the following two cases.

\begin{itemize}
    \item[(A)] There exists $x\in\bit^{n(\secp)}$ such that $y=\TDF.\Eval(\ek_\TDF,x)$:
By almost-all-keys inversion correctness of TDF-SKL, 
\begin{align}
\label{eq:good_y}
    \Pr \left[x'=x :
    \begin{gathered}
        (\ek_\TDF,\qtd_\TDF,\vk_\TDF)\gets\TDF.\qKG(1^\secp) \\ 
        (x',\qtd'_\TDF)\gets\TDF.\qInv(\qtd_\TDF,y)
    \end{gathered}\right] \ge 1-\negl(\secp).
\end{align}

Let $X$ be a POVM element that acts on $\qtd_\TDF$ and corresponds to the following measurement result:
\begin{enumerate}
    \item Run $(x',\qtd'_\TDF)\gets\TDF.\qInv(\qtd_\TDF,y)$.
    \item $x'=x$.
\end{enumerate}
By the gentle measurement lemma (\cref{lem:gentle}) and \cref{eq:good_y}, we have 
\begin{align}
    \TD(\qdk,\qdk')=\TD(\qtd_\TDF,\qtd_\TDF')\le\negl(\secp).
\end{align}
\item[(B)] For any $x\in\bit^{n(\secp)}$, $y\neq\TDF.\Eval(\ek_\TDF,x)$:
By the definition of $\qDec$, we have
\begin{align}
\label{eq:bad_y}
    \Pr \left[x'=\bot :
    \begin{gathered}
        (\ek_\TDF,\qtd_\TDF,\vk_\TDF)\gets\TDF.\qKG(1^\secp) \\ 
        (x',\qtd'_\TDF)\gets\TDF.\qInv(\qtd_\TDF,y)
    \end{gathered}\right] =1.
\end{align}
Then the gentle measurement lemma (\cref{lem:gentle}) and \cref{eq:bad_y} imply 
\begin{align}
    \TD(\qdk,\qdk')=\TD(\qtd_\TDF,\qtd'_\TDF)=0.
\end{align}
\end{itemize}
Therefore, in both cases, we have $\TD(\qdk,\qdk')\le\negl(\secp)$ and complete the proof.
\end{proof}

\section{Proof of \cref{lem:PTDF_OW}}
\label{sec:Proof_PTDF}
We consider the following sequence of hybrids:

\noindent
$\mathsf{Hyb}_0$: This is identical to the original security game, except that the challenge images $y^*_1$ and $y^*_2$ are defined as the outputs of the original circuit $C$, rather than those of the obfuscated circuit $\hat{C}$. 
\begin{enumerate}
    \item The challenger generates $f\gets\OWF.\Gen(1^\secp)$, $k\gets\PPRF.\Gen(1^\secp)$ and $\hat{C}\gets iO(1^\secp,C)$.
    \item The challenger generates $x^*_1\gets\bit^\secp$, $t^*_1\coloneqq f(x^*_1)$, $y^*_1\coloneqq C(x^*_1)$, $x^*_2\gets\bit^\secp$, $t^*_2\coloneqq f(x^*_2)$, $y^*_2\coloneqq C(x^*_2)$, and $k_{t^*_1,t^*_2}\gets\PPRF.\Puncture(k,(t^*_1,t^*_2))$.
    Let $\td^*\coloneqq(f,k_{t^*_1,t^*_2},\{t^*_1,t^*_2\},y^*_1,y^*_2)$.
    \item $(x'_1,x'_2)\gets\qA(\hat{C},\td^*,y^*_1,y^*_2)$.
    \item Output 1 if $x'_1=x^*_1$ or $x'_2=x^*_2$.
\end{enumerate}
By the perfect correctness of $iO$, the output distributions in the original security game and $\mathsf{Hyb}_0$ are identical.
Hence, our goal is to show that for any QPT adversary $\qA$,
\begin{align}
    \Pr[1\gets\mathsf{Hyb}_0] \le \negl(\secp).
\end{align}

\noindent
$\mathsf{Hyb}_1$: This is identical to $\mathsf{Hyb}_0$ except that $\qA$ obtains the obfuscation of the circuit $C^*$ (\cref{fig:circuit_TDF_punc}) instead of that of $C$.
\begin{enumerate}
    \item The challenger generates $f\gets\OWF.\Gen(1^\secp)$, $k\gets\PPRF.\Gen(1^\secp)$ and $\hat{C}\gets iO(1^\secp,C)$.  
    \item The challenger generates $x^*_1\gets\bit^\secp$, $t^*_1\coloneqq f(x^*_1)$, $y^*_1\coloneqq C(x^*_1)$, $x^*_2\gets\bit^\secp$, $t^*_2\coloneqq f(x^*_2)$, $y^*_2\coloneqq C(x^*_2)$, $k_{t^*_1,t^*_2}\gets\PPRF.\Puncture(k,(t^*_1,t^*_2))$, and $\hat{C^*}\gets iO(1^\secp,C^*)$.
    Let $\td^*\coloneqq(f,k_{t^*_1,t^*_2},\{t^*_1,t^*_2\},y^*_1,y^*_2)$.
    \item $(x'_1,x'_2)\gets\qA(\hat{C^*},\td^*,y^*_1,y^*_2)$.
    \item Output 1 if $x'_1=x^*_1$ or $x'_2=x^*_2$.
\end{enumerate}

\begin{claim}
For any QPT adversary $\qA$, 
    \begin{align}
        |\Pr[1\gets\mathsf{Hyb}_0]-\Pr[1\gets\mathsf{Hyb}_1]|\le\negl(\secp).
    \end{align}
\end{claim}
\begin{proof}
We show this claim by using the security of $iO$.
Let us consider the QPT adversary $\qB$ against iO that behaves as follows:
\begin{enumerate}
    \item $\qB$ generates $f\gets\OWF.\Gen(1^\secp)$, $k\gets\PPRF.\Gen(1^\secp)$, $x^*_1\gets\bit^\secp$, $t^*_1\coloneqq f(x^*_1)$, $y^*_1\coloneqq C(x^*_1)$, $x^*_2\gets\bit^\secp$, $t^*_2\coloneqq f(x^*_2)$, $y^*_2\coloneqq C(x^*_2)$, and $k_{t^*_1,t^*_2}\gets\PPRF.\Puncture(k,(t^*_1,t^*_2))$.
    It sets $\td^*\coloneqq(f,k_{t^*_1,t^*_2},\{t^*_1,t^*_2\},y^*_1,y^*_2)$.
    $\qB$ sends the circuits $(C,C^*)$ to the challenger. 
    \item The challenger samples $b\gets\bit$ and runs $\hat{C}_b\gets iO(1^\secp,C_b)$, where $C_0\coloneqq C$ and $C_1\coloneqq C^*$.
    The challenger sends $\hat{C}_b$ to $\qB$.
    \item $\qB$ runs $(x'_1,x'_2)\gets\qA(\hat{C}_b,\td^*,y^*_1,y^*_2)$.
    $\qB$ outputs 1 if $x'_1=x^*_1$ or $x'_2=x^*_2$.
\end{enumerate}
Then, 
\begin{align}
    &\Pr[1\gets\mathsf{Hyb}_0] = \Pr[1\gets\qB\mid b=0] \\
    &\Pr[1\gets\mathsf{Hyb}_1] = \Pr[1\gets\qB\mid b=1].
\end{align}
Note that circuits $C$ and $C^*$ have the same functionality with overwhelming probability over the choice of $f\gets\OWF.\Gen(1^\secp)$ because of the injectivity of $f$.
Thus, by the security of iO, we have
\begin{align}
    |\Pr[1\gets\mathsf{Hyb}_0] - \Pr[1\gets\mathsf{Hyb}_1]| \le \negl(\secp).
\end{align}
\end{proof}

\noindent
$\mathsf{Hyb}_2$: This is identical to $\mathsf{Hyb}_1$ except that $y^*_1$ and $y^*_2$ are replaced with $(t^*_1,r_1\oplus x^*_1)$ and $(t^*_2,r_2\oplus x^*_2)$ respectively, where $r_1$ and $r_2$ are random strings.
\begin{enumerate}
    \item The challenger generates $f\gets\OWF.\Gen(1^\secp)$, $k\gets\PPRF.\Gen(1^\secp)$.
    \item The challenger generates $r_1\gets\bit^\secp$, $r_2\gets\bit^\secp$, $x^*_1\gets\bit^\secp$, $t^*_1\coloneqq f(x^*_1)$, $y^*_1\coloneqq (t^*_1,r_1\oplus x^*_1)$, $x^*_2\gets\bit^\secp$, $t^*_2\coloneqq f(x^*_2)$, $y^*_2\coloneqq (t^*_2,r_2\oplus x^*_2)$, $k_{t^*_1,t^*_2}\gets\PPRF.\Puncture(k,(t^*_1,t^*_2))$ and $\hat{C^*}\gets iO(1^\secp,C^*)$.
    Let $\td^*\coloneqq(f,k_{t^*_1,t^*_2},\{t^*_1,t^*_2\},y^*_1,y^*_2)$.
    \item $(x'_1,x'_2)\gets\qA(\hat{C^*},\td^*,y^*_1,y^*_2)$.
    \item Output 1 if $x'_1=x^*_1$ or $x'_2=x^*_2$.
\end{enumerate}
\begin{claim}
For any QPT adversary $\qA$,
    \begin{align}
        |\Pr[1\gets\mathsf{Hyb}_1]-\Pr[1\gets\mathsf{Hyb}_2]|\le\negl(\secp).
    \end{align}
\end{claim}
\begin{proof}
We show this claim by using the security of the puncturable PRF $\PPRF$.
Let us consider the QPT adversary $\qD$ against $\PPRF$ that behaves as follows:
\begin{enumerate}
    \item $\qD$ generates $f\gets\OWF.\Gen(1^\secp)$, $x^*_1\gets\bit^\secp$, $t^*_1\coloneqq f(x^*_1)$, $x^*_2\gets\bit^\secp$, and $t^*_2\coloneqq f(x^*_2)$.
    \item The challenger generates $k\gets\PPRF.\Gen(1^\secp)$ and $k_{t^*_1,t^*_2}\gets\PPRF.\Puncture(k,(t^*_1,t^*_2))$.
    The challenger samples $b\gets\bit$. 
    If $b=0$, the challenger sets $a_1\coloneqq \PPRF.\Eval(k,t^*_1)$ and $a_2\coloneqq \PPRF.\Eval(k,t^*_2)$.
    If $b=1$, the challenger samples $a_1\gets\bit^\secp$ and $a_2\gets\bit^\secp$.
    The challenger sends $(k_{t^*_1,t^*_2},a_1,a_2)$ to $\qD$.
    \item $\qD$ generates $y^*_1\coloneqq (t^*_1,a_1\oplus x^*_1)$, $y^*_2\coloneqq (t^*_2,a_2\oplus x^*_2)$, and $\hat{C^*}\gets iO (1^\secp,C^*)$.
    It sets $\td^*\coloneqq(f,k_{t^*_1,t^*_2},\{t^*_1,t^*_2\},y^*_1,y^*_2)$.
    \item $\qD$ runs $(x'_1,x'_2)\gets\qA(\hat{C^*},\td^*,y^*_1,y^*_2)$ and outputs 1 if $x'_1=x^*_1$ or $x'_2=x^*_2$.
\end{enumerate}
Then
\begin{align}
    &\Pr[1\gets\mathsf{Hyb}_1] = \Pr[1\gets\qD \mid b=0] \\
    &\Pr[1\gets\mathsf{Hyb}_2] = \Pr[1\gets\qD \mid b=1].
\end{align}
By the pseudorandomness at punctured points of puncturable PRFs, we have
\begin{align}
    |\Pr[1\gets\mathsf{Hyb}_1] - \Pr[1\gets\mathsf{Hyb}_2] |\le\negl(\secp).
\end{align}
\end{proof}

\noindent
$\mathsf{Hyb}_3$: This is identical to $\mathsf{Hyb}_2$ except that $y^*_1$ and $y^*_2$ are replaced with $(t^*_1,r_1)$ and $(t^*_2,r_2)$ respectively, where $r_1$ and $r_2$ are sampled uniformly at random.
\begin{enumerate}
    \item The challenger generates $f\gets\OWF.\Gen(1^\secp)$, $k\gets\PPRF.\Gen(1^\secp)$.
    \item The challenger generates $r_1\gets\bit^\secp$, $r_2\gets\bit^\secp$, $x^*_1\gets\bit^\secp$, $t^*_1\coloneqq f(x^*_1)$, $y^*_1\coloneqq (t^*_1,r_1)$, $x^*_2\gets\bit^\secp$, $t^*_2\coloneqq f(x^*_2)$, $y^*_2\coloneqq (t^*_2,r_2)$, $k_{t^*_1,t^*_2}\gets\PPRF.\Puncture(k,(t^*_1,t^*_2))$ and $\hat{C^*}\gets iO(1^\secp,C^*)$.
    Let $\td^*\coloneqq(f,k_{t^*_1,t^*_2},\{t^*_1,t^*_2\},y^*_1,y^*_2)$.
    \item $(x'_1,x'_2)\gets\qA(\hat{C^*},\td^*,y^*_1,y^*_2)$.
    \item Output 1 if $x'_1=x^*_1$ or $x'_2=x^*_2$.
\end{enumerate}
\begin{claim}
    \begin{align}
        \Pr[1\gets\mathsf{Hyb}_2] = \Pr[1\gets\mathsf{Hyb}_3]
    \end{align}
\end{claim}
\begin{proof}
    This claim holds because the distributions of $y^*_1$ and $y^*_2$ are the same in $\mathsf{Hyb}_2$ and $\mathsf{Hyb}_3$. 
\end{proof}

\begin{claim}
For any QPT adversary $\qA$,
    \begin{align}
        \Pr[1\gets\mathsf{Hyb}_3] \le \negl(\secp).
    \end{align}
\end{claim}
\begin{proof}
We show this claim by using the security of the injective OWF $f$.
Let us consider the QPT adversary $\qE$ against $f$ that behaves as follows:
\begin{enumerate}
    \item $\qE$ takes $(f,t)$ as input, where $f\gets\OWF.\Gen(1^\secp)$, $x\gets\bit^\secp$ and $t\coloneqq f(x)$.
    \item $\qE$ samples $b\gets\bit$ and sets $t^*_b\coloneqq t$
    \item $\qE$ generates $k\gets\PPRF.\Gen(1^\secp)$, $x^*_{1-b}\gets\bit^\secp$, $t^*_{1-b}\coloneqq f(x^*_{1-b})$, $k_{t^*_0,t^*_1}\gets\PPRF.\Puncture(k,(t^*_0,t^*_1))$, $r_0\gets\bit^\secp$, $r_1\gets\bit^\secp$, $y^*_0\coloneqq (t^*_0,r_0)$, $y^*_1\coloneqq (t^*_1,r_1)$, and $\hat{C^*}\gets iO(1^\secp,C^*)$.
    It sets $\td^*\coloneqq(f,k_{t^*_0,t^*_1},\{t^*_0,t^*_1\},y^*_0,y^*_1)$.
    \item $\qE$ runs $(x'_0,x'_1)\gets\qA(\hat{C^*},\td^*,y^*_0,y^*_1)$ and outputs $x'_b$.
\end{enumerate}
Then
\begin{align}
    \Pr[1\gets\mathsf{Hyb}_3] 
    &= \Pr[x'_0=x^*_0 \lor x'_1=x^*_1 : (x'_0,x'_1)\gets\qA(\hat{C^*},\td^*,y^*_0,y^*_1)] \\
    &\le \sum_{b\in\bit} \Pr[x'_b=x^*_b : (x'_0,x'_1)\gets\qA(\hat{C^*},\td^*,y^*_0,y^*_1)] \\
    &= \Pr[x\gets\qE(f,t) \mid b=0] + \Pr[x\gets\qE(f,t) \mid b=1] \\ 
    &\le \negl(\secp).
\end{align}
In the last inequality, we used the one-wayness of the injective OWF $f$. 
\end{proof}

Therefore, we finally obtain $\Pr[1\gets\mathsf{Hyb}_0]\le\negl(\secp)$ for any QPT adversary $\qA$ and complete the proof.

\fi

\end{document}